\documentclass{article}
\usepackage{wright}
\newcommand{\dtr}{\mathrm{D}_{\tr}}
\newcommand{\fidelity}{\mathrm{F}}
\newcommand{\spec}{\mathrm{spec}}
\newcommand{\swap}{\mathrm{SWAP}}

\title{Beating full state tomography \\ for unentangled spectrum estimation}
\author{Angelos Pelecanos\thanks{UC Berkeley. \texttt{apelecan,ewin,jswright@berkeley.edu}} 
\and Xinyu Tan\thanks{MIT. \texttt{norahtan@mit.edu}} 
\and Ewin Tang\footnotemark[1] 
\and John Wright\footnotemark[1]}

\date{}

\begin{document}

\maketitle

\begin{abstract}
How many copies of a mixed state $\rho \in \mathbb{C}^{d \times d}$ are needed to learn its spectrum?
To date, the best known algorithms for spectrum estimation require as many copies as full state tomography, 
suggesting the possibility that learning a state's spectrum might be as difficult as learning the entire state. 
We show that this is not the case in the setting of unentangled measurements, by giving a spectrum estimation algorithm that uses $n = O(d^3\cdot (\log\log(d) / \log(d))^4 )$ copies of $\rho$,
which is asymptotically fewer than the $n = \Omega(d^3)$ copies necessary for full state tomography.
Our algorithm is inspired by the technique of local moment matching from classical statistics, and shows how it can be applied in the quantum setting.

As an important subroutine in our spectrum estimation algorithm, we give an estimator of the $k$-th moment $\tr(\rho^k)$ which performs unentangled measurements and uses $O(d^{3-2/k})$ copies of $\rho$ in order to achieve a constant multiplicative error. This directly translates to an additive-error estimator of quantum R\'{e}nyi entropy of order $k$ with the same number of copies.

Finally, we present numerical evidence that the sample complexity of spectrum estimation can only improve over full state tomography by a sub-polynomial factor.
Specifically, for spectrum learning with fully entangled measurements, we run simulations which suggest a lower bound of $\Omega(d^{2 - \gamma})$ copies for any constant $\gamma > 0$.
From this, we conclude the current best lower bound of $\Omega(d)$ is likely not tight.
\end{abstract}

\newpage
\hypersetup{linktocpage}
\tableofcontents

\newpage
\section{Introduction}

We study the fundamental learning theoretic task of estimating a mixed state $\rho$'s spectrum given access to identical copies of $\rho$. If $\rho$ is $d$-dimensional, its spectrum can be written as $\alpha = (\alpha_1, \ldots, \alpha_d)$, where $\alpha_1 \geq \cdots \geq\alpha_d$.
In this case, the goal is to output an estimator $\widehat{\balpha}$ which is $\eps$-close in total variation distance to $\alpha$, $\dtv{\alpha}{\widehat{\balpha}}\leq \epsilon$,
with probability~99\%.

The spectrum captures all unitarily invariant statistics of a mixed state $\rho$, and so many important properties can be derived from it.
For example, $\alpha$ encodes the purity of a state: $\rho$ is a pure state if $\alpha = (1, 0, \ldots, 0)$, and $\rho$ is the maximally mixed state if $\alpha = (\frac{1}{d}, \ldots, \frac{1}{d})$.
Likewise, if $\rho = \psi_A$ is the reduced density matrix of a bipartite pure state $\ket{\psi_{AB}}$,
then its spectrum $\alpha$ coincides with the Schmidt coefficients of $\ket{\psi_{AB}}$,
and so $\alpha$ encodes many interesting properties of $\ket{\psi_{AB}}$'s entanglement.
For example, $\ket{\psi_{AB}}$ is unentangled if $\alpha = (1, 0, \ldots, 0)$ and entangled otherwise.
If it is entangled, then the amount of entanglement can be quantified by the \emph{entanglement entropy} of $\ket{\psi_{AB}}$, which is equal to the \emph{von Neumann entropy} of $\rho$, in turn equal to the \emph{Shannon entropy} of $\alpha$, $\mathrm{H}(\alpha) = \sum_{i=1}^d \alpha_i \cdot \log_2(1/\alpha_i)$.
The importance of the spectrum
has led to a variety of theoretical works giving algorithms for estimating the spectrum of~$\rho$~\cite{ARS88,KW01,HM02,CM06,OW15,OW16,OW17a}
and for estimating  Shannon and R\'{e}nyi entropies of $\alpha$~\cite{AISW20,BMW17,OW15,BOW19}.

One final application of spectrum estimation is as a necessary ingredient in any quantum state tomography algorithm.
Quantum state tomography entails computing an estimate $\widehat{\brho}$ of the state $\rho$ such that $\dtr(\rho, \widehat{\brho})\leq \epsilon$,
and this requires estimating both $\rho$'s eigenvalues \emph{and} its eigenvectors.
Indeed, two of the sample optimal entangled tomography algorithms~\cite{OW16,HHJ+16} begin by first running a well-studied spectrum estimation algorithm known as the \emph{Empirical Young Diagram (EYD) algorithm}~\cite{ARS88,KW01} (also known as the \emph{Keyl--Werner algorithm}) in order to estimate $\rho$'s spectrum.
However, even tomography algorithms without an explicit spectrum estimation subroutine must still be implicitly learning the spectrum.
This is because if $\widehat{\balpha}$ is the spectrum of $\widehat{\brho}$, then $\dtv{\alpha}{\widehat{\balpha}} \leq \dtr(\rho, \widehat{\brho}) \leq \epsilon$
(see~\cite[Proposition 2.2]{OW15} for a proof of this fact).

Spectrum estimation is therefore always possible with a number of samples equal to the number of samples needed for full state tomography.
But does spectrum estimation \emph{require} the same number of copies as full state tomography, or can it can be solved with asymptotically fewer copies?
To our knowledge, this question was first posed in \cite[Section 10.2]{Wri16}, and it remains open for both entangled and unentangled measurements.
In the case of entangled measurements,
all that is known is that spectrum estimation can be solved in $n = O(d^2/\epsilon^2)$ copies and requires $n = \Omega(d/\epsilon^2)$ copies.
This is because full state tomography can be solved in $n = \Theta(d^2/\epsilon^2)$ copies~\cite{OW16,HHJ+16} (though the EYD spectrum estimation algorithm can be analyzed independently of full state tomography, but it too requires $n = \Theta(d^2/\epsilon^2)$ copies~\cite{OW15,OW16}),
and testing if $\alpha = (\frac{1}{d}, \ldots, \frac{1}{d})$, i.e.\ if $\rho$ is the maximally mixed state, is known to require $n = \Omega(d/\epsilon^2)$ copies~\cite{OW15}.
In the case of unentangled measurements, all that is known is that
spectrum estimation can be solved in $n = O(d^3/\epsilon^2)$ copies via full state tomography~\cite{KRT14}, and that $n = \Omega(d^{1.5}/\epsilon^2)$ copies are needed even to test if $\rho$ is maximally mixed~\cite{CHLL22}.
This gives a quadratic gap between upper and lower bounds for spectrum estimation in both entangled and unentangled cases.
In our experience, experts seem divided about whether spectrum estimation should require the same number of copies as full state tomography,
whether it can be solved with quadratically fewer copies,
or whether the truth lies somewhere in between.

The main result of this work is the following.

\begin{theorem}[Main result]\label{thm:main}
    There is an algorithm which solves spectrum estimation in
    \begin{equation*}
        n = O\Big(d^3 \cdot \Big(\frac{\log\log(d)}{\log(d)}\Big)^4 \cdot \frac{1}{ \epsilon^6}\Big)
    \end{equation*}
    copies using unentangled measurements.
\end{theorem}
\noindent
As full state tomography requires $n = \Omega(d^3/\epsilon^2)$ copies for unentangled measurements~\cite{CHL+23},
this shows that spectrum estimation can be performed with asymptotically fewer copies than full state tomography, at least in the ``large $\epsilon$'' regime.
In particular, our algorithm improves on full state tomography in the regime where $\epsilon = \omega(\log\log(d)/\log(d))$.
This is often the relevant regime: typically, we imagine tomography on a system of $q$ qubits, so that $d = 2^q$; then, this regime translates to $\epsilon = \omega(\log(q)/q)$, inverse polynomial in the number of qubits.
We have not tried to optimize our algorithm's dependence on $\epsilon$, and we believe that further improvements should be possible, which we leave to future work.

\paragraph{Lower bounds.}
Although we have not been able to show matching lower bounds,
we provide numerical evidence that the sample complexity of spectrum estimation can only improve over full state tomography by a sub-polynomial factor, at least in the case of entangled measurements.
In particular, for each $k=2,3,4$, we construct two possible mixed state spectra $\alpha^{(k)}$ and $\beta^{(k)}$ which (1) are far from each other, i.e.\ $\dtv{\alpha^{(k)}}{\beta^{(k)}} = \Omega(1)$, and (2) match on their first $k-1$ moments and differ at the $k$-th moment. 
We believe that $n = \Theta(d^{2-2/k})$ copies are necessary and sufficient to distinguish whether a mixed state $\rho$ has spectrum~$\alpha^{(k)}$ or~$\beta^{(k)}$.
We provide numerical evidence that suggests that this is indeed the case.
Of course, this task is solvable with a spectrum estimation algorithm, since one can always learn the spectrum of $\rho$ and check whether it is close to $\alpha$ or $\beta$,
and so our evidence suggests a lower bound of $n = \Omega(d^{2-\gamma})$ copies for any constant $\gamma > 0$ holds for spectrum estimation as well.
This would improve on the existing lower bound of $n = \Omega(d)$ for entangled measurements in the $\epsilon = \Omega(1)$ regime due to~\cite{CHTW04,OW15}.

\paragraph{Classical analogues.}
The quantum problems we study have natural classical analogues in the field of statistics.
In particular, suppose one is given $n$ samples $\bx = (\bx_1, \ldots, \bx_n)$
from a (not necessarily sorted) probability distribution $\alpha = (\alpha_1, \ldots, \alpha_d)$
over $d$ items.
The natural classical analogue of full state tomography
is the problem of learning the distribution $\alpha$, which entails computing an estimate $\widehat{\balpha}$ of the distribution $\alpha$ such that $\dtv{\alpha}{\widehat{\balpha}} \leq\epsilon$.
It is well known that this can be solved using only $n = O(d/\epsilon^2)$ samples~\cite{Can20}, and furthermore that this bound is optimal.
The natural classical analogue of spectrum estimation, on the other hand,
is the problem of learning the \emph{sorted distribution}
$\alpha^{\geq} \coloneqq \mathrm{sort}(\alpha)$, where $\mathrm{sort}(\cdot)$ is the function which sorts its input from largest to smallest.
This entails outputting an estimator of the sorted distribution $\widehat{\balpha}^{\geq}$ such that $\dtv{\alpha^{\geq}}{\widehat{\balpha}^{\geq}} \leq \epsilon$.

Estimating the sorted distribution can be solved in $n = O(d/\epsilon^2)$ samples by first computing an estimate $\widehat{\balpha}$ of $\alpha$ and then sorting it,
but it was shown in a work of Valiant and Valiant~\cite{VV11a,VV17} that one can improve on this naive algorithm and estimate the sorted distribution using only $O(d/\log(d))$ samples when $\epsilon$ is constant.  
In follow-up work,
Han, Jiao, and Weissman~\cite{HJW18} gave an essentially optimal algorithm for estimating the sorted distribution in terms of both the dimension $d$ and error $\epsilon$ parameters.
For any parameter $\gamma >0$,
they give an algorithm with sample complexity $n = O(d/(\log(d) \cdot \epsilon^2))$ so long as $\epsilon \geq 1/d^{1-\gamma}$, and when $\epsilon \leq 1/d$
the above bound of $n = O(d/\epsilon^2)$ samples can be applied;
moreover, they show that these bounds are in fact optimal in these two regimes.
Their estimator is based on a technique they introduced called \emph{local moment matching}.
Our algorithm for spectrum estimation is inspired by their approach.

\paragraph{Moment estimation and R\'{e}nyi entropy estimation.}
A key subroutine of our algorithm is estimating $\tr(\sigma^k)$ from copies of the state $\sigma$, a task known as \emph{moment estimation}.
Given an estimator $\bZ_k$ for $\tr(\sigma^k)$, two types of guarantees one might hope for are \emph{additive error} and \emph{multiplicative error} guarantees, defined as follows.
\begin{align*}
\text{(Additive error):}& \quad \tr(\sigma^k) - \delta \leq \bZ_k \leq \tr(\sigma^k) + \delta,\\
\text{(Multiplicative error):}& \quad     (1- \delta) \cdot \tr(\sigma^k) \leq \bZ_k \leq (1 + \delta) \cdot \tr(\sigma^k).
\end{align*}
Multiplicative error guarantees are stronger than additive error guarantees because the magnitude of the error scales with $\tr(\sigma^k)$,
whereas with an additive error guarantee, the error $\delta$ might completely swamp the potentially much smaller $\tr(\sigma^k)$ term.
It is well-known that a multiplicative error approximation for the $k$-th moment can be converted to an additive error approximation for the \emph{quantum R\'{e}nyi entropy of order $k$},
and vice versa,
where the quantum R\'{e}nyi entropy is defined as
\begin{equation*}
        S_{k}(\sigma) = \frac{1}{1-k} \log \tr(\sigma^k).
\end{equation*}
Equivalently, $S_k(\sigma)$ is just the \emph{classical} R\'{e}nyi entropy of $\sigma$'s spectrum.
If $\sigma = \psi_A$, where $\ket{\psi_{AB}}$ is a bipartite pure state, then the quantities $S_{k}(\sigma)$, for $k \geq 2$, are referred to as the \emph{R\'{e}nyi entanglement entropies}, and they give a rich description of the entanglement properties of $\ket{\psi_{AB}}$.
Indeed, these R\'{e}nyi entanglement entropies have been estimated on bipartite quantum states in experimental settings dating back to the works ~\cite{IMP+15,KTL+16}; as the first of these works puts it, ``[t]he R\'{e}nyi entropies are rapidly gaining importance in theoretical condensed matter physics, as they can be used to extract information about the `entanglement
spectrum'.''

We give an algorithm for moment and R\'{e}nyi entropy estimation with the following guarantees.
\begin{theorem}[Quantum R\'{e}nyi entropy estimation]\label{thm:quantum-renyi}
    For any integer $k\geq 2$ and $d$-dimensional quantum state $\sigma$, 
    there is an algorithm which, with probability $99\%$, estimates $\tr(\sigma^k)$ to $\delta$ multiplicative error and $S_k(\sigma)$ to $\delta$ additive error using
    \begin{equation*}
        n = O\left(\max \Big\{\frac{d^{2-2/k}}{\delta^2}, \frac{d^{3 - 2/k}}{\delta^{2/k}} \Big\}\right)
    \end{equation*}
    copies of $\sigma$ and unentangled measurements only.
\end{theorem}
\noindent
For constant $\delta$, the second term dominates, and the number of copies scales as $O(d^{3-2/k})$, and so for large $k$ the sample complexity mirrors the $O(d^3)$ which appears in full state tomography. The cross-over point between the two terms happens when $\delta = 1/d^{k/(2k-2)}$, and for $\delta$ smaller than this the first term dominates.

Quantum R\'{e}nyi entropy estimation was previously studied in the work of Acharya, Issa, Shende, and Wagner~\cite{AISW20}.
They give an algorithm which uses entangled measurements and achieves a sample complexity of
\begin{equation*}
    n = \Theta\left(\max \Big\{\frac{d^{1-1/k}}{\delta^2}, \frac{d^{2 - 2/k}}{\delta^{2/k}} \Big\}\right),
\end{equation*}
which they show is optimal by proving matching lower bounds.
As in the case of our bounds, these two bounds trade off when $\delta = 1/d^{k/(2k-2)}$, and for constant $\delta$ and large $k$ the sample complexity mirrors the $O(d^2)$ needed for full state tomography with entangled measurements.
For the related classical problem of estimating  R\'{e}nyi entropies of discrete distributions, it is known that $\Theta(d^{1-1/k}/\delta^2)$ samples are necessary and sufficient~\cite{AOST17}.

We leave open the question of whether \Cref{thm:quantum-renyi} is tight or can be improved. Let us note that the hard examples which give the tight lower bounds in both \cite{AISW20,AOST17} are quite simple and only involve distributions which have one heavy element and are uniform on the remaining elements.

\paragraph{Organization of the paper.}

In \Cref{sec:learning-sorted}, we survey the problem of learning the sorted distribution in the classical setting. 
We review the method of local moment matching by including a self-contained simple analysis using only two ``buckets'' that gives an optimal sample complexity in terms of the $d$ dependence, albeit a suboptimal sample complexity in terms of the $\epsilon$ dependence. 
In \Cref{sec:technical_overview_quantum}, we give a technical overview of our approach for spectrum and moment estimation in the quantum setting. 
We include some preliminaries in~\cref{sec:prelim}. 
Next, our spectrum learning algorithm consists of three main components, and we devote one section to each.
\begin{itemize}
    \item[$\circ$] Beginning in~\cref{sec:moment_estimation}, we construct a simple unbiased estimator for the $k$-th moment of any quantum state, bound its variance in~\cref{thm:generalmomentvar}, and prove the sample complexity for multiplicative-error moment estimation and additive-error quantum R\'{e}nyi entropy estimation. 
    We then generalize the variance bound to a subnormalized state projecting onto the small bucket in \Cref{sec:subnormalized-moment-estimation}. 

    \item[$\circ$] Then in~\cref{sec:bucket}, we give a bucketing algorithm, which splits the spectrum of $\rho$ into a large bucket and a small bucket. We analyze the performance of the bucketing algorithm in~\cref{thm:bucket-algo}. 

    \item[$\circ$] Finally in~\cref{sec:local_moment_matching}, we study the framework of using moment estimates within a local interval to estimate a sorted probability distribution. In particular, we focus on the moment matching in the smallest bucket and analyze its performance in~\cref{thm:smallestbucketLMM}. 
\end{itemize}
We put these three components together to give our main spectrum learning algorithm in~\cref{sec:full_algo}, and then we analyze its sample complexity and prove our main result. 
With this done, in~\cref{sec:lmm-fails}
we analyze general bucketing algorithms and show that their sample complexity is related to the sample complexity needed to perform full state tomography in fidelity. 
Finally, in~\cref{sec:lower}, we study the spectrum learning in the setting of entangled measurements. We give numerical evidence that the existing lower bound $\Omega(d)$ based on uniformity testing is not tight.

\section{Learning the sorted distribution}\label{sec:learning-sorted}

One of the most fundamental tasks in classical statistics is that of estimating \emph{symmetric properties} of $\alpha$,
i.e.\ those properties which remain invariant under permutations of $\alpha$'s $d$ coordinates.
Two important examples are the \emph{support size} of $\alpha$, given by the number of nonzero coordinates in $\alpha$, and the \emph{Shannon entropy} $\mathrm{H}(\alpha)$.
These properties depend on the multiset of probability values $\{\alpha_1, \ldots, \alpha_d\}$ but not on their order, and so they are symmetric.

The most straightforward estimators for both support size and entropy give good approximations using a linear $n = O(d)$ number of samples,
but recent years have seen the development of more sophisticated estimators for both of these properties which only need a \emph{sublinear} $n = o(d)$ number of samples.
For entropy, this began with the work of Paninski~\cite{Pan04}, who gave the first proof of the existence of an estimator which uses an unspecified sublinear number of samples; surprisingly, this proof is nonconstructive!
Following this, the breakthrough work of Valiant and Valiant~\cite{VV11a} gave an explicit estimator for entropy achieving sample complexity $O(d/\log(d))$.
Subsequent works~\cite{VV13,VV17,VV11b} gave improved sample complexities which captured the dependence on the error parameter $\epsilon$ in addition to the dimension $d$,
culminating in the works of Wu and Yang~\cite{WY16} and Jiao et al.~\cite{JVHW15} which achieved an optimal sample complexity of $n = \Theta(d/(\epsilon \log(d)) + \log^2(d)/\epsilon^2)$.
The story for support size is similar: an estimator with sublinear sample complexity was first demonstrated in the work of Valiant and Valiant~\cite{VV11a},
and following the improvements in~\cite{VV17},
Wu and Yang~\cite{WY15} showed that the optimal sample complexity was $n = \Theta(d / \log(d) \cdot \log^2(1/\epsilon))$.

The estimators for entropy and support size, as well as those for other symmetric properties such as power sum polynomials and R\'{e}nyi entropies~\cite{JVHW15,AOST17}, are often bespoke and tailored to the particular symmetric property that is being estimated.
The work of Valiant and Valiant~\cite{VV11a}, however, gave a unified approach to estimating symmetric properties, via an estimator $\widehat{\balpha}^{\geq}$ for the sorted distribution $\alpha^{\geq}$.
This then yields an estimator for general symmetric properties by ``plugging it in''; for example, they show that $\mathrm{H}(\widehat{\balpha}^{\geq})$ is a good estimator for $\mathrm{H}(\alpha)$,
and that the support size of $\widehat{\balpha}^{\geq}$ is a good estimator for the support size of $\alpha$.
This is the approach also taken by Han, Jiao, and Weissman~\cite{HJW18}, and they show that ``plugging in'' their estimator for the sorted distribution gives an optimal estimator for both entropy and support size in certain regimes of parameters.
Below, we survey several approaches for estimating $\alpha^{\geq}$
in order to motivate their approach of local moment matching.

\subsection{The empirical distribution}
How to estimate the sorted distribution?
The most natural approach is to compute the \emph{empirical sorted distribution}, given by the following algorithm.
\begin{enumerate}
    \item Compute the \emph{histogram} $\bh = (\bh_1, \ldots, \bh_d)$ of $\bx$, where $\bh_i$ is the number of $i$'s which appear in $\bx$.
    \item Compute the \emph{empirical distribution} $\widehat{\balpha} = \frac{1}{n} \cdot \bh = (\frac{1}{n} \cdot \bh_1, \ldots, \frac{1}{n} \cdot \bh_d)$.
    \item Output the sorted empirical distribution $\widehat{\balpha}^{\geq} = \mathrm{sort}(\widehat{\balpha})$.
\end{enumerate}
It is a classic fact in statistics
that the empirical distribution satisfies $\dtv{\alpha}{\widehat{\balpha}} \leq \epsilon$ with high probability once $n = O(d/\epsilon^2)$~\cite{Can20},
and so $\dtv{\alpha^{\geq}}{\widehat{\balpha}^{\geq}} \leq \epsilon$ with high probability when $n = O(d/\epsilon^2)$ as well.
This analysis is also tight,
as this estimator requires $n = \Omega(d/\epsilon^2)$ samples in the special case when $\alpha$ is the uniform distribution.
(For the full $\Omega(d/\epsilon^2)$ lower bound, see~\cite{DDIS12}. 
For the simpler $\Omega(d)$ lower bound, note that with $n = o(d)$ samples the sorted empirical distribution $\widehat{\balpha}^{\geq}$ will have support size at most $o(d)$, and so $\dtv{\mathrm{unif}_d}{\widehat{\balpha}^{\geq}}$ will tend to 1.)
However, as we have seen, this is a sub-optimal sample complexity for learning the sorted distribution,
and the reason for this is that the empirical distribution is also wastefully learning the labels of the probability values.

\subsection{Profile maximum likelihood}
A second natural approach is known as the \emph{profile maximum likelihood (PML) estimator}.
Letting $\bh^{\geq} = \mathrm{sort}(\bh)$ be the sorted histogram, the PML estimator computes the sorted distribution $\widehat{\balpha}^{\geq}$ which has the largest probability 
of producing a sample $\bx$ with sorted histogram $\bh^{\geq}$.
It was shown by Acharya et al.~\cite{ADOS17} (with improvements due to \cite{HO19}) that the PML estimator does indeed yield an optimal sample complexity for estimating $\alpha^{\geq}$, and ``plugging it in'' yields optimal sample complexities for symmetric properties such as the entropy.
However, using the PML comes with two main challenges.
The first is that computing the PML estimator might be computationally intractable,
as it requires a maximization over all sorted probability distributions.
This issue was resolved in the works \cite{CSS19,ACSS21,ACSS20,CJSS22}, which give sophisticated polynomial-time algorithms for computing approximations to the PML estimator which are sufficiently good to estimate a variety of symmetric properties.
The second, which is more important to us,
is that it is difficult to directly analyze the sample complexity of the PML estimator.
Instead, the analysis of Acharya et al.\
is only able to show that if there exists an estimator for the sorted distribution (or for the entropy, etc.)
with a given sample complexity, then the PML estimator has essentially the same sample complexity.
So to actually prove that the PML estimator has a given sample complexity,
one has to first demonstrate that another estimator already possesses this sample complexity.
This means that the PML estimator is not particularly useful for our purpose, which is establishing the sample complexity of learning the spectrum.

\subsection{Learning moments}
A third, and final, natural approach to estimating the sorted distribution is to learn its \emph{moments}, given by the quantities
\begin{equation*}
    p_k(\alpha) = \sum_{i=1}^d \alpha_i^k.
\end{equation*}
These $p_k$'s are symmetric polynomials known as the \emph{power sum polynomials}. 
They contain information only about the multiset of $\alpha_i$'s and not about their labels.
Indeed, the first $d$ moments $p_1(\alpha)$ through $p_d(\alpha)$ are enough to uniquely specify the distribution $\alpha$.
To see this, Newton's identities imply that the first $d$ power sum polynomials
uniquely specify the first $d$ \emph{elementary symmetric polynomials} $e_1(\alpha), \ldots, e_d(\alpha)$, where
\begin{equation*}
    e_k(\alpha) = \sum_{1 \leq i_1 < \cdots < i_k \leq d} \alpha_{i_1} \cdots \alpha_{i_k}.
\end{equation*}
Next, if we write $r_\alpha(x)$ for the degree-$d$ polynomial whose roots are $\alpha_1, \ldots, \alpha_d$, we have that
\begin{equation*}
    r_{\alpha}(x) = (x - \alpha_1) \cdots (x - \alpha_d) = \sum_{k=0}^d x^{d- k} \cdot (-1)^k \cdot e_k(\alpha).
\end{equation*}
Hence, the first $d$ elementary symmetric polynomials uniquely specify $r_{\alpha}(x)$, and from $r_{\alpha}(x)$ we can learn the multiset $\{\alpha_1, \ldots, \alpha_d\}$ by inspecting its roots.

In practice, we will not have access to the moments $p_k(\alpha)$.
Instead, we will have to estimate them from samples $\bx = (\bx_1, \ldots, \bx_n)$.
If $z_k(\bx)$ is an estimator for $p_k(\alpha)$, what properties might we want it to satisfy?
Perhaps the simplest property is that of being an \emph{unbiased estimator}, which means that it equals $p_k(\alpha)$ in expectation, i.e.\ $\E_{\bx}[z_k(\bx)] = p_k(\alpha)$.
For example, a natural unbiased estimator for $p_k(\alpha)$ is
\begin{equation*}
    z_k(\bx) = \bone[\bx_1 = \cdots = \bx_k],
\end{equation*}
which checks if there is a \emph{$k$-wise collision} among the first $k$ samples.
This is indeed an unbiased estimator,
because it is equal to 1 with probability $p_k(\alpha)$ and 0 otherwise, but it can be very far from its mean of $p_k(\alpha)$ for any fixed sample $\bx$ because it only outputs values in $\{0, 1\}$.
This issue is reflected in the fact that $z_k(\bx)$ has a large variance, and suggests that a second property we want for our estimator is for it to have as small a variance as possible.
Fortunately, there is a standard method called \emph{U-statistics} (``U'' for ``unbiased'') for reducing the variance of unbiased estimators such as $z_k(\bx)$, which involves averaging the estimators over all permutations of the sample $\bx$.
In our case, the U-statistic corresponding to $z_k(\bx)$ is the estimator
\begin{equation}\label{eq:collision-estimator}
    c_k(\bx) = \frac{1}{\binom{n}{k}} \cdot \sum_{1 \leq i_1 < \cdots < i_k \leq n}  \bone[\bx_{i_1} = \cdots = \bx_{i_k}].
\end{equation}
Each term in the sum has the same expectation of $p_k(\alpha)$, and so it is an unbiased estimator by linearity of expectation;
however, its variance is greatly reduced, and indeed it turns out to be the minimum variance unbiased estimator for $c_k(\bx)$.
Moreover, it has a natural interpretation in terms of the \emph{collision statistics} of $\bx$, in that it counts the number of $k$-wise collisions in $\bx$ and then normalizes.
As a function of the histogram $\bh$, we can write it as
\begin{equation*}
    c_k(\bx) = \frac{1}{\binom{n}{k}}\cdot \sum_{i=1}^d \binom{\bh_i}{k}.
\end{equation*}
Thus, a natural algorithm is to take the sample, compute $c_1(\bx), \ldots, c_d(\bx)$, and use these to somehow compute an estimator $\widehat{\balpha}^{\geq}$ of $\alpha^{\geq}$;
typically the way that one does this is to find a distribution $\widehat{\balpha}^{\geq}$ whose moments approximately match the estimated values $c_1(\bx), \ldots, c_d(\bx)$,
an approach known as \emph{moment matching}.

Unfortunately, this approach does not work well in practice.
The reason is that the moments $p_k(\alpha)$ are dominated by the high-probability elements in the sample, and these larger elements tend to ``wash out'' the contribution from the low-probability elements.
But capturing low-probability elements, even those on the order of $1/d$, is still important for estimating in total variation distance,
as a probability distribution might be largely or even entirely supported on elements of this size.
We illustrate this problem with the following example of uniformity testing.

\begin{example}[Uniformity testing]\label{ex:uniformity}
    Consider the problem of distinguishing the case when (i) $\alpha$ is uniform over all of $[d]$ from the case when (ii) $\alpha$ is uniform over some subset $S \subseteq [d]$ of size $d/2$ (where $S$ is unknown).
    It is well-known that $n = \Theta(\sqrt{d})$ samples are necessary and sufficient to solve this task~\cite{GR11,BFF+01}.
    One method for doing so is to note that in case (i), $p_2(\alpha) = 1/d$, whereas in case (ii) $p_2(\alpha) = 2/d$. So a natural algorithm is to draw $n$ samples $\bx = (\bx_1, \ldots, \bx_n)$, compute $c_2(\bx)$, and output ``uniform'' if and only if $c_2(\bx) \leq 1.5/d$.

    To analyze this approach, note that since $c_2(\bx)$ is an unbiased estimator for $p_2(\alpha)$, it suffices to show that $c_2(\bx)$ is close to its mean with high probability. 
    In particular, we want that it deviates from its mean by less than $0.5/d$ with high probability when $n = O(\sqrt{d})$.
    To show this, a routine calculation gives the following for the variance of $c_2(\bx)$~(cf.\ the proof of \cite[Lemma 3]{DGPP19}):
    \begin{equation}\label{eq:collision-variance}
        \Var[c_2(\bx)] = \frac{1}{\binom{n}{2}} (p_2(\alpha) - p_2(\alpha)^2) + \frac{2(n-2)}{\binom{n}{2}}(p_3(\alpha) - p_2(\alpha)^2).
    \end{equation}
    In both case (i) and case (ii) this variance is $O(1/(dn^2))$,
    and so the standard deviation of $c_2(\bx)$ in both cases is $O(1/(\sqrt{d} n))$. Thus, we can make this significantly smaller than $0.5/d$ by taking $n = O(\sqrt{d})$, as desired.

    Now let us see how things change if we throw in a single element with large probability.
    Suppose we are given samples from a $(d+1)$-dimensional distribution $\alpha = (\alpha_1, \ldots, \alpha_d, \alpha_{d+1})$ and asked to distinguish between the following two cases:
    \begin{equation}\label{eq:harder-cases}
        \text{Case (i):}~\mathrm{sort}(\alpha) = \Big(\frac{1}{2}, \frac{1}{2d}, \ldots, \frac{1}{2d}\Big),
        \qquad
        \text{Case (ii):}~\mathrm{sort}(\alpha) = \Big(\frac{1}{2}, \frac{1}{d}, \ldots, \frac{1}{d}, 0, \ldots, 0\Big).
    \end{equation}
    The second moment $p_2(\alpha)$ is $1/4 + 1/(4d)$ and $1/4 + 1/(2d)$ in cases (i) and (ii), respectively, so we would like to estimate it to accuracy better than $\pm 1/(8d)$. 
    If we compute the variance in \Cref{eq:collision-variance},
    however, we see that in both cases $p_3(\alpha) - p_2(\alpha)^2$ is now $\Omega(1)$.
    So the variances in both cases are $\Omega(1/n)$,
    their standard deviations are $\Omega(1/\sqrt{n})$,
    and to make this smaller than $\pm 1/(8d)$,
    we require $n = \Omega(d^2)$, a power of 4 worse
    than if we had no large probability element.
    In this example, then, we see that although we want to estimate the distribution's low probability elements, their contribution to $c_2(\bx)$ is washed out by the existence of the single large probability element.
\end{example}

\subsection{Local moment matching}\label{sec:lmm}
There is, however, a simple algorithm for distinguishing between the two cases in \Cref{eq:harder-cases} using only $O(\sqrt{d})$ samples:
simply spend $O(1)$ samples to learn the index $i$ for which $\alpha_i = 1/2$, and then use $O(\sqrt{d})$ samples to test if $\alpha$ is uniform on $[d]\setminus \{i\}$.
This hints at a more general approach for salvaging moment matching.
\begin{enumerate}
    \item (Sample splitting): Draw $2n$ samples. Call the first $n$ samples $\bx = (\bx_1, \ldots, \bx_n)$ and call the second $n$ samples $\by = (\by_1, \ldots, \by_n)$.
    \item (Bucketing): Let $\widehat{\bp}$ be the empirical distribution of $\bx$. Pick a threshold $0 \leq \tau \leq 1$ and set
    \begin{equation*}
    \mathsf{Large} = \{i \mid \widehat{p}_i \geq \tau\}\quad\text{and}\quad\mathsf{Small} = [d] \setminus \mathsf{Large}.
    \end{equation*}
    \item (Estimation): Write $\by|_{\mathsf{Large}}$ and $\by|_{\mathsf{Small}}$ for the samples in $\by$ which fall in the $\mathsf{Large}$ and $\mathsf{Small}$ sets, respectively.
    Note that $\by|_{\mathsf{Large}}$ are samples drawn from the distribution $\alpha|_{\mathsf{Large}}$,
    and similarly $\by|_{\mathsf{Small}}$ are samples drawn from $\alpha|_{\mathsf{Small}}$.
    Use these samples separately to estimate $\alpha|_{\mathsf{Large}}$ and $\alpha|_{\mathsf{Small}}$.
\end{enumerate}

Typically, the threshold $\tau$ is chosen to be small enough so that moment matching on just the $\by|_{\mathsf{Small}}$ samples is sufficient to learn $\alpha|_{\mathsf{Small}}$.
However, this means that the $\mathsf{Large}$ bucket will still contain a wide range of probability values, potentially ranging from the extremely small ($\tau$) to the extremely large ($1$), and so moment matching will still not be effective within this bucket.
This two-bucketing approach, then, is most useful for symmetric properties in which the chief difficulty is estimating the contribution to them from the small elements.
For example, the sample-optimal algorithms for estimating the Shannon entropy of $\alpha$~\cite{WY16,JVHW15} work in this manner: for the small bucket they use moment matching, and for the large bucket they take the empirical distribution $\widehat{\balpha}|_{\mathsf{Large}}$ of the sample $\by|_{\mathsf{Large}}$ and use a simple variant of the plug-in estimator $\mathrm{H}(\widehat{\balpha}|_{\mathsf{Large}})$ known as the bias-corrected plug-in estimator.

More general properties, however, require a more fine-grained approach to the large elements, which entails further splitting the $\mathsf{Large}$ bucket into more buckets, each of which contains a small enough range of probability values that moment matching within the bucket becomes effective.
Since the moment matching is now being done locally within each bucket, this approach is known as \emph{local moment matching},
and this is the approach shown to give a sample-optimal estimator for the sorted spectrum $\alpha^{\geq}$ by~\cite{HJW18}.
Below, we outline the bucketing and estimation steps of local moment matching.

\subsubsection{Bucketing}
Let $\bh = (\bh_1, \ldots, \bh_d)$ be the histogram of $\bx$.
Then $\bh_i$ is distributed as $\mathrm{Binomial}(n, \alpha_i)$ and therefore has mean $\alpha_i n$ and variance $\alpha_i (1-\alpha_i) n = O(\alpha_i n)$.
This means that if $\widehat{\balpha} = \frac{1}{n} \cdot \bh$ is the empirical distribution of $\bx$,
then $\widehat{\balpha}_i$ has mean $\alpha_i$ and variance $O(\alpha_i/n)$. Hence, we have that with probability $0.99$,
\begin{equation*}
    \widehat{\balpha}_i = \alpha_i  \pm O(\sqrt{\alpha_i/n}).
\end{equation*}
More generally, a Chernoff bound tells us that it deviates from its mean by at most $t \cdot O(\sqrt{\alpha_i/n})$ except with probability $\exp(-O(t^2))$.
So if we set $t = \sqrt{\log(d)}$, we get that
\begin{equation}\label{eq:interval}
    \widehat{\balpha}_i = \alpha_i  \pm \sqrt{\log(d)} \cdot O(\sqrt{\alpha_i/n})
\end{equation}
except with probability $0.01/d$.
Since there are only $d$ indices $i$, this is small enough that we can union bound over all the $i$'s and say that each $\widehat{\balpha}_i$ falls inside the interval from \Cref{eq:interval} except with probability $0.01$.
Rewriting \Cref{eq:interval}, we see that
\begin{equation*}
\widehat{\balpha}_i = \alpha_i \cdot (1 \pm O(\sqrt{\log(d)/(\alpha_i n)}).
\end{equation*}
Hence, $\widehat{\balpha}_i$ gives a multiplicative approximation to $\alpha_i$ once $\alpha_i \geq \log(d)/n$,
and as $\alpha_i$ increases beyond this threshold, the quality of the approximation increases with it.

Based on this, \cite{HJW18} define the $M = \sqrt{n/\log(d)}$ intervals $I_1, \ldots, I_M$ via
\begin{equation*}
    I_j = \Big[(j-1)^2 \cdot \frac{\log d}{n},~j^2 \cdot \frac{\log d}{n}\Big],
\end{equation*}
and we correspondingly bucket our indices into $M$ buckets $[d] = B_1 \cup \cdots \cup B_M$ by setting
\begin{equation*}
    B_j = \{i \mid \widehat{\balpha}_i \in I_j\}.
\end{equation*}
For intuition behind the definition of these intervals, note that the midpoint of the $j$-th interval is $m_j \coloneqq j (j-1) \cdot \log(d)/n$ and the radius of the interval around its midpoint is, essentially, $j \cdot \log(d)/n$.  Suppose that we had a probability value which matched the midpoint, i.e.\ $\alpha_i = j (j-1) \cdot \log(d)/n$. Then applying \Cref{eq:interval} (and dropping the Big-Oh for simplicity),
\begin{equation*}
\widehat{\balpha}_i
= \alpha_i \pm \sqrt{\log(d)} \cdot \sqrt{\alpha_i/n}
= \alpha_i \pm \sqrt{\log(d)} \cdot \sqrt{(j (j-1) \cdot \log(d)/n)/n}
= \alpha_i \pm j \log(d)/n,
\end{equation*}
so the error we estimate $\widehat{\balpha}_i$ to is precisely the width of its bucket. 
In general, then, each $\widehat{\balpha}_i$ will be placed in the correct bucket or a closely neighboring bucket.
As for the smallest bucket, note that it contains those elements for which $\widehat{\balpha}_i \leq \log(d)/n$,
precisely those for which $\widehat{\balpha}_i$ cannot give a good multiplicative approximation to.

\subsubsection{Moment estimation}
Having bucketed $\alpha$'s probability values,
local moment matching proceeds by estimating the moments of $\alpha$ within each bucket.
For a given bucket $B_j$
and moment $k$, we would like to estimate the $k$-th moment restricted to $B_j$, given by
\begin{equation*}
    p_k(\alpha|_{B_j}) = \sum_{i \in B_j} \alpha_i^k.
\end{equation*}
However, especially for buckets containing large probability values,
this $k$-th moment can be poorly behaved with respect to small errors in our estimates of the $\alpha_i$'s.
To address this, we will recenter the $\alpha_i$'s around the midpoint of the bucket $m_j$ and instead estimate the \emph{centered power sum polynomial}
\begin{equation*}
    p_{k, j}(\alpha) = \sum_{i \in B_j} (\alpha_i - m_j)^k.
\end{equation*}
It is possible to modify the collision-based estimator unbiased estimator for $p_k(\cdot)$ from \Cref{eq:collision-estimator} to give an unbiased estimator for $p_{k, j}(\alpha)$.
This allows us to produce estimates $\widehat{\bp}_{1, j}, \ldots, \widehat{\bp}_{K, j}$ for the centered moments $p_{1, j}(\alpha), \ldots, p_{K, j}(\alpha)$, where $K$ is some integer of our choice.

Let us consider how this works for the smallest bucket $B_1$, consisting of those probability values which are at most $\log(d)/n$.
Since these values are small, it turns out that the $k$-th moment is \emph{not} poorly behaved with respect to small errors in the estimates of the $\alpha_i$'s, and so it suffices to directly estimate the un-centered moment $p_k(\alpha |_{B_1})$ rather than the centered moment $p_{k, 1}(\alpha)$. 
To estimate this, we use the following natural modification of the $k$-wise collision statistic:
\begin{equation*}
    c_{k,B_1}(\bx) = \frac{1}{\binom{n}{k}} \cdot \sum_{a \in B_1}\sum_{1 \leq i_1 < \cdots < i_k \leq n}  \bone[\bx_{i_1} = \cdots = \bx_{i_k}=a].
\end{equation*}
Routine calculations show that
\begin{equation}\label{eq:routine}
    \E[c_{k, B_1}(\bx)] = p_k(\alpha |_{B_1}),
    \quad
    \text{and}
    \quad
    \Var[c_{k, B_1}(\bx)]
    \leq \frac{1}{\binom{n}{k}}\sum_{i=1}^k \binom{k}{i} \cdot \binom{n-k}{i} \cdot p_{k + i}(\alpha |_{B_1}).
\end{equation}
Hence, $c_{k, B_1}(\bx)$ is an unbiased estimator for $p_k(\alpha |_{B_1})$.
To bound the variance, we will crudely bound each $\alpha_i$ within $B_1$ by the largest possible value $L = \log(d)/n$, which allows us to bound $p_{k+i}(\alpha |_{B_1}) \leq |B_1| \cdot L^{k+i}$.
Hence,
\begin{align*}
    \Var[c_{k, B_1}(\bx)]
    &\leq \frac{1}{\binom{n}{k}}\sum_{i=1}^k \binom{k}{i} \cdot \binom{n-k}{i} \cdot |B_1| \cdot L^{k+i}\\
    & = |B_1| \cdot L^{2k} \cdot \sum_{i=1}^k \frac{\binom{k}{i} \cdot \binom{n-k}{i}}{\binom{n}{k}} \cdot L^{i-k}\\
    &\leq |B_1| \cdot L^{2k} \cdot \sum_{i=1}^k 2^k \cdot \Big(\frac{k}{n-k}\Big)^{k-i} \cdot L^{i-k}\\
    &\leq |B_1| \cdot L^{2k} \cdot \sum_{i=1}^k 2^k \cdot \Big(\frac{k}{(n-k) \cdot L}\Big)^{k-i},
\end{align*}
where the second inequality uses a binomial coefficient identity that we prove in \Cref{eq:use-this-in-intro} below.
So long as we only estimate moments $1 \leq k \leq K$, where $K \leq O(\log(d))$,
then we have that $k/((n-k) \cdot L) \approx k/(n \cdot L) = O(\log(d))/(n \cdot L) = O(1)$ because of our choice of $L$, and so the variance is bounded above by
\begin{equation}\label{eq:classical_var_bound}
    \Var[c_{k, B_1}(\bx)]
    \leq |B_1| \cdot L^{2k} \cdot \sum_{i=1}^k 2^k \cdot O(1)^{k-i}
    = |B_1| \cdot L^{2k} \cdot 2^{O(k)}.
\end{equation}
Applying Chebyshev's inequality, we expect that
\begin{equation}\label{eq:chebyshev-moment}
    c_{k, B_1}(\bx) = p_k(\alpha|_{B_1}) \pm t \cdot \sqrt{|B_1|} \cdot L^k \cdot 2^{O(k)},
\end{equation}
except with probability $1/t^2$.
Heuristically, if we assume that each $\alpha_i$ is roughly equal to the maximum value of $L$, then $p_k(\alpha|_{B_1}) \approx |B_1| \cdot L^k$, which means that this gives us a \emph{multiplicative} approximation of the $k$-th moment of the form
\begin{equation}\label{eq:look-at-the-denom}
    c_{k, B_1}(\bx) = p_k(\alpha|_{B_1})\cdot (1 \pm t \cdot 2^{O(k)}/\sqrt{|B_1|}).
\end{equation}
Indeed, we saw in \Cref{ex:uniformity} that a multiplicative approximation to the second moment (enough to distinguish $p_2(\alpha) = 1/d$ versus $2/d$) is needed to distinguish the uniform distribution from a distribution which is uniform on half the entries, and so this is the type of guarantee we will need.
Let us mention briefly that we will typically choose $K = c \log(d)$, for $c$ an arbitrarily small constant, in which case the $2^{O(k)} \leq 2^{O(K)}$ factor will scale as $d^{O(c)}$, which is a small and manageable polynomial in $d$; for example, if $|B_1| = \Theta(d)$, then it will be dwarfed by the denominator of $\sqrt{|B_1|}$ in \Cref{eq:look-at-the-denom}.

There are $M = \sqrt{n/\log(d)}$ buckets and $K = O(\log(d))$ moments to estimate within each bucket.
Since our application of Chebyshev's inequality has failure probability $1/t^2$, we need to set $t^2 \gg M K$, i.e.\ $t \gg \sqrt[4]{n \log(d)}$, in order to be able to union bound over all $MK$ moments. 
This introduces an error into \Cref{eq:chebyshev-moment} which is too large for this statement to be useful.
To address this, \cite{HJW18} use Hoeffding's inequality to prove a stronger concentration bound for their moment estimators, showing that \Cref{eq:chebyshev-moment} holds except with probability $\exp(-t^2)$. This allows them to take $t$ to be a much smaller $t = \sqrt{\log(MK)}$.
Let us note, however, that \emph{our} eventual quantum algorithm will only need to estimate a small number (roughly $\log(d)$) of moments, in which case it will suffice to analyze the moment estimators by computing their variance and applying Chebyshev's inequality.

\subsubsection{Moment matching}\label{sec:moment-matching-intro}

Now that we have estimated the moments for each bucket, we want to apply moment matching within each bucket. For each bucket $B_j$, this entails computing a sub-distribution $\widehat{\balpha}_{B_j}$ which is supported on the interval $I_j$ corresponding to $B_j$ whose centered moments approximately match $\widehat{\bp}_{1, j}, \ldots, \widehat{\bp}_{K, j}$.
This will serve as our estimate of $\alpha|_{B_j}$.
That such a sub-distribution exists follows from the fact that $\alpha|_{B_j}$ itself is supported on the interval corresponding to $B_j$ and has centered moments which approximately match $\widehat{\bp}_{1, j}, \ldots, \widehat{\bp}_{K, j}$;
however, there might be other sub-distributions which approximately match these moments as well, and as part of the proof we must show that these distributions are close to $\alpha|_{B_j}$.
Actually finding such a sub-distribution $\widehat{\balpha}_{B_j}$ can be done, albeit inefficiently,  by brute force searching over possible sub-distributions until one is found which approximately matches the learned moments

However, it turns out that searching for this sub-distribution can also be cast as a linear program, and~\cite{HJW18} give an algorithm for rounding this linear program and show how to analyze it.
To explain the guarantees that this algorithm has, let us again focus on the case of the smallest bucket $B_1$.
Then their rounding algorithm produces an estimate $\widehat{\balpha}_{B_1}$ such that
\begin{align}\label{eq:lp-rounding}
    \E \dtv{\alpha|_{B_1}^{\geq}}{\widehat{\balpha}_{B_1}^{\geq}} &= O\Big(\vphantom{\sum_{k=1}^K}\frac1K\sqrt{Ld} + 25^{K}L\sum_{k=1}^K L^{-k} \Big\lvert p_k(\alpha|_{B_1}) - \widehat{\bp}_{k,1}\Big\rvert\Big),
\end{align}
where again we are writing (i) $L = \log(d)/n$ for the largest probability value in bucket $B_1$  and (ii) $\widehat{\bp}_{k,1}$ for the estimate of the $k$-th moment rather than the $k$-th central moment.
In this expression, there are two sources of error which govern how close $\widehat{\balpha}_{B_j}$ is to $\alpha|_{B_j}$, which are referred to as the \emph{bias} and the \emph{variance}, given by the first and second terms, respectively.
The bias corresponds to the error we incur from only learning the first $K$ moments of $\alpha|_{B_j}$, and the variance corresponds to the error we incur from estimating these moments rather than computing them exactly.
Increasing the number $K$ of moments that we estimate decreases the bias term but increases the variance term,
and so these two sources of error have to be traded off with each other when picking the number of moments $K$.
Note that the error of the $k$-th estimate $\widehat{\bp}_{k,1}$ from the true value of $p_k(\alpha|_{B_1})$ is penalized by an additional factor of $L^{-k}$, meaning that higher moments must be estimated to better accuracy than lower moments.
Indeed, our modified collision estimators from \Cref{eq:chebyshev-moment} have an error which scales as $L^k$, which nicely cancels with the $L^{-k}$ ``penalty'' factor.

\subsubsection{Putting it all together.}
Now we sketch and analyze a simple local moment matching algorithm in order to illustrate how all of the ingredients combine. Our algorithm will follow the same outline as the algorithm sketched at the beginning of \Cref{sec:lmm} which involves splitting the sample into just two buckets.
Although it will not achieve the optimal $n = O(d/(\log(d) \cdot \epsilon^2))$ sample complexity, it will still improve on the trivial bound of $n = O(d/\epsilon^2)$ which comes from using the empirical sorted distribution, and it will serve as an inspiration for our eventual quantum algorithm.
\begin{enumerate}
    \item (Bucketing): 
    Draw $n$ samples $\bx = (\bx_1, \ldots, \bx_n)$.
    Let $\widehat{\bp}$ be the empirical distribution of $\bx$. Pick a threshold $0 \leq L \leq 1$ and set
    \begin{equation*}
    \mathsf{Large} = \{i \mid \widehat{p}_i \geq L\}\quad\text{and}\quad\mathsf{Small} = [d] \setminus \mathsf{Large}.
    \end{equation*}
    \item (Estimating the large bucket): Since every probability in $\mathsf{Large}$ is at least $L$, $\mathsf{Large}$ has at most $1/L$ items.
    Use $O(L^{-1}\epsilon^{-2})$
    samples to produce an estimate $\widehat{\alpha}_{\mathsf{Large}}$ of $\alpha|_{\mathsf{Large}}$.
    \item (Estimating the small bucket): Draw $n$ more samples $\by = (\by_1, \ldots, \by_n)$ in order to estimate the first $K$ moments of $\alpha|_{\mathsf{Small}}$.
    Compute the collision statistics $c_{1, \mathsf{Small}}(\by), \ldots, c_{K, \mathsf{Small}(\by)}$, and use these as estimates of these moments.
    Use moment matching to compute an estimate $\widehat{\balpha}_{\mathsf{Small}}$ of $\alpha|_{\mathsf{Small}}$.
    \item Output $(\widehat{\balpha}_{\mathsf{Large}}, \widehat{\balpha}_{\mathsf{Small}})$ as the final estimate of $\alpha$.
\end{enumerate}

Let us now analyze this algorithm in order to choose the parameters $n$, $L$, and $K$. First, we want all 3 steps to consume $n$ samples of $\alpha$, which means that in the 3rd step we need $O(L^{-1} \epsilon^{-2}) = C \cdot L^{-1} \epsilon^{-2} \leq n$, for some constant $C \geq 1$. We can achieve this so long as $L \geq C /(n \epsilon^2)$. However, we have also seen in \Cref{eq:interval} that for bucketing to work, we want $L \geq \log(d)/n$. To satisfy both of these, we will set $L = C \log(d)/(n \epsilon^2)$. In addition, we have argued in the moment estimation section that we will want the number of moments $K = c \log(d)$ for some small constant $c > 0$.

Since we are using $O(L^{-1}\epsilon^2)$ samples in the second step, $\widehat{\bp}_{\mathsf{Large}}$ will be a good estimate of $p|_{\mathsf{Large}}$ with high probability, and so it suffices to analyze the third step. From \Cref{eq:chebyshev-moment}, we know that for each $1 \leq k \leq K$,
\begin{equation*}
    c_{k, \mathsf{Small}}(\bx) = p_k(\alpha|_{\mathsf{Small}}) \pm t \cdot \sqrt{|\mathsf{Small}|} \cdot L^k \cdot 2^{O(k)},
\end{equation*}
except with probability $1/t^2$.
(Above, this bound was argued assuming that $L = \log(d)/n$, but the proof only uses the fact that $L \geq \log(d)/n$, which is the case here.)
We will apply the trivial bound $|\mathsf{Small}| \leq d$.
In addition, to be able to union bound over all $K$ moments, we will take $t = \sqrt{K}$. This gives us that for all $1 \leq k \leq K$, with high probability,
\begin{equation*}
    c_{k, \mathsf{Small}}(\bx) = p_k(\alpha|_{\mathsf{Small}}) \pm \sqrt{Kd}\cdot L^k \cdot 2^{O(k)},
\end{equation*}
Now applying our moment matching guarantee in \Cref{eq:lp-rounding}, we have
\begin{align*}
    \E \dtv{\alpha|_{\mathsf{Small}}^{\geq}}{\widehat{\balpha}_{\mathsf{Small}}^{\geq}} &= O\Big(\vphantom{\sum_{k=1}^K}\frac1K\sqrt{Ld} + 25^{K}L\sum_{k=1}^K L^{-k} \cdot \sqrt{Kd} \cdot L^k \cdot 2^{O(k)}\Big)\\
&=O\Big(\vphantom{\sum_{k=1}^K}\frac1K\sqrt{Ld} + 2^{O(K)}L\sqrt{d}\Big)\\
&=O\Big(\sqrt{\frac{C d}{c^2 \log(d) n \epsilon^2}} + d^{0.5 + O(c)}\cdot \frac{C \log(d)}{n \epsilon^2}\Big),
\end{align*}
where in the last step we have plugged in our settings of $L$ and $K$.
We are aiming for a total variation distance of at most $\epsilon$.
So long as $c$ is chosen to be small enough so that $d^{0.5 + O(c)} \leq d$, then both terms can be made to satisfy this by setting $n = O(d/(\log(d) \epsilon^4))$.
This gives our final sample complexity for $n$, which improves on the trivial bound of $n = O(d/\epsilon^2)$ for sufficiently large $\epsilon$, i.e.\ whenever $\epsilon \geq 1/\sqrt{\log(d)}$.

\section{Technical overview of the quantum case}\label{sec:technical_overview_quantum}

In the quantum setting, we are given $n$ copies of a mixed state $\rho \in \C^{d \times d}$ with spectrum $\alpha = (\alpha_1, \ldots, \alpha_d)$, where $\alpha_1 \geq \cdots \geq \alpha_d$.
Our goal is to produce an estimate $\widehat{\balpha}$ of $\alpha$,
and our approach for doing so will be inspired by the framework of local moment matching.
As in the beginning of \Cref{sec:lmm},
we will divide $\alpha$ into just two buckets,
the first containing the large elements and the second containing the small elements.
We will learn the elements in the large bucket by using a simple empirical estimator, and we will learn the elements in the small bucket by estimating their moments and applying local moment matching.
Below, we describe how we bucket and learn moments in the quantum setting,
and explain our decision to use two buckets.

\subsection{Bucketing}
As in the classical case,
we will take $2n$ copies of $\rho$
and split them into two batches of size $n$.
We will use the first batch to learn a projective measurement $\{\bPi, \overline{\bPi}\}$, where $\bPi$ is intended to be the projection onto $\rho$'s largest eigenvalues and $\overline{\bPi}$ is intended to be the projection onto $\rho$'s smallest eigenvalues.
Having done this, we will measure the remaining $n$ copies of $\rho$ with the $\{\bPi, \overline{\bPi}\}$ measurement;
for those copies where we receive the $\bPi$ outcome,
it is as if we are sampling from the large part of $\alpha$,
and for those copies where we receive the $\overline{\bPi}$ outcome,
it is as if we are sampling from the small part of $\alpha$.
We can view this process as converting the second half of our copies of $\rho$ into copies of the state $\bPi \rho \bPi + \overline{\bPi} \rho \overline{\bPi}$.

To learn $\{\bPi, \overline{\bPi}\}$, we will run a tomography algorithm on the first $n$ copies of $\rho$ to produce an estimate $\widehat{\brho}$ of $\rho$.
This estimate can be written as $\widehat{\brho} = \bU \cdot \widehat{\balpha} \cdot \bU^{\dagger}$,
where $\widehat{\balpha}$ is an estimate for $\rho$'s spectrum $\alpha$ and $\bU$ is an estimate for $\rho$'s eigenvectors. Assuming that $\widehat{\balpha} = (\widehat{\balpha}_1, \ldots, \widehat{\balpha}_d)$ is sorted, so that $\widehat{\balpha}_1 \geq \cdots \geq \widehat{\balpha}_d$,
we will select a threshold $\tau$
and define $\bk$ to be the largest index such that $\widehat{\balpha}_{\bk} \geq \tau$.
Then $\alpha_1, \ldots, \alpha_{\bk}$ correspond to 
the $\mathsf{Large}$ eigenvalues and $\alpha_{\bk+1}, \ldots, \alpha_{d}$ correspond to the $\mathsf{Small}$ eigenvalues.
We can then define the projection onto $\widehat{\balpha}$'s top $\bk$ eigenvalues as
\begin{equation*}
    \bPi = \bU \cdot (\ketbra{1} + \cdots + \ketbra{\bk}) \cdot \bU^{\dagger},
\end{equation*}
which will serve as our estimate for the projection onto $\rho$'s top $\bk$ eigenvalues as well.
We can then set $\overline{\bPi} = I - \bPi$ and we have our projective measurement.

Which tomography algorithm to pick?
If we want to use entangled measurements,
we could use Keyl's algorithm~\cite{Key06},
which was analyzed in~\cite{OW16},
or either of the entangled tomography algorithms from Haah et al.~\cite{HHJ+16}.
If we want to use unentangled measurements,
then one option is the uniform POVM algorithm independently due to Krishnamurthy and Wright~\cite[Section 5.1]{Wri16} and Guta et al.~\cite{GKKT20},
or we could use either of the more recent algorithms of Chen et al.~\cite{CHL+23} or Flammia and O'Donnell~\cite{FO24} which achieve near-optimal copy complexity for estimation in fidelity.
In principle, we believe that many of these algorithms are a good choice, but in practice some are significantly more easy to analyze than others.

To see why, let us consider the sources of error that incur in this bucketing step. Recall that after learning the measurement $\{\bPi, \overline{\bPi}\}$, we convert the remaining $n$ copies of $\rho$ to the state $\bPi \rho \bPi + \overline{\bPi} \rho \overline{\bPi}$. We want this state to satisfy two properties.
First, $\bPi \rho \bPi$ should only contain large eigenvalues and $\overline{\bPi} \rho \overline{\bPi}$ should only contain small eigenvalues;
if this does not occur, we call it a \emph{misclassification error}.
Second, the spectrum of $\bPi \rho \bPi + \overline{\bPi} \rho \overline{\bPi}$ should be close to the spectrum of $\rho$; if this does not occur, we call it an \emph{alignment error}, referring to the fact that $\{\bPi, \overline{\bPi}\}$ is not properly aligned with $\rho$'s eigenbasis.

\paragraph{Misclassification error.}
In the classical case of local moment matching,
misclassification error corresponds to placing some probability value $\alpha_i$ into the wrong bucket $B_j$.
There, we argued that this wouldn't happen with high probability because our estimator $\widehat{\balpha}$ was a good estimator of $\alpha$ ``in an $\ell_{\infty}$ sense'', meaning that each coordinate $\widehat{\balpha}_i$ was (multiplicatively) close to $\alpha_i$, for all $i$.
Our analysis suggests that this is also the case in the quantum setting: if we can guarantee that $\widehat{\brho}$ is close to $\rho$ in $\ell_{\infty}$ norm, then we can avoid misclassification error. Unfortunately, of the above tomography algorithms, the only two that are known to give $\ell_{\infty}$ norm guarantees are the the uniform POVM algorithm (due to the analysis of Guta et al.~\cite{GKKT20}) and the Chen et al.~\cite{CHL+23} fidelity algorithm.
This rules out using entangled measurements (at least, given our current understanding of these entangled measurements) and is the reason why we only consider unentangled measurements in this paper. In particular, we choose the uniform POVM algorithm.

\paragraph{Alignment error.}
Alignment error, on the other hand, is entirely a quantum phenomenon.
In the classical case, even if you misclassify some probability values, the distribution your second batch of samples are drawn from is still~$\alpha$.
But in the quantum case, measuring $\rho$ with $\{\bPi, \overline{\bPi}\}$ will inevitably disturb the state, and so we need to bound the total amount of disturbance that occurs.
We show several ways to do so.
First, we show that this disturbance can be bounded in the case that the tomography algorithm we use is able to perform principal component analysis (PCA) tomography. To expand on this, let $\widehat{\brho}_{\leq k} = \bPi \cdot \widehat{\brho} \cdot \bPi$ be the projection onto $\widehat{\brho}$'s top $\bk$ eigenvalues.
If $\widehat{\brho}_{\leq \bk}$ happened to perfectly equal the projection of $\rho$ onto its top $\bk$ eigenvalues, then we would have
\begin{equation*}
    \dtr(\rho, \widehat{\brho}_{\leq \bk}) = \alpha_{\bk+1} + \cdots + \alpha_d.
\end{equation*}
If this equation is satisfied up to error $\epsilon$, then the algorithm is performing trace distance rank-$\bk$ PCA up to error $\epsilon$, and we show that if this condition is satisfied,
we will only introduce $\epsilon$ error when we measure $\rho$ with $\{\bPi, \overline{\bPi}\}$.
In our case, the uniform POVM algorithm's $\ell_{\infty}$ norm guarantees are essentially strong enough to show that it gives a rank-$\bk$ trace distance PCA algorithm
(although we are even able to give a direct proof that the uniform POVM has small alignment error, short-cutting around trace distance PCA).
In fact, we can strengthen this result and show that it actually suffices to perform \emph{fidelity} PCA up to error $\epsilon$, rather than the more costly trace distance PCA.

\paragraph{Related work.}
Let us conclude by discussing the fidelity tomography algorithms of Chen et al.~\cite{CHL+23} and Flammia and O'Donnell~\cite{FO24},
whose strong similarities with our bucketing step we became aware of partway through this project.
Among many other results, both of these works show that rank-$r$ fidelity tomography can be performed with unentangled measurements using $n = \widetilde{O}(d r^2/\epsilon)$ copies of $\rho$.
Their starting point is the basic uniform POVM tomography algorithm, which gives optimal copy complexities for $\ell_{\infty}$, $\ell_{1}$, and~$\ell_2$ unentangled tomography, but cannot give an optimal copy complexity for fidelity tomography as it is a nonadaptive algorithm (see~\cite{CHL+23}, which shows that any rank-$r$ nonadaptive algorithm for fidelity tomography must use $\Omega(dr^2/\epsilon^2)$ copies). Learning in fidelity requires learning $\rho$ to higher accuracy on its small eigenvalues than on its large eigenvalues, but the uniform POVM is unable to do so as the presence of the large eigenvalues interferes with learning the small eigenvalues. Intriguingly, this is highly reminiscent of the issue with learning moments in the classical setting that motivated the local moment matching approach.

To deal with this issue, they proceed in an iterative approach.
In the first round, they run the uniform POVM algorithm to produce an estimate $\widehat{\brho}_1$ of $\rho$.
They let $\bPi_{1}$ be the projection onto the ``large'' eigenvalues of $\widehat{\brho}_1$ and set $\overline{\bPi}_{1} = I - \bPi_{1}$.
Then they measure all remaining copies of $\rho$ with $\{\bPi_1, \overline{\bPi}_1\}$;
those for which the second outcome was observed have collapsed to $\overline{\bPi}_1 \cdot \rho \cdot\overline{\bPi}_1$, which should be the projection onto $\rho$'s smaller eigenvalues, and then they recurse this procedure onto these states.
The result is a sequence of projectors $\bPi_1, \bPi_2, \ldots$ for which $\bPi_1$ should project onto $\rho$'s highest eigenvalues, $\bPi_2$ should project onto its next highest eigenvalues, and so forth.

To be precise, this is the guarantee that the Chen et al.~\cite{CHL+23} algorithm provides.
They make use of the $\ell_{\infty}$ tomography guarantee of the uniform POVM algorithm due to Guta et al.~\cite{GKKT20}, which allows them to control the magnitude of the eigenvalues which fall within each bucket $\bPi_i$.
The Flammia and O'Donnell algorithm, on the other hand, only requires the weaker $\ell_2$ tomography guarantee of the uniform POVM, but as a result it is not able to precisely control the magnitude of the eigenvalues within each bucket.
This means that of the two, the Chen et al.\ algorithm appears to be more suitable for our purposes,
and we believe that a modification of it can be shown to successfully split $\rho$ into multiple buckets \`a la 
local moment matching with small misclassification and alignment error.

The reason we use the uniform POVM rather than the Chen et al.\ algorithm is that our algorithm will use $O(\eps^{-6}\cdot d^3\cdot (\log \log (d) / \log(d))^4 )$ copies, whereas we believe that the best we could hope for by using the Chen et al.\ algorithm is $O(\eps^{-5}\cdot d^3\cdot (\log \log (d) / \log(d))^4 )$ copies, at the expense of significant added complexity in the algorithm description and proof of correctness.
(This would entail having to re-analyze the Chen et al.\ algorithm in addition to implementing local moment matching in buckets of various sizes, rather than just the small bucket.)
Since we believe that $O(\eps^{-5})$ is still not the optimal dependence on $\epsilon$, we have opted to prioritize the simplicity of our algorithm over a slight improvement in copy complexity.

\subsection{Moment estimation}\label{sec:quantum-moment-estimation}

The final step is to estimate the moments of the small part of the state $\overline{\bPi} \rho \overline{\bPi}$ and perform local moment matching.
Before discussing how to estimate the moments of $\overline{\bPi} \rho \overline{\bPi}$, which is in general a subnormalized state,
let us first discuss how to estimate the moments of a properly normalized quantum state $\sigma$.
Given~$\sigma$, estimating its moments $\tr(\sigma^k)$ is a well-studied topic in quantum information, and there are various off-the-shelf estimators available for our use.
For example, if we were in the entangled setting,
we could use the estimators introduced in~\cite{OW15}, which were further studied in~\cite{AISW20} and~\cite{BOW19};
in particular, the latter work reinterpreted these estimators as natural quantum analogues of the classical collision-based estimators from \Cref{eq:collision-estimator} and showed that these are the minimum-variance unbiased estimators for the moments $\tr(\sigma^k)$.

We are working in the unentangled setting, so we use a different estimator.
Ours is based on the fact that the uniform POVM tomography algorithm, when run on a single copy of $\sigma$, outputs a matrix $\widehat{\bsigma}$ which is an unbiased estimator for $\sigma$, meaning that $\E \widehat{\bsigma} = \sigma$.
With $k$ copies of $\sigma$, then, we can generate $k$ independent copies of this estimator $\widehat{\bsigma}_1, \ldots, \widehat{\bsigma}_k$; given these,  $\tr(\widehat{\bsigma}_1 \dots \widehat{\bsigma}_k)$ is an unbiased estimator of $\tr(\sigma^k)$.
Generalizing this to $n$ copies of $\sigma$, we have the corresponding U-statistic
\begin{equation*}
    \bZ_k \coloneqq \frac{1}{n(n-1)\cdots (n-k+1)}\cdot\sum_{\text{distinct }i_1, i_2, \ldots, i_k \in [n]} \tr\left(\widehat{\bsigma}_{i_1} \widehat{\bsigma}_{i_2} \cdots \widehat{\bsigma}_{i_k}\right).
\end{equation*}
This unbiased estimator is a natural non-commutative generalization of the collision estimator in \Cref{eq:collision-estimator}.
To our knowledge, we are the first to explicitly study this estimator.
That said, similar estimators have appeared in the literature before; for example, it can be viewed as a special case of an estimator for nonlinear functions of $\sigma$ proposed in~\cite{HKP20}.
In addition, a related estimator for $\tr(\rho\sigma)$, where $\rho$ and $\sigma$ are two distinct quantum states, was proposed and analyzed in~\cite[Appendix D]{ALL22}.

Our main technical result is the following variance bound on $\bZ_k$ (cf.\ \Cref{eq:for-use-in-intro?} below):
\begin{equation}\label{eq:cool-variance-bound}
    \Var[\bZ_k] \leq \frac{1}{\binom{n}{k}} \sum_{i=0}^{k-1} \binom{k}{k-i} \binom{n-k}{i}\cdot 6^k \cdot d^{k-i-1} \cdot \tr(\sigma^{2i}).
\end{equation}
This is a direct analogue to the variance bound for the classical collision estimators from \Cref{eq:routine},
though it is worse due to the $d^{k-i-1}$ term and the presence of ``small'' moments $\tr(\sigma^0) = 1, \ldots, \tr(\sigma^k)$ which do not appear in the classical bound.
This is of course as expected, as estimating moments in the quantum case should only be more difficult than in the classical case.
As one application of this variance bound, we are able to show that $\bZ_k$ approximates $\tr(\sigma^k)$ with multiplicative error bounds; in particular, we show that for a fixed constant $k$,
given
\begin{equation} \label{eq:intro-moment}
    n = O\left(\max \left\{\frac{d^{2-2/k}}{\delta^2}, \frac{d^{3 - 2/k}}{\delta^{2/k}} \right\}\right)
\end{equation}
copies of a state $\sigma$, the estimator satisfies
\begin{equation}\label{eq:multiplicative-error-bound-intro}
    (1-\delta)\cdot\tr(\sigma^k) < \bZ_k < (1+\delta)\cdot\tr(\sigma^k)
    \end{equation}
with probability at least 99\%.
Here, the $O(\cdot)$ is hiding a $k^k$ dependence, which is a constant so long as $k$ is a constant.
(We note that a similar $k^k$ factor appears in the sample complexities of both the classical and the entangled quantum moment estimators~\cite{AOST17,AISW20}.)
As an corollary, this immediately implies the sample complexity bound for quantum R\'{e}nyi entropy estimation given in \Cref{thm:quantum-renyi}.
Even for $k = 2$, our variance bound slightly improves on the bound given in~\cite[Appendix D]{ALL22},
which is why we can show multiplicative error bounds versus their additive error bounds.
We provide a detailed comparison between our algorithm and the algorithms of~\cite{HKP20,ALL22} in \Cref{sec:moment_estimation}.

For our downstream application of moment estimation to spectrum learning, we need to modify the estimator $\bZ_k$ to apply to subnormalized states of the form $\bsigma = \overline{\bPi} \rho \overline{\bPi}$.
This is relatively straightforward and can be done by first measuring $\rho$ according to $\{\bPi, \overline{\bPi}\}$ and using those samples which fall in $\overline{\bPi}$ in the estimator. The result is an unbiased estimator $\bY_k$ for the $k$-th moment $\tr(\bsigma^k)$.
Furthermore, one can adapt the analysis of the variance bound from \Cref{eq:cool-variance-bound} to show an analogous bound for $\bY_k$, which we will use to show concentration of $\bY_k$.

The proof of our variance bound is significantly more challenging than in the classical case and follows from a careful analysis of our estimator's second moment, $\E[\bZ_k^2]$.
This second moment expands to an average over products of two traces,
\begin{equation*}
    \tr\left(\widehat{\bsigma}_{i_1} \widehat{\bsigma}_{i_2} \cdots \widehat{\bsigma}_{i_k}\right)
    \cdot \tr\left(\widehat{\bsigma}_{j_1} \widehat{\bsigma}_{j_2} \cdots \widehat{\bsigma}_{j_k}\right).
\end{equation*}
When all of the indices above are distinct, this trace product is $\tr(\sigma^k)^2$ in expectation, matching $\E[\bZ_k]^2$.
When the $i$'s and $j$'s have $t \geq 1$ indices in common, we use the trick that
\begin{equation*}
    \tr\left(\widehat{\bsigma}_{i_1} \widehat{\bsigma}_{i_2} \cdots \widehat{\bsigma}_{i_k}\right)
    \cdot \tr\left(\widehat{\bsigma}_{j_1} \widehat{\bsigma}_{j_2} \cdots \widehat{\bsigma}_{j_k}\right)
    = \tr\left(P \cdot \parens[\Big]{\widehat{\bsigma}_{i_1} \otimes \cdots \otimes \widehat{\bsigma}_{i_k} \otimes \widehat{\bsigma}_{j_1} \otimes \cdots \otimes \widehat{\bsigma}_{j_k}}\right)
\end{equation*}
for the permutation matrix $P$ that rearranges qudits in the appropriate way.
We can then bound the expectation of this expression to get something which degrades with $t$: specifically, our bound is $6^k d^{t-1} \tr(\sigma^{2(k-t)})$ as shown in \Cref{eq:individual_term_in_var_bound}.
The dependence of the second moment on $n$ comes from the distribution over $t$: the probability of two random subsets $i, j \subseteq [n]$ having $t$ elements in common is about $1/n^t$, so as $n$ grows large, $\E[\bZ_k^2]$ tends to the $t = 0$ case, $\E[\bZ_k]^2$.
Appropriately balancing these parameters gives the copy complexity in \Cref{eq:intro-moment}.

\paragraph{Related work.}
Curiously, the copy complexity of moment estimation to (constant) multiplicative error in the unentangled measurement setting appears to be open.
Our estimator shows a bound of $n = O(d^{3 - 2/k})$ for any constant $k$, but this may be sub-optimal: we measure our copies of the state with a fixed POVM, which has been shown to lead to worse complexities in some other settings~\cite{LA24}.
The best lower bound for multiplicative-error moment estimation comes from the \emph{fully entangled} setting, and is $n = \Omega(d^{2 - 2/k})$~\cite{AISW20}.
It is not clear to us whether $d^{3-2/k}$ is the correct scaling: the existing literature does not rule out the possibility of a scaling of $d^{3-3/k}$, for example.

We survey this literature now.
In the unentangled setting, it has focused on the $k = 2$ setting of estimating $\tr(\sigma^2)$, the purity of $\sigma$.
Prior work gives estimators for the purity which involve repeatedly measuring $\sigma$ in a Haar random basis~\cite{ALL22}.
The best-known upper and lower bounds~\cite{ALL22,GHYZ24} for estimating purity to \emph{additive} error do not resolve the question of estimating to \emph{multiplicative} error:
the upper bound only gives $n = O(d^2)$ for estimating to multiplicative error, and the lower bound of $n = \Omega(d^{1/2})$ is too loose if directly translated to constant multiplicative error.
A crucial setting for multiplicative-error moment estimation is when the input state is close to maximally mixed; so, a closely related task is to distinguish whether $\sigma$ is maximally mixed or constant far from maximally mixed.
For this, $n = \Theta(d^{3/2})$ copies of $\sigma$ are sufficient~\cite{BCL20} and necessary~\cite{CHLL22}.
In the classical setting, this task is solved by computing an unbiased estimator for the purity, but these results in the quantum setting do not give good estimates on the purity, despite being closely related to the purity estimator of \cite{ALL22}.

In summary, $d^{3/2}$ could be the correct scaling for estimating purity to multiplicative error in the unentangled setting, which extrapolates to a scaling of $d^{3 - 3/k}$ for general $k$.
So, there may be room to improve unentangled moment estimators, even for $k = 2$.
This is not the bottleneck of our argument, though, so we do not attempt to optimize them further.

\subsection{Putting everything together}
Now let us describe how these ingredients combine to give our spectrum estimating algorithm.
Let $B$ be our intended upper bound on the ``small bucket'' eigenvalues.
We first run the uniform tomography algorithm to produce a measurement $\{\bPi, \overline{\bPi}\}$ which buckets $\rho$ into its large and small eigenvalues, respectively.
We show in \Cref{thm:bucket-algo} that if we use $n = O(dB^{-2}\eps^{-2})$ copies of $\rho$ to learn $\{\bPi, \overline{\bPi}\}$,
then we will achieve alignment error at most $\epsilon$,
and all eigenvalues in $\bsigma = \overline{\bPi} \cdot \rho \cdot \overline{\bPi}$ will be at most $2B$.
Furthermore, this theorem also shows that the spectrum of the large bucket, $\spec(\bPi \rho \bPi)$, can also be estimated up to error $\epsilon$ with this number of samples.
Thus, it remains to estimate the spectrum of the small bucket $\bsigma$,
which we denote $\bbeta = \{\bbeta_i\}$.

To do this, we take $n = O(dB^{-2} \epsilon^{-2})$ copies of $\rho$, measure all of them with the uniform POVM, and compute the moment estimators $\bY_1, \ldots, \bY_K$ from \Cref{sec:quantum-moment-estimation} for some number of moments $K$ to be specified later.
Let us note that since each of these estimators relies on samples from the uniform POVM, we can reuse the same samples to compute all $K$ estimators.
For our number of samples $n$, we are able to show the following variance bound on $\bY_k$:
\begin{equation*}
    \Var[\bY_k] \leq B^{2k} \cdot k^{O(k)} \cdot \epsilon^2,
\end{equation*}
which is analogous to the classical variance bound in \Cref{eq:classical_var_bound}.
The key difference between these two bounds is that the factor of $2^{O(k)}$ in the classical bound is replaced by a factor of $k^{O(k)}$ in the quantum bound; this difference means that although we can use $K = O(\log(d))$ moments classically, we will only be able to use $K = O(\log(d)/\log \log(d))$ moments quantumly.
Applying Chebyshev's inequality, we have that
\begin{equation*}
    \bY_k = \tr(\bsigma^k) \pm t \cdot B^k \cdot k^{O(k)} \cdot \epsilon,
\end{equation*}
except with probability $1/t^2$.
In order to union bound over all $K$ moments, we will set $t = \sqrt{K}$, in which case we get that with high probability,
\begin{equation*}
    \bY_k = \tr(\bsigma^k) \pm \sqrt{K} \cdot B^k \cdot k^{O(k)} \cdot \epsilon
\end{equation*}
for all $1 \leq k \leq K$.
At this point, converting these estimates of $\bsigma$'s moments to an estimate of $\bsigma$'s spectrum is a purely classical problem, and it can be solved by appealing to the moment matching algorithm from \Cref{sec:moment-matching-intro}.
In particular, that algorithm will produce an estimate $\widehat{\bbeta}$ of $\bbeta$, and \Cref{eq:lp-rounding} provides the guarantee that
\begin{align*}
    \E \dtv{\bbeta^{\geq}}{\widehat{\bbeta}^{\geq}}
    &= O\Big(\vphantom{\sum_{k=1}^K}\frac1K\sqrt{Bd} + 25^{K}B\sum_{k=1}^K B^{-k} \cdot \sqrt{K} \cdot B^k \cdot k^{O(k)} \cdot \epsilon\Big)\\
    &= O\Big(\vphantom{\sum_{k=1}^K}\frac1K\sqrt{Bd} + K^{O(K)}B \epsilon\Big).
\end{align*}
For this to be at most $\epsilon$, the first term must be $O(\epsilon)$,
which forces us to pick $B = O(\epsilon^2K^2/d)$. With this choice, we have
\begin{equation*}
    \E \dtv{\bbeta^{\geq}}{\widehat{\bbeta}^{\geq}}
    =O\Big(\epsilon + \frac{1}{d}K^{O(K)} \epsilon^3\Big).
\end{equation*}
For the second term to be at most $O(\epsilon)$ as well, we select $K = c \cdot \log(d)/\log \log(d)$ for some small enough constant $c > 0$.
In total, this gives an estimate $\widehat{\bbeta}$ which is $O(\epsilon)$ close to the true small bucket spectrum $\bbeta$; combining this with our estimate of the large bucket gives a full algorithm for estimating $\rho$'s spectrum.
In total, this algorithm uses
\begin{equation*}
    n
    = O\Big(\frac{d}{B^2 \epsilon^2}\Big)
    = O\Big(\frac{d^3}{K^4\epsilon^6}\Big)
    = O\Big(d^3 \cdot \parens[\Big]{\frac{\log \log(d)}{\log(d)}}^4 \cdot \frac{1}{ \epsilon^6}\Big)
\end{equation*}
copies of $\rho$, as promised.
In principle, this algorithm can also be made to run with $\poly(d, 1/\eps)$ quantum gate complexity and classical overhead, but for simplicity, we limit our discussion to sample complexity.

This argument incurs a noticeable loss in terms of error, scaling as $1/\eps^6$.
This comes from the bucket threshold $B$ scaling as $\eps^2$, which then inflates the cost of creating the buckets, which is $O(dB^{-2}\eps^{-2})$.
The bucket threshold is identical to that in the classical setting~\cite{HJW18}, but the cost of bucketing is higher in the quantum setting, incurring a dependence on $B$ which is not present in the classical setting.
These complications are more or less due to the alignment error discussed in previous sections.\footnote{
    Our argument also introduces a $\log\log(d)$ dependence which is not present in the classical LMM argument; this overhead may appear for similar reasons.
}
Further, in \Cref{sec:lmm-fails}, we argue that this issue is inherent to the strategy of bucketing.
In total, then, it is not clear what kind of algorithm could achieve the correct dependence on $\eps$.

\subsection{Discussion}
In summary, we show that spectrum estimation can be performed with fewer samples than state tomography in the unentangled setting.
Still open is the question of the true copy complexity of spectrum estimation, both in the unentangled and entangled settings, and even for constant $\eps$.
We now discuss avenues towards resolving this question.

Our algorithm requires $O(d^3 (\log\log(d)/ \log(d))^4)$ copies to perform unentangled spectrum estimation for constant $\eps$, an improvement which is unexpectedly large compared to the mere $\log(d)$ savings in the classical setting.
We lack a clear explanation for why four log factors can be saved, though we expect this scaling to persist for $\eps$ smaller than constant, in a similar parameter regime as in classical sorted distribution estimation.
We give some evidence that a straightforward adaptation of a local moment matching scheme will not suffice: in \Cref{sec:lmm-fails}, we give a family of rank-$r$ quantum states for which learning in trace distance $\epsilon$ reduces to bucketing with $< \eps^2$ alignment error.
Prior work by Haah et al. \cite{HHJ+16} has demonstrated rank-$r$ full state tomography lower bounds against this family of quantum states. This suggests that bucketing into any number of buckets is at least as hard as performing rank-$\frac{1}{B}$ full state tomography, where $B$ is the upper threshold of the \emph{smallest} bucket.
Since in local moment matching, this threshold scales linearly with $\eps$, this approach cannot attain the (presumably correct) quadratic dependence on $1/\eps$.
This barrier holds for both entangled and unentangled settings.
In short, our evidence suggests that a different algorithm is needed to perform spectrum estimation optimally.

As for the entangled setting, in \cref{sec:lower} we give computational evidence that $n = O(d)$ samples does not suffice for spectrum estimation.
In fact, this evidence points to $d^{2 - \gamma}$ being insufficient for any constant $\gamma > 0$.
As for the upper bound, the central barrier to adapting our algorithm to the entangled setting is proving an $\ell_\infty$ guarantee for a sample-optimal fully entangled tomography algorithm.
Overall, we still lack formal proofs beyond the upper bound of $O(d^2)$ and the lower bound of $\Omega(d)$; closing this gap remains an interesting open problem.

\section{Preliminaries}\label{sec:prelim}

We use \textbf{boldface} to denote random variables, and define $[d] = \{1, \dots, d\}$.

\subsection{Classical and quantum distances}

\begin{definition}[Total variation distance]
    The \emph{total variation (TV) distance} between two vectors $x, y \in \R^d$ is defined as
    \begin{equation*}
        \dtv{x}{y} = \frac12 \sum_{i=1}^d |x_i - y_i|. 
    \end{equation*}
\end{definition}

\begin{definition}[Schatten $k$-norm]
    Let $M \in \mathbb{C}^{d \times d}$ be a Hermitian matrix with eigenvalues $\lambda_1, \cdots, \lambda_d$. The \emph{Schatten $k$-norm} is defined as
    \begin{equation*}
        \lVert M \rVert_k = \left(\sum_{i=1}^d |\lambda_i|^k\right)^{1/k}.
    \end{equation*}
    In particular, the Schatten-$\infty$ norm $\norm*{M}_\infty = \max \{|\lambda_1|, \cdots, |\lambda_d|\}$ is also known as the \emph{operator norm}.
\end{definition}

\begin{definition}[Trace distance]
    \label{def:trace-dist}
    The \emph{trace distance} between two density matrices $\rho$ and $\sigma$ is defined as
    \begin{equation*}
        \dtr(\rho,\sigma) = \frac{1}{2}\lVert \rho - \sigma \rVert_1 = \max_{\text{projectors}~\Pi} \left\{\tr\left(\Pi(\rho - \sigma)\right)\right\}.
    \end{equation*}
\end{definition}

\begin{definition}[Fidelity]
    The \emph{fidelity} of the density matrices $\rho$ and $\sigma$ is defined as
    \begin{equation*}
        \fidelity(\rho, \sigma) = \norm{\sqrt{\rho} \sqrt{\sigma}}_1 = \tr\sqrt{\sqrt{\rho} \sigma \sqrt{\rho}}.
    \end{equation*}
\end{definition}

Our version of the fidelity is sometimes referred to as the ``square root fidelity''. In~\Cref{sec:lmm-fails}, we will compute the fidelity and trace distance of sub-normalized density matrices. It is not hard to verify that the definitions above can be extended to any pair of matrices, as long as they are PSD.
Fidelity and trace distance are related by the following inequalities, which can be found in~\cite[Section 9.2]{NC10}.

\begin{lemma}[Fuchs-van de Graaf inequalities]
    \label{lem:trace-vs-fidelity}
    The trace distance and fidelity are related as follows:
    \begin{equation*}
        1 - \fidelity(\rho, \sigma) \leq \dtr(\rho, \sigma) \leq \sqrt{1 - \fidelity(\rho, \sigma)^2}.
    \end{equation*}
\end{lemma}

\subsection{Haar random vectors}

\begin{definition}[The Haar measure]
Let $\mathrm{U}(d)$ be the group of $d \times d$ complex unitary matrices. 
The \emph{Haar measure} on $\mathrm{U}(d)$ is the unique measure with the following property:
if $\bU$ is distributed according to the Haar measure
then for any unitary $V \in \mathrm{U}(d)$,
both $V \cdot \bU$ and $\bU \cdot V$ are distributed according to the Haar measure.
\end{definition}

\begin{definition}[Haar random vectors]
    A \emph{Haar random vector} in $\C^d$ is a vector distributed as $\bU \cdot \ket{1}$, where $\bU$ is a Haar random unitary.
    A \emph{Haar random basis} is a set of orthonormal vectors $\ket{\bu_1}, \ldots, \ket{\bu_d}$
    which are distributed as $\bU \cdot \ket{1}, \ldots, \bU \cdot \ket{d}$.
\end{definition}

\begin{definition}[A representation of the symmetric group]\label{def:rep_sym_group}
    Let $S_n$ be the symmetric group consisting of permutations on $\{1, \ldots, n\}$.
    Given a permutation $\pi \in S_n$, we write
    $P(\pi)$ for the unitary matrix acting on $(\C^d)^{\otimes n}$ acting as follows.
    First, for any $i_1, \ldots, i_n \in [d]$,
    $P(\pi)$ acts on the corresponding basis element by permuting the $n$ registers according to $\pi$:
\begin{equation*}
    P(\pi) \cdot \ket{i_1}\otimes \cdots \otimes\ket{i_n} = \ket{i_{\pi^{-1}(1)}}\otimes \cdots \otimes \ket{i_{\pi^{-1}(n)}}.
\end{equation*}
    We can then define $P(\pi)$ on the whole space $(\C^d)^{\otimes n}$ via linearity.
    As a result, for any $d\times d$ matrices $M_1, \ldots, M_k$, we have that
    \begin{equation*}
        P(\pi^{-1}) \cdot M_{1}\otimes M_2\otimes \cdots \otimes M_k \cdot P(\pi) = M_{\pi(1)} \otimes M_{\pi(2)} \otimes \cdots \otimes M_{\pi(k)}. 
    \end{equation*}
    These matrices form a \emph{representation} of $S_n$, meaning that $P(\pi) \cdot P(\sigma) = P(\pi \cdot \sigma)$ for any $\pi, \sigma \in S_n$.
    When it is clear from context, we will often write $\pi$ in place of $P(\pi)$.
    Finally, in the $n = 2$ case, we will often write $\swap = P((1,2))$.
\end{definition}

We will make use of the following expression appearing in~\cite[Proposition 6]{har13} which expresses the moments of a Haar random vector in terms of the above symmetric group representation.
\begin{proposition}[Moments of a Haar random vector]\label{prop:haar-moment}
Let $\ket{\bu}$ be a Haar random vector in $\C^d$.
Then
\begin{equation*}
    \E_{\bu} \ketbra{\bu}^{\otimes n} = \frac{1}{d(d+1)\cdots (d+n-1)} \cdot \sum_{\pi\in S_n} P(\pi).
\end{equation*}
\end{proposition}

\subsection{The uniform POVM}

If $\ket{\bu} \in \C^d$ is a Haar random vector, then following from~\Cref{prop:haar-moment}, we have
\begin{equation}\label{eq:def-of-M}
    M \coloneqq \E_{\bu}\ketbra{\bu} = \frac{1}{d} \cdot I.
\end{equation}
Alternatively, to see why, 
note that because $\ket{\bu}$ is a Haar random vector, then $U \cdot \ket{\bu}$ is also a Haar random vector, for any unitary matrix $U$. This means that
\begin{equation*}
    M = \E_{\bu}[U \cdot \ketbra{\bu} \cdot U^{\dagger}] = U \cdot M \cdot U^{\dagger}.
\end{equation*}
The only way that $M$ can satisfy this for all unitaries $U$ is if it is a constant multiple of the identity. To compute the scalar, let us simply take the trace of $M$: 
\begin{equation*}
    \tr(M)
    = \tr\Big(\E_{\bu}\ketbra{\bu}\Big)
    = \E_{\bu}\Big[\tr(\ketbra{\bu})\Big]
    = 1.
\end{equation*}
Thus, $M = I/d$, proving \Cref{eq:def-of-M}.
This means that $\E_{\bu}[d\cdot \ketbra{\bu}] = I$,
which we can interpret as giving a decomposition of the identity for a POVM known as the \emph{uniform POVM}.

\begin{definition}[Uniform POVM]
    The \emph{uniform POVM} is the measurement that assigns a uniform probability to all pure state projectors $\ketbra{u}$. Formally, the uniform POVM is
    \begin{equation*}
        \left\{d\cdot \ketbra{u} \cdot \diff u\right\},
    \end{equation*}
    where $\diff u$ is the Haar measure over pure states $\ket{u}\in \C^d$.
\end{definition}

 The uniform POVM is equivalent to the following randomized measurement.
 \begin{enumerate}
     \item Sample a Haar random basis $\ket{\bu_1}, \ldots, \ket{\bu_d}$.
     \item Measure $\rho$ in this basis and let $\ket{\bu_i}$ be the outcome.
     \item Output $\ket{\bu_i}$.
 \end{enumerate}
 Thus, the uniform POVM can be interpreted as measuring $\rho$ in a uniformly random basis, which is perhaps the most natural measurement to perform if one does not have any prior information about $\rho$.

\subsubsection{Moments of the uniform POVM}
 We will need to compute the first and second moments of the outcome vector of the uniform POVM.
 These calculations are standard and we include them for completeness.
 To begin, we will need the following helper lemma.

 \begin{lemma}[Partial trace helper lemma]\label{lem:partial-trace-helper}
     Let $\rho$ be Hermitian.
     Then $\tr_2(\swap \cdot (I \otimes \rho)) = \rho$.
 \end{lemma}
 \begin{proof}
     Let
     \begin{equation*}
        \rho = \sum_{i=1}^d \alpha_i \cdot \ketbra{v_i}
     \end{equation*}
     be the eigendecomposition of $\rho$.
     We can expand the identity in this basis as well, i.e.\ $I = \sum_{i=1}^d \ketbra{v_i}$. Then
     \begin{align*}
         \tr_2(\swap\cdot(I \otimes \rho))
         &= \tr_2\Big(\swap \cdot \Big(\sum_{i=1}^d \ketbra{i} \otimes \sum_{j=1}^d \alpha_j \cdot \ketbra{j}\Big)\Big)\\
         &= \sum_{i, j = 1}^d \alpha_j \cdot \tr_2( \swap \cdot (\ketbra{i}\otimes \ketbra{j}))
         = \sum_{i, j=1}^d \alpha_j \cdot \tr_2(\ketbratwo{j}{i} \otimes \ketbratwo{i}{j}).
     \end{align*}
     Now the partial trace is simple enough that we can calculate it directly:
     \begin{equation*}
        \tr_2(\ketbratwo{j}{i} \otimes \ketbratwo{i}{j})
        = \ketbratwo{j}{i} \cdot \tr(\ketbratwo{i}{j})
        = \left\{\begin{array}{cl}
                \ketbra{i} & \text{if } i = j,\\
                0 & \text{otherwise}.
                \end{array}\right.
     \end{equation*}
     Thus,
     \begin{equation*}
         \tr_2(\swap\cdot(I \otimes \rho))
         = \sum_{i=1}^d \alpha_i \cdot \ketbra{i}
         = \rho. \qedhere
     \end{equation*}
 \end{proof}

 Next, we give a formula for the $k$-th moment of the uniform POVM.

 \begin{lemma}[$k$-th moment formula]\label{lem:moment-helper}
     Let $\rho \in \C^{d \times d}$ be a density matrix.
     Suppose we measure $\rho$ with the uniform POVM
     and receive outcome $\ket{\bu} \in \C^d$.
     Then
     \begin{equation*}
         \E_{\bu} \ketbra{\bu}^{\otimes k}
         = \frac{1}{(d+1) \cdots (d+k)} \cdot \sum_{\pi \in S_{k+1}} \tr_{k+1}(\pi \cdot (I^{\otimes k} \otimes \rho)).
     \end{equation*}
 \end{lemma}
 \begin{proof}
     Measuring $\rho$ with the uniform POVM produces
      $\ket{u} \in \C^d$ with measure $d \cdot \tr(\ketbra{u} \cdot \rho) \cdot \diff u$.
     Thus,
     \begin{align*}
         \E_{\bu} \ketbra{\bu}^{\otimes k}
         &= \int_u \ketbra{u}^{\otimes k} \cdot (d \cdot \tr(\ketbra{u} \cdot \rho) \cdot \diff u)\\
         &= d \cdot \int_u \tr_{k+1}(\ketbra{u}^{\otimes k} \otimes (\ketbra{u} \cdot \rho)) \cdot \diff u\\
         &= d \cdot \int_u \mathrm{tr}_{k+1}(\ketbra{u}^{\otimes k+1} \cdot (I^{\otimes k} \otimes \rho)) \cdot \diff u\\
         &= d \cdot \tr_{k+1}\Big(\Big(\int_u \ketbra{u}^{\otimes k+1} \cdot \diff u\Big) \cdot (I^{\otimes k} \otimes \rho)\Big)\\
         &= d \cdot \tr_{k+1}\Big(\Big(\frac{1}{d(d+1) \cdots (d+k)} \cdot\sum_{\pi \in S_{k+1}} \pi\Big) \cdot (I^{\otimes k} \otimes \rho)\Big)\tag{by \Cref{prop:haar-moment}}\\
         &= \frac{1}{(d+1) \cdots (d+k)} \cdot \sum_{\pi \in S_{k+1}} \tr_{k+1}(\pi \cdot (I^{\otimes k} \otimes \rho)).
    \end{align*}
    This completes the proof.
 \end{proof}

 Now we specialize \Cref{lem:moment-helper}
  to derive explicit expressions for the first and second moments.

 \begin{proposition}[First moment of the uniform POVM]\label{prop:uniform-POVM-first-moment}
     Let $\rho \in \C^{d \times d}$ be a density matrix.
     Suppose we measure $\rho$ with the uniform POVM
     and receive outcome $\ket{\bu} \in \C^d$.
     Then
     \begin{equation*}
         \E_{\bu} \ketbra{\bu} = \Big(\frac{1}{d+1}\Big) \cdot \rho + \Big(\frac{d}{d+1}\Big) \cdot (I/d).
     \end{equation*}
 \end{proposition}
 \begin{proof}
     By \Cref{lem:moment-helper},
     \begin{align*}
         \E_{\bu} \ketbra{\bu}
         &= \Big(\frac{1}{d+1}\Big) \cdot \tr_2(I \otimes \rho)
            + \Big(\frac{1}{d+1}\Big)\cdot \tr_2(\swap \cdot (I \otimes \rho))\\
        &= \Big(\frac{1}{d+1}\Big) \cdot I + \Big(\frac{1}{d+1}\Big) \cdot \rho,
     \end{align*}
     where the second step uses \Cref{lem:partial-trace-helper}.
     The proposition now follows by rewriting $I$ as $d \cdot (I/d)$.
 \end{proof}

  \begin{proposition}[Second moment of the uniform POVM]\label{prop:uniform-POVM-second-moment}
     Let $\rho \in \C^{d \times d}$ be a density matrix.
     Suppose we measure $\rho$ with the uniform POVM
     and receive outcome $\ket{\bu} \in \C^d$.
     Then
     \begin{equation*}
         \E\ketbra{\bu}^{\otimes 2} = \frac{1}{(d+1)(d+2)}
         \cdot (I + \swap) \cdot \Big(I \otimes I + \rho \otimes I + I \otimes \rho\Big).
     \end{equation*}
 \end{proposition}
 \begin{proof}
     By \Cref{lem:moment-helper},
     \begin{equation}\label{eq:plug-in-moment-helper}
         \E_{\bu} \ketbra{\bu}^{\otimes 2}
         = \frac{1}{(d+1)(d+2)} \cdot  \tr_3\Big(\sum_{\pi \in S_3}\pi \cdot (I \otimes I \otimes \rho)\Big).
     \end{equation}
     The permutations in $S_3$ can be written as
     $e, (1,3), (2,3)$
     and $(1,2) \cdot e, (1, 2) \cdot (1,3), (1, 2) \cdot (2,3)$.
     Hence,
     \begin{equation*}
         \sum_{\pi \in S_3}\pi
         = (e + (1, 2)) \cdot (e + (1, 3) + (2, 3)).
     \end{equation*}
     Thus,
     \begin{align*}
         \eqref{eq:plug-in-moment-helper}
         &= \frac{1}{(d+1)(d+2)}
         \cdot \tr_3\Big(\Big((e + (1, 2)) \cdot (e + (1, 3) + (2, 3))\Big) \cdot (I \otimes I \otimes \rho)\Big)\\
         &= \frac{1}{(d+1)(d+2)}
         \cdot (e + (1,2)) \cdot \tr_3((e + (1, 3) + (2, 3)) \cdot (I \otimes I \otimes \rho))\\
         &= \frac{1}{(d+1)(d+2)}
         \cdot (e + (1,2)) \cdot \Big(\tr_3(I \otimes I \otimes \rho) + \tr_3((1,3) \cdot(I \otimes I \otimes \rho))+ \tr_3((2,3) \cdot(I \otimes I \otimes \rho))\Big)\\
         &= \frac{1}{(d+1)(d+2)}
         \cdot (e + (1,2)) \cdot \Big(I \otimes I + \rho \otimes I + I \otimes \rho\Big). \tag{by \Cref{lem:partial-trace-helper}}
     \end{align*}
     In the second equality we used the fact that $(e + (1,2))$ only acts on the first two registers and hence can be pulled out of the $\tr_3(\cdot)$.
     This completes the proof.
 \end{proof}

\subsection{The uniform POVM tomography algorithm}

Suppose we have one copy of a density matrix $\rho \in \C^{d \times d}$ and we want to learn $\rho$.
Since we do not have any prior information about $\rho$,
a natural thing to do is to measure $\rho$ with the uniform POVM.
If $\ket{\bu} \in \C^d$ is the measurement outcome,
we might hope to use $\ketbra{\bu}$ as our estimator for $\rho$.
However, \Cref{prop:uniform-POVM-first-moment} shows that this is not a good idea, even in expectation.
In particular, the expectation
\begin{equation*}
         \E_{\bu} \ketbra{\bu} = \Big(\frac{1}{d+1}\Big) \cdot \rho + \Big(\frac{d}{d+1}\Big) \cdot (I/d)
\end{equation*}
is mostly noise (the second term), but it does have a small amount of signal (the first term).
Correcting for this noise
suggests that a better estimator is $(d+1) \cdot \ketbra{\bu} - I$, and indeed it is an \emph{unbiased estimator} for $\rho$:
\begin{equation*}
    \E[(d+1) \cdot \ketbra{\bu} - I] = \rho.
\end{equation*}
This motivates the following natural uniform POVM tomography algorithm.

\begin{definition}[Uniform POVM tomography algorithm] \label{def:uniform-povm-alg}
    Given $n$ copies of a state $\rho$, the \emph{uniform POVM tomography algorithm} works as follows.
    \begin{enumerate}
        \item Measure each copy of $\rho$ with the uniform POVM $\{d \cdot \ketbra{u} \cdot \diff u\}$.
        \item Set $\brho_i = (d+1)\cdot\ketbra{\bu_i} - I$, where $\ket{\bu_i}$ is the $i$-th measurement outcome.
        \item Output $\widehat{\brho} = \frac{1}{n}\cdot (\brho_1 + \cdots + \brho_n)$.
    \end{enumerate}
\end{definition}

From the above discussion, each $\brho_i$ is an unbiased estimator for $\rho$, i.e.\ $\E \brho_i = \rho$.
Extending this to $\widehat{\brho}$ using linearity expectation, we have the following proposition.

\begin{proposition}[The uniform POVM tomography algorithm gives an unbiased estimator]\label{prop:uniform-povm-is-an-unbiased-estimator}
    Let $\widehat{\brho}$ be the estimator produced by performing the uniform POVM tomography algorithm on $\rho$. Then
    $\E[\widehat{\brho}] = \rho$.
\end{proposition}

The uniform POVM tomography algorithm was introduced independently by Krishnamurthy and Wright~\cite[Section 5.1]{Wri16} and Guta et al.~\cite{GKKT20}.
Both works showed that the $\widehat{\brho}$ produced by the uniform POVM tomography algorithm is $\epsilon$-close to $\rho$ with high probability once $n = O(d^3/\epsilon^2)$.
This is optimal among all algorithms which use unentangled measurements, as \cite{CHL+23} showed that $n = \Omega(d^3/\epsilon^2)$ copies are required to perform trace distance tomography with unentangled measurements.
Krishnamurthy and Wright achieve this by first showing that $\widehat{\brho}$ is close to $\rho$ in $\ell_2$ distance; 
Guta et al.\ instead show that $\widehat{\brho}$ is close to $\rho$ in the stronger $\ell_\infty$ distance,
and they can use this to derive various additional interesting consequences,
such as an $n = O(d r^2/\epsilon^2)$ tomography algorithm in the case when $\rho$ is promised to be rank $r$.
We will need the following operator norm bound from their work.

\begin{theorem}[{\cite[Theorem 5]{GKKT20}}]
\label{thm:opnormconcentration}
    There exists a universal constant $C_1>0$ so that for all $n$, the output of the uniform POVM tomography algorithm satisfies
    \begin{equation*}
        \lVert \widehat{\brho} - \rho \rVert_\infty \leq C_1 \cdot \sqrt{d/n} \qquad \text{with probability }0.99. 
    \end{equation*}
\end{theorem}

\noindent We note that a similar statement appears as Theorem 5.4 in \cite{CHL+23}, except with a slightly weaker bound of $C_1 \cdot \max\{d/n, \sqrt{d/n}\}$ on the right-hand side.

Finally, we will need the following expression for the second moment of the $n = 1$ uniform POVM tomography algorithm.

\begin{proposition}[Second moment of the uniform POVM tomography algorithm]\label{prop:uniform-povm-tomograph-second-moment}
    Let $\rho \in \C^{d \times d}$ be a density matrix.
     Suppose we measure $\rho$ with the uniform POVM
     and receive outcome $\ket{\bu} \in \C^d$.
     Let $\widehat{\brho} = (d+1)\cdot \ketbra{\bu} - I$.
     Then
     \begin{equation*}
         \E[\widehat{\brho} \otimes \widehat{\brho}]= \frac{1}{d+2}
         \cdot ((d+1) \cdot \swap - I) \cdot \Big(I \otimes I + \rho \otimes I + I \otimes \rho\Big).
     \end{equation*}
\end{proposition}
\begin{proof}
    Expanding $\widehat{\brho}$ according to its definition,
    \begin{align*}
        \E[\widehat{\brho} \otimes \widehat{\brho}]
        & = \E[((d+1) \cdot \ketbra{\bu} - I) \otimes ((d+1) \cdot\ketbra{\bu} - I)]\\
        & = (d+1)^2 \cdot \E\ketbra{\bu}^{\otimes 2} - (d+1) \cdot I \otimes \E \ketbra{\bu}
        - (d+1) \cdot \E \ketbra{\bu} \otimes I
        + I \otimes I.
    \end{align*}
    By \Cref{prop:uniform-POVM-second-moment}, the first term is equal to
    \begin{equation*}
        (d+1)^2 \cdot \E\ketbra{\bu}^{\otimes 2}
        = \frac{d+1}{d+2}
         \cdot (I + \swap) \cdot \Big(I \otimes I + \rho \otimes I + I \otimes \rho\Big).
    \end{equation*}
    By \Cref{prop:uniform-POVM-first-moment}, the second and third terms are equal to
    \begin{align*}
        (d+1) \cdot I \otimes \E \ketbra{\bu}
        + (d+1) \cdot \E \ketbra{\bu} \otimes I
        &= I \otimes (I + \rho) + (I + \rho) \otimes I\\
        &= 2 \cdot I \otimes I + \rho \otimes I + I \otimes \rho.
    \end{align*}
    Putting everything together,
    \begin{align*}
        \E[\widehat{\brho} \otimes \widehat{\brho}]
        &= \frac{d+1}{d+2}
         \cdot (I + \swap) \cdot \Big(I \otimes I + \rho \otimes I + I \otimes \rho\Big)
         - (I \otimes I + \rho \otimes I + I \otimes \rho)\\
         &= \frac{d+1}{d+2}
         \cdot \swap \cdot \Big(I \otimes I + \rho \otimes I + I \otimes \rho\Big)
         - \frac{1}{d+2} \cdot (I \otimes I + \rho \otimes I + I \otimes \rho)\\
         &= \frac{1}{d+2}
         \cdot ((d+1) \cdot \swap - I) \cdot \Big(I \otimes I + \rho \otimes I + I \otimes \rho\Big).
    \end{align*}
    This completes the proof.
\end{proof}

\section{Moment estimation}\label{sec:moment_estimation}

Given a $d$-dimensional quantum state $\sigma$, we define a natural estimator $\bZ_k$ for its $k$-th moment $\tr(\sigma^k)$ based on the uniform POVM.

\begin{definition}[Moment estimator]\label{def:moment_estimator}
    Suppose we have $n$ copies of $\sigma$. Let $k \leq n$ be a positive integer. For each $1 \leq i \leq n$, perform the uniform POVM on the $i$-th copy of $\sigma$. Let $\ket{\bu_i}$ be the outcome, and set $\widehat{\bsigma}_i = (d+1)\cdot \ketbra{\bu_i}- I$. The $k$-moment estimator is defined as
    \begin{equation*}
        \bZ_k \coloneqq \frac{1}{n(n-1)\cdots (n-k+1)}\cdot\sum_{\text{distinct }i_1, i_2, \ldots, i_k \in [n]} \tr\left(\widehat{\bsigma}_{i_1} \widehat{\bsigma}_{i_2} \cdots \widehat{\bsigma}_{i_k}\right).
    \end{equation*}
\end{definition}
Since each $\widehat{\bsigma}_i$ is an independent, unbiased estimator for $\sigma$, $\bZ_k$ is an unbiased estimator for $\tr(\sigma^k)$. 
Indeed, it is the natural unbiased estimator for $\tr(\sigma^k)$ suggested by U-statistics.
As mentioned in the introduction, related estimators have appeared in the literature before; for example, it can be viewed as a special case of an estimator for nonlinear functions of $\sigma$ proposed in~\cite{HKP20}.
In addition, a related estimator for $\tr(\rho\sigma)$, where $\rho$ and $\sigma$ are two distinct quantum states, was proposed in~\cite{ALL22};
we will compare the performance of their estimator when $\sigma = \rho$ with our $k = 2$ estimator below.

Our estimator can be viewed as a natural quantum analogue of the classical collision-based moment estimator from \Cref{eq:collision-estimator} above.
One difference between these estimators, however,
is that in the classical estimator it suffices to sum over only those indices $i_1 < \cdots <i_k$ which are arranged in increasing order,
whereas in our quantum estimator we sum over all distinct $i_1, \ldots, i_k$, which need not be arranged in increasing order.
This is because in the classical setting, the indicator function $\mathbbm{1}[\bx_{i_1}=\bx_{i_2}=\cdots=\bx_{i_k}]$ is invariant under permuting its indices,
and so summing over all distinct $i_1, \ldots, i_k$ yields the same estimator as summing over all increasing $i_1 < \cdots <i_k$.
However, in the quantum setting, the estimators $\widehat{\bsigma}_i$ need not commute with each other, and so in general it is the case that
\begin{equation*}
    \tr\left(\widehat{\bsigma}_{i_1} \widehat{\bsigma}_{i_2} \cdots \widehat{\bsigma}_{i_k}\right)
    \neq \tr \left(\widehat{\bsigma}_{i_{\pi(1)}} \widehat{\bsigma}_{i_{\pi(2)}} \cdots \widehat{\bsigma}_{i_{\pi(k)}}\right),
    \qquad \text{for }\pi\in S_k,
\end{equation*}
with the one exception of the $k = 2$ case.
Hence, summing over only those indices in which $i_1 < \cdots <i_k$ would actually yield a different and, we believe, worse estimator.
One additional subtlety arising from the noncommutativity of the $\widehat{\bsigma}_i$'s is that the $\tr\left(\widehat{\bsigma}_{i_1} \widehat{\bsigma}_{i_2} \cdots \widehat{\bsigma}_{i_k}\right)$ terms are, in general, complex-valued.
However, because each term appears in the sum with its complex conjugate 
$\tr\left(\widehat{\bsigma}_{i_k} \widehat{\bsigma}_{i_{k-1}} \cdots \widehat{\bsigma}_{i_1}\right)$,
the overall estimator $\bZ_k$ is still real-valued.

Since $\bZ_k$ is an unbiased estimator for $\tr(\sigma^k)$,
our main goal is to show that it concentrates well around its mean.
To do this, we will bound its variance.
This entails bounding the expression $\E[\bZ_k^2]$, which involves terms like
\begin{equation*}
    \E \tr\left(\widehat{\bsigma}_{i_1} \widehat{\bsigma}_{i_2} \cdots \widehat{\bsigma}_{i_k}\right) \cdot \tr\left(\widehat{\bsigma}_{j_1} \widehat{\bsigma}_{j_2} \cdots \widehat{\bsigma}_{j_k}\right). 
\end{equation*}
When $i_1,\ldots,i_k,j_1,\ldots,j_k$ are all distinct, this term equals $(\tr(\sigma^k))^2 = (\E[\bZ_k])^2$. However, when the sample indices $\{i_1,\ldots,i_k\}$ and $\{j_1,\ldots,j_k\}$ intersect nontrivially, the non-commutativity of $\widehat{\bsigma}_i$ makes it challenging to analyze it directly. 
Nevertheless, we are able to prove the following bound on the variance of our estimator.

\begin{theorem}\label{thm:generalmomentvar}
    For any positive integer $k$ at most $n/2$, the variance of $\bZ_k$ is at most
    \begin{align*}
        \frac{24^k}{d} \sum_{j=0}^{k-1} \left(\frac{k d}{n}\right)^{k-j} \tr(\sigma^{2j}).
    \end{align*}
\end{theorem}

To understand this bound, let us consider an example.
When $k = 2$, $\bZ_k$ is an unbiased estimator for $\tr(\sigma^2)$, the purity of $\sigma$,
and \Cref{thm:generalmomentvar} bounds its variance by
\begin{equation*}
    \frac{24^2}{d} \cdot \Big(\Big(\frac{2d}{n}\Big)^2 \cdot \tr(\sigma^0) + \Big(\frac{2d}{n}\Big)^1 \cdot \tr(\sigma^2)\Big)
    = O\Big(\frac{d^2}{n^2} + \frac{\tr(\sigma^2)}{n}\Big),
\end{equation*}
where we have used the fact that $\tr(\sigma^0) = d$ for any state $\sigma$.
Applying Chebyshev's inequality, this allows us to estimate $\tr(\sigma^2)$ up to error
\begin{equation}\label{eq:k=2-case}
    O\Big(\frac{d}{n} + \sqrt{\frac{\tr(\sigma^2)}{n}}\Big)
\end{equation}
with probability 99\%.
As $\tr(\sigma^2)\leq 1$ for all $\sigma$, we can upper bound this error by $O(d/n + 1/\sqrt{n})$.
This means that $\bZ_2$ is $\epsilon$-close to $\tr(\sigma^2)$ with probability 99\% once $n = O(d/\epsilon + 1/\epsilon^2)$, which scales as $O(d/\epsilon)$ when $\epsilon \geq 1/d$ and as $O(1/\epsilon^2)$ when $\epsilon \leq 1/d$. This gives an \emph{additive error} guarantee, in the sense that it promises that $\bZ_2 = \tr(\sigma^2) \pm \epsilon$.
However, it is useful to keep the $\tr(\sigma^2)$ term in \Cref{eq:k=2-case} around, rather than upper bounding it by $1$, because its presence allows us to also achieve a \emph{multiplicative error} guarantee as well, in the sense that
\begin{equation*}
    (1-\delta) \cdot \tr(\sigma^2) \leq \bZ_2 \leq (1+\delta) \cdot \tr(\sigma^2).
\end{equation*}
Writing $\bZ_2 = \tr(\sigma^2) + \Delta$, this is equivalent to asking that $|\Delta|/\tr(\sigma^2) \leq \delta$.
Applying our bound on $\Delta$ from \Cref{eq:k=2-case}, we have that
\begin{equation*}
    \frac{|\Delta|}{\tr(\sigma^2)}
    \leq O\left(\frac{d}{\tr(\sigma^2) \cdot n} + \sqrt{\frac{1}{\tr(\sigma^2)\cdot n}}\right)
    \leq O\Big(\frac{d^2}{n} + \sqrt{\frac{d}{n}}\Big),
\end{equation*}
where in the last step we have used the fact that $\tr(\sigma^2) \geq 1/d$ always, where equality holds when $\sigma = I/d$ is maximally mixed.
This is at most $\delta$ once $n = O(d^2/\delta + d/\delta^2)$,
and so this algorithm achieves a multiplicative error guarantee given this many copies.
(Note that upper bounding $\tr(\sigma^2) \leq 1$ would have yielded a worse sample complexity of $O(d^2/\delta^2)$.)

As mentioned above, an estimator quite similar to our $\bZ_2$ was studied by Anshu, Landau, and Liu~\cite{ALL22}
for the task of estimating $\tr(\rho\sigma)$, given copies of two quantum states $\rho$ and $\sigma$.
Theirs is also an unbiased estimator, and they prove a variance bound of
$O(d^2/n + 1/n)$~\cite[Equation (180)]{ALL22}.
In fact, their estimator can be shown to have a stronger variance bound of $O(d^2/n+\tr(\rho\sigma)/n)$, matching that of our estimator when $\rho = \sigma$.
As also mentioned above, our estimator $\bZ_2$ can also be viewed as a special case of the estimators for quadratic functions from~\cite{HKP20} (simply set their $O_i = \swap$). However, they do not prove explicit variance or sample complexity bounds for these estimators.

As a corollary of~\Cref{thm:generalmomentvar}, we derive the following multiplicative error bounds for estimating the $k$-th moment. 

\begin{corollary}[Multiplicative-error moment estimator]\label{cor:multi_error_moment_estimator}
    For any quantum state $\sigma$ of dimension $d$ and a fixed positive integer $k\geq 2$, with probability $0.99$, $\bZ_k$ can estimate $\tr(\sigma^k)$ to multiplicative error $\delta$ using
    \begin{equation*}
        n = O\left(\max \left\{\frac{d^{2-2/k}}{\delta^2}, \frac{d^{3 - 2/k}}{\delta^{2/k}} \right\}\right)
    \end{equation*}
    copies of $\sigma$. 
\end{corollary}
\begin{proof}
Since $\E \bZ_k = \tr(\sigma^k)$, using Chebyshev's inequality, we have that for any $\delta > 0$, 
\begin{equation*}
    \Pr \left[ \ABS{\bZ_k - \tr(\sigma^k)} \geq \delta\cdot \tr(\sigma^k)\right] = \Pr \left[ \ABS{\bZ_k - \tr(\sigma^k)} \geq \frac{\delta \cdot\tr(\sigma^k)}{\sqrt{\Var[\bZ_k]}}   \sqrt{\Var[\bZ_k]}\right] \leq \frac{\Var[\bZ_k]}{\delta^2 \cdot (\tr(\sigma^k))^2}. 
\end{equation*}
For any normalized quantum state $\sigma$, let us consider two cases. First, when $2j\geq k$, by the monotonicity of norms,  we have
\begin{equation*}
    \tr(\sigma^{2j})^{1/(2j)}\leq \tr(\sigma^{k})^{1/k}.
\end{equation*}
Since we always have $j\leq k-1$ in the expression for $\Var[\bZ_k]$, this implies that
\begin{equation*}
    \frac{\tr(\sigma^{2j})}{(\tr(\sigma^k))^2}\leq \tr(\sigma^k)^{2j/k - 2} = \left(\frac{1}{\tr(\sigma^k)}\right)^{\frac{2}{k}(k-j)}
    \leq d^{(k-1)\frac{2}{k}(k-j)} = d^{(2-2/k)(k-j)},
\end{equation*}
where we used $\tr(\sigma^k) \geq \tr((I/d)^k)= d^{1-k}$ for the last inequality. 
For the second case, when $2j\leq k$, we have that $g(x)=x^{k/(2j)}$ is a convex function. Moreover, let the random variable $\bX$ take the value of $\alpha_i^{2j}$ with uniform probability $1/d$ for all $i\in[d]$. It then follows from Jensen's inequality $g(\E[\bX]) \leq \E g(\bX)$ that
\begin{equation*}
    \left(\frac{\tr(\sigma^{2j})}{d}\right)^{k/(2j)} \leq \frac{\tr(\sigma^{k})}{d}. 
\end{equation*}
Then
\begin{equation*}
    \frac{\tr(\sigma^{2j})}{(\tr(\sigma^k))^2}\leq (\tr(\sigma^k))^{2j/k-2} \cdot d^{1 - 2j/k} \leq d^{(k-1)\frac{2}{k}(k-j) + 1-2j/k} = d^{2(k-j)-1}.
\end{equation*}
Together with \Cref{thm:generalmomentvar}, we have that
\begin{align*}
    \frac{\Var[\bZ_k]}{(\tr(\sigma^k))^2}&\leq \frac{24^k}{d} \sum_{j=\lceil k/2 \rceil}^{k-1} \left( \frac{kd^{3 - 2/k}}{n} \right)^{k-j} + \frac{24^k}{d^2} \sum_{j=0}^{\lceil k/2 \rceil - 1} \left( \frac{kd^{3}}{n} \right)^{k-j}. 
\end{align*}
Note that for any integers $b\geq a$, 
\begin{align*}
    \sum_{i=a}^b x^i &\leq \begin{cases}
        (b-a)\cdot x^b, & \text{if }x\geq 1, \\
        (b-a)\cdot x^a, & \text{if } x< 1.
    \end{cases}\\
    &\leq \begin{cases}
        (b-a)\cdot x^{b+1/2}, & \text{if }x\geq 1, \\
        (b-a)\cdot x^{a-1/2}, & \text{if } x< 1.
    \end{cases}
\end{align*}
Therefore, when $k$ is constant, we have that
\begin{align*}
    \frac{24^k}{d} \sum_{j=\lceil k/2 \rceil}^{k-1} \left( \frac{kd^{3 - 2/k}}{n} \right)^{k-j} &\leq O\left(\max \left\{\frac{1}{d}\cdot \left( \frac{d^{3-2/k}}{n}\right)^{ \lfloor k/2 \rfloor }, \ \frac{1}{d}\cdot \frac{d^{3-2/k}}{n}\right\}\right) \\
    &\leq O\left(\max \left\{\frac{1}{d}\cdot \left( \frac{d^{3-2/k}}{n}\right)^{ k/2 }, \ \frac{1}{d}\cdot \frac{d^{3-2/k}}{n}\right\}\right)
\end{align*}
and
\begin{align*}
    \frac{24^k}{d^2} \sum_{j=0}^{\lceil k/2 \rceil - 1} \left( \frac{kd^{3}}{n} \right)^{k-j} &\leq O\left(\max \left\{
    \frac{1}{d^2}\cdot \left(\frac{d^3}{n}\right)^k, \ \frac{1}{d^2}\cdot \left(\frac{d^3}{n}\right)^{\lfloor k/2 \rfloor + 1}
    \right\}\right) \\
    &\leq O\left(\max \left\{
    \frac{1}{d^2}\cdot \left(\frac{d^3}{n}\right)^k, \ \frac{1}{d^2}\cdot \left(\frac{d^3}{n}\right)^{k/2}
    \right\}\right).
\end{align*}
In the second inequality, we used the fact that the second term only dominates when $d^3/n \leq 1$, and reducing its power gives an upper bound in that case.
This means that with probability $99\%$, 
$\bZ_k$ can estimate the $k$-th moment of $\sigma$ to multiplicative error $\delta$ provided that 
\begin{equation*}
    \max \left\{\frac{1}{d}\cdot \left( \frac{d^{3-2/k}}{n}\right)^{ k/2 }, \ \frac{1}{d}\cdot \frac{d^{3-2/k}}{n}, \ \frac{1}{d^2}\cdot \left(\frac{d^3}{n}\right)^k, \ \frac{1}{d^2}\cdot \left(\frac{d^3}{n}\right)^{ k/2 } \right\} \leq O(\delta^2),
\end{equation*}
which can be simplified as
\begin{equation*}
    n = O\left(\max \left\{\frac{d^{2-2/k}}{\delta^2}, \frac{d^{3 - 4/k}}{\delta^{4/k}}, \frac{d^{3 - 2/k}}{\delta^{2/k}} \right\}\right) = O\left(\max \left\{\frac{d^{2-2/k}}{\delta^2}, \frac{d^{3 - 2/k}}{\delta^{2/k}} \right\}\right). 
\end{equation*}
Note that here we drop the second term $\frac{d^{3 - 4/k}}{\delta^{4/k}}$ because it never dominates: $\frac{d^{3 - 4/k}}{\delta^{4/k}} \geq \frac{d^{3 - 2/k}}{\delta^{2/k}}$ if and only if $\delta \leq 1/d$. But when $\delta \leq 1/d$, we always have $\frac{d^{3 - 4/k}}{\delta^{4/k}} \leq \frac{d^{2-2/k}}{\delta^2}$ when $k\geq 2$. 
\end{proof}

\subsection{Application: quantum R\'{e}nyi entropy}

\begin{definition}[Quantum R\'{e}nyi entropy]
    Let $k$ be a positive real number. The \emph{quantum R\'{e}nyi entropy of order $k$} of a density matrix $\sigma$ is defined as
    \begin{equation*}
        S_{k}(\sigma) = \frac{1}{1-k} \log \tr(\sigma^k).
    \end{equation*}
\end{definition}

Due to the relationship between moment estimation and quantum R\'{e}nyi entropy, our multiplicative-error moment estimator from~\Cref{cor:multi_error_moment_estimator} can be used to obtain an additive-error approximation to the quantum R\'{e}nyi entropy for fixed integers $k$.

\begin{corollary}[Additive-error R\'{e}nyi entropy estimator]
    \label{cor:add_error_renyi_entropy_estimator}
    For any quantum state $\sigma$ of dimension $d$ and a fixed positive integer $k$, with probability $0.99$, the quantity $\frac{1}{1-k}\log \bZ_k$ can estimate $S_{k}(\sigma)$ up to an additive error $\delta < 1/2$ using
    \begin{equation*}
        n = O\left(\max \left\{\frac{d^{2-2/k}}{\delta^2}, \frac{d^{3 - 2/k}}{\delta^{2/k}} \right\}\right)
    \end{equation*}
    copies of $\sigma$.
\end{corollary}

\begin{proof}
    From~\Cref{cor:multi_error_moment_estimator}, we know that with probability $0.99$,
    \begin{equation*}
        \tr(\sigma^k)(1 - \delta) \leq \bZ_k \leq  \tr(\sigma^k)(1 + \delta).
    \end{equation*}
    Taking the logarithm on all sides gives that
    \begin{equation*}
        \log \tr(\sigma^k) + \log(1 - \delta) \leq \log\bZ_k \leq  \log \tr(\sigma^k) + \log(1 + \delta).
    \end{equation*}
    We rewrite the term $\log(1-\delta) = -\log \frac{1}{1-\delta} = -\log\left(1 + \frac{\delta}{1-\delta}\right)$ and use the well-known inequality $\log(1 + x) \leq x$ for all positive $x$ to deduce that
    \begin{equation*}
        \log \tr(\sigma^k) - \frac{\delta}{1-\delta} \leq \log\bZ_k \leq  \log \tr(\sigma^k) + \delta.
    \end{equation*}
    Finally, we multiply all sides by $\frac{1}{1-k}$ to conclude that
    \begin{equation*}
        \ABS{\frac{1}{1-k}\log \bZ_k - S_k(\sigma)} \leq \frac{\delta}{(1-\delta)(k-1)}.
    \end{equation*}
    Since $k$ is a fixed integer and $\delta < 1/2$, the quantity $\frac{1}{1-k}\log \bZ_k$ estimates $S_k(\sigma)$ up to an additive error $O(\delta)$. The statement of the corollary follows by adjusting the number of copies $n$ by a constant.
\end{proof}

\subsection{Helper lemmas: the trace of permutations}\label{subsec:helper_lemmas_trace_perm}

Before diving into the proof of \Cref{thm:generalmomentvar} in \Cref{sec:variance_proof}, let us take a detour to prove some helper lemmas related to the trace of permutations. 
One of the key ingredients in proving \Cref{thm:generalmomentvar} is \Cref{lem:cycles_swaps}, whose proof is purely combinatorial. 

The following lemma is immediate using the tensor network diagram notation. 
\begin{lemma}\label{lem:num_disjoint_cycles}
Let $k$ be a positive integer. 
For any $k$-cycle $\tau\in S_k$ and $d\times d$ matrices $M_1,\ldots, M_k$,
\begin{align*}
    \tr\left(P(\tau^{-1}) \cdot M_{1}\otimes M_2\otimes \cdots \otimes M_k \right) = \tr(M_{\tau(1)}M_{\tau^2(1)}\cdots M_{\tau^{k}(1)}).
\end{align*}
\end{lemma}

We will make frequent use of the special case when $M_1=\cdots=M_k = I$. 
For any permutation $\pi\in S_k$, let $c(\pi)$ denote the number of disjoint cycles in the cycle decomposition of $\pi$. For example, $c(\pi)=1$ if $\pi$ is a $k$-cycle. 
By~\Cref{lem:num_disjoint_cycles}, we have that $\tr(P(\pi)) = d^{c(\pi)}$.

\begin{lemma}\label{lem:perm_PSD}
    For any PSD operator $0\preceq \sigma\preceq I$ and permutation $\pi\in S_k$, we have that
    \begin{enumerate}
        \item\label{item:always-real} $\tr\left(P(\pi) \cdot \sigma^{a_1}\otimes \cdots \otimes \sigma^{a_k}\right)$ is always real and nonnegative for any $a_i\in\R$. 

        \item\label{item:increasing} $\tr\left(P(\pi) \cdot \sigma^{a_1}\otimes \cdots \otimes \sigma^{a_k}\right) \leq \tr\left(P(\pi) \cdot \sigma^{b_1}\otimes \cdots \otimes \sigma^{b_k}\right)$ whenever $a_i \geq b_i\geq 0$ for all $i\in [k]$.
    \end{enumerate}
\end{lemma}
\begin{proof}
    Suppose $\tau \in S_m$ is a $m$-cycle. By \Cref{lem:num_disjoint_cycles}, for any $a_1, \cdots, a_m$, we have that $\tr(P(\tau) \cdot \sigma^{a_1}\otimes \cdots \otimes \sigma^{a_m}) = \tr(\sigma^{a_1+\cdots + a_m})$ is always real and nonnegative for any $a_i\in \R$. 
    Moreover, since every eigenvalue of $\sigma$ is in $[0,1]$, we have $\tr(\sigma^a)\leq \tr(\sigma^b)$ for any $a\geq b\geq 0$. 
    The lemma then follows from the fact that any permutation $\pi$ can be decomposed into a product of disjoint cycles. 
\end{proof}

\begin{lemma}\label{lem:cycles_swaps}
    Let $\tau$ be a $k$-cycle on the odd numbers in $[2k]$ and $\tau'$ be a $k$-cycle on the even numbers in $[2k]$. For any integer $j\in \{0,1,2,\cdots, k-1\}$ and a nonempty set of indices $T \subseteq [k-j]$, define the following permutation
    \begin{equation*}
        \mu \coloneqq (\tau\cdot \tau')^{-1} \cdot \parens[\Bigg]{\prod_{i\in T} \swap_{2i-1,2i}} \in S_{2k}. 
    \end{equation*}
    Then
    \begin{equation}\label{eq:cycles_swaps_sigma_not_full}
        \tr\left( P(\mu)
        \cdot I^{\otimes 2(k-j)}\otimes \sigma^{\otimes 2j} \right) \leq d^{k-j-1}\cdot \tr(\sigma^{2j}).
    \end{equation}
\end{lemma}

\begin{proof}
This follows from some observations about the cycle decomposition of $\mu$.
For a simple example, consider when $j = 0$; then, the statement reduces to proving that
\begin{equation}
    \tr \parens{P(\mu)} \leq d^k.
\end{equation}
We know by \Cref{lem:num_disjoint_cycles} that $\tr(P(\mu)) = d^{c(\mu)}$ where $c(\mu)$ is the number of cycles in $\mu$.
We know that $\mu$ is the composition of three permutations in $S_{2k}$: first, swap adjacent elements $(2i-1, 2i)$ for some subset of $i$'s between $1$ and $k-j$; then, permute the odd elements in a $k$-cycle $\tau^{-1}$; finally, permute the even elements in a $k$-cycle $(\tau')^{-1}$.
Then $c(\mu) \leq k$, since cycle lengths are at least two (even indices get mapped to odd indices, and vice versa).

For the general case, by \Cref{lem:perm_PSD} we have that
\begin{align*}
    \tr\left( P(\mu)
    \cdot I^{\otimes 2(k-j)}\otimes \sigma^{\otimes 2j} \right)
    = \tr(\sigma^{p_1} \cdot I^{q_1}) \cdots \tr(\sigma^{p_c} \cdot I^{q_c})
\end{align*}
where $p_s, q_s$ are non-negative integers such that $\sum_{s \in [c]} p_s = 2j$ and $\sum_{s \in [c]} q_s = 2(k-j)$.
Specifically, these numbers come from the cycle decomposition of $\mu$.
There is a trace for each cycle, $c = c(\mu)$, and for the $s$-th cycle, $p_s$ and $q_s$ are the number of elements in the cycle which correspond to $\sigma$'s and $I$'s, respectively.
Concretely, $p_s$ is the number of elements in the $s$-th cycle which are at least $2(k-j)+1$, and $q_s$ is the number of elements in the $s$-th cycle which are at most $2(k-j)$.

We will show that every $q_s \geq 2$.
This suffices to show the lemma, since then the number of cycles is bounded by half the number of identities, $c \leq k-j$, and so
\begin{align*}
    \tr \left(P_\pi \cdot I^{\otimes 2(k-j)} \otimes \sigma^{\otimes 2j}\right)
    &= \tr(\sigma^{p_1}) \cdots \tr(\sigma^{p_c}) \\
    &\leq d^{c-1}\cdot \tr(\sigma^{p_1 + \dots + p_c}) \leq d^{k-j-1}\cdot \tr(\sigma^{2j}).
\end{align*}
The first inequality follows from repeatedly using Chebyshev's sum inequality, which states that if $x_1\geq x_2\geq \cdots \geq x_d$ and $y_1\geq y_2\geq \cdots \geq y_d$, then $(\frac{1}{d}\sum_{i=1}^d x_i)(\frac{1}{d}\sum_{i=1}^d y_i)\leq \frac{1}{d}\sum_{i=1}^d x_iy_i$.
In other words, we have that $\tr(\sigma^{p_1}) \tr(\sigma^{p_2}) \leq d \tr(\sigma^{p_1 + p_2})$, and so on.
The second inequality uses that $c \leq k-j$.

It remains to show that every $q_s \geq 2$, meaning that every cycle in the cycle decomposition of $\mu$ contains two elements which are at most $2(k-j)$.
Recall that $\mu$ is the composition of three permutations: first, swap adjacent elements $(2i-1, 2i)$ for some subset of $i$'s between $1$ and $k-j$; then, permute the odd elements in a $k$-cycle $\tau^{-1}$; finally, permute the even elements in a $k$-cycle $(\tau')^{-1}$.
Consider a cycle in $\mu$'s cycle decomposition, $(i, \mu(i), \mu^2(i), \dots)$: we will show that every such cycle alternates parity at least twice.
If this is true, then $q_s \geq 2$, since $q_s$ is at least the number of parity changes in the cycle.
If $\mu(i)$ has different parity from $i$, then $i \leq 2(k-j)$, since the only way parity changes is through the SWAPs, which only operate on indices which are at most $2(k-j)$.

To see why the cycle alternates parity at least twice, consider an odd $i$ (the even case is identical).
If $\mu(i)$ is also odd, then $\mu(i) = \tau^{-1}(i)$.
So, if the cycle never changes parity, then the cycle consists of all $k$ odd elements.
But this cannot happen: there is at least one SWAP, so somewhere in the cycle, $\tau^{-1}$ will eventually take $i$ to this SWAP, and alternate parity.
Then, parity must flip twice, to get back from even to odd when looping back to the beginning of the cycle.
This completes the proof.
\end{proof}

\subsection{Proof of \texorpdfstring{\Cref{thm:generalmomentvar}}{Theorem 5.2}}\label{sec:variance_proof}
Since $\bZ_k\in\R$, we have $\Var[\bZ_k] = \E[\bZ_k^2] - (\E[\bZ_k])^2$.
To begin, recall that we defined $\bZ_k$ in the following way:
\begin{equation*}
    \bZ_k \coloneqq \frac{1}{n(n-1)\cdots (n-k+1)}\cdot\sum_{\text{distinct }i_1, i_2, \ldots, i_k \in [n]} \tr\left(\widehat{\bsigma}_{i_1} \widehat{\bsigma}_{i_2} \cdots \widehat{\bsigma}_{i_k}\right).
\end{equation*}
So, we can rewrite this as an expectation,
\begin{equation*}
    \bZ_k = \E_{\bj} \E_{\bpi \sim S_k} \tr\left(\widehat{\bsigma}_{\bj_{\bpi(1)}} \widehat{\bsigma}_{\bj_{\bpi(2)}} \cdots \widehat{\bsigma}_{\bj_{\bpi(k)}}\right), 
\end{equation*}
where $\bpi$ is a uniformly random permutation in $S_k$ and $\bj = \{\bj_1, \ldots, \bj_k\}$ is a uniformly random subset of $[n]$ of size $k$. We refer to $\bj$ as the sample indices. 
Then, 
\begin{align} \label{eq:k-moment-variance}
    \E[\bZ_k^2] 
    &= \E_{\bi,\bj} \ \E_{\bpi,\bpi' \sim S_k}  \ \E_{\widehat{\bsigma}}\left[ \tr\left(\widehat{\bsigma}_{\bi_{\bpi(1)}} \widehat{\bsigma}_{\bi_{\bpi(2)}} \cdots \widehat{\bsigma}_{\bi_{\bpi(k)}}\right) \cdot \tr\left(\widehat{\bsigma}_{\bj_{\bpi'(1)}} \widehat{\bsigma}_{\bj_{\bpi'(2)}} \cdots \widehat{\bsigma}_{\bj_{\bpi'(k)}}\right) \right].
\end{align}
We can write $\E[\bZ_k^2]$ in another way, as the expectation of the output of the following procedure:
\begin{enumerate}
    \item Sample a $\bt \in [k]$ where $t$ is sampled with probability $\binom{k}{t}\binom{n-k}{k-t} / \binom{n}{k}$, i.e.\ the probability that two random $k$-element subsets $\bi, \bj \sim [n]$ have $\abs{\bi \cap \bj} = t$.
    \item Sample disjoint subsets $\bh, \ba, \bb \sim [n]$ of size $\bt$, $k-\bt$, and $k-\bt$ respectively.
    In this way, $\bh \cup \ba$ and $\bh \cup \bb$ are uniformly random $k$-element subsets, conditioned on their intersection being $\bt$.
    Denote $\bi = (\bh_1,\dots,\bh_{\bt},\ba_1,\dots,\ba_{k-\bt})$ to be the first subset, ordered so that $\bh$ comes first, and similarly for $\bj$.
    \item Sample random permutations $\bpi, \bpi' \sim S_k$.
    \item Output $\tr(\widehat{\bsigma}_{\bi_{\bpi(1)}} \dots \widehat{\bsigma}_{\bi_{\bpi(k)}})\cdot \tr(\widehat{\bsigma}_{\bj_{\bpi'(1)}} \dots \widehat{\bsigma}_{\bj_{\bpi'(k)}})$.
\end{enumerate}
This produces the identical expectation, because steps 1 and 2 above produce an $\bi$ and a $\bj$ which are independent and uniformly sampled, as in \Cref{eq:k-moment-variance}.
They are ordered such that their intersection comes first, but this does not matter because $\bpi$ and $\bpi'$ fully randomize their ordering in the subsequent trace.
We can write this mathematically:
\begin{align} \label{eq:k-moment-variance-alt}
    \E[\bZ_k^2]
    &= \E_{\bt} \ \E_{\bh, \ba, \bb \mid \bt} \ \E_{\bpi, \bpi'} \ \E_{\widehat{\bsigma}}\left[ \tr(\widehat{\bsigma}_{\bi_{\bpi(1)}} \dots \widehat{\bsigma}_{\bi_{\bpi(k)}})\cdot \tr(\widehat{\bsigma}_{\bj_{\bpi'(1)}} \dots \widehat{\bsigma}_{\bj_{\bpi'(k)}}) \right].
\end{align}
Let $\tau_0\in S_k$ denote the $k$-cycle that maps $1\to 2\to 3\to \cdots \to k\to 1$. 
By~\Cref{lem:num_disjoint_cycles}, for any permutation $\pi\in S_k$ and $i = (i_1,\dots,i_k) \in [n]^k$,
\begin{align*}
    \tr(\widehat{\bsigma}_{i_{\pi(1)}} \widehat{\bsigma}_{i_{\pi(2)}} \dots \widehat{\bsigma}_{i_{\pi(k)}})
    &= \tr\left( P(\tau_0^{-1}) \cdot \widehat{\bsigma}_{i_{\pi(1)}} \otimes \widehat{\bsigma}_{i_{\pi(2)}} \otimes \dots \otimes \widehat{\bsigma}_{i_{\pi(k)}} \right) \\
    &= \tr\left( P(\tau_0^{-1}) \cdot P(\pi^{-1}) \cdot \widehat{\bsigma}_{i_1} \otimes \widehat{\bsigma}_{i_{2}} \otimes \dots \otimes \widehat{\bsigma}_{i_{k}} \cdot P(\pi) \right) \\
    &= \tr\left( P(\pi\tau_0^{-1}\pi^{-1}) \cdot \widehat{\bsigma}_{i_1} \otimes \widehat{\bsigma}_{i_{2}} \otimes \dots \otimes \widehat{\bsigma}_{i_{k}} \right)
\end{align*}
where the last equality follows from the cyclic property of the trace and that $P$ is a representation of $S_k$. 
Hence, taking one output of the above procedure and looking at its expectation over the POVM outcomes,
\begin{align}
    & \E_{\widehat{\bsigma}}\left[ \tr(\widehat{\bsigma}_{\bi_{\bpi(1)}} \dots \widehat{\bsigma}_{\bi_{\bpi(k)}})\cdot \tr(\widehat{\bsigma}_{\bj_{\bpi'(1)}} \dots \widehat{\bsigma}_{\bj_{\bpi'(k)}}) \right] \nonumber \\
    = \ & \E_{\widehat{\bsigma}}\left[ \tr \parens[\Big]{ P(\bpi\tau_0^{-1}\bpi^{-1}) \cdot \widehat{\bsigma}_{\bi_1} \otimes \widehat{\bsigma}_{\bi_{2}} \otimes \dots \otimes \widehat{\bsigma}_{\bi_{k}} } \cdot \tr \parens[\Big]{ P(\bpi'\tau_0^{-1}(\bpi')^{-1}) \cdot \widehat{\bsigma}_{\bj_1} \otimes \widehat{\bsigma}_{\bj_{2}} \otimes \dots \otimes \widehat{\bsigma}_{\bj_{k}} } \right] \nonumber \\
    = \ & \E_{\widehat{\bsigma}}\left[ \tr \parens[\Big]{ P(\bpi\tau_0^{-1}\bpi^{-1}) \otimes P(\bpi'\tau_0^{-1}(\bpi')^{-1}) \cdot (\widehat{\bsigma}_{\bi_1} \otimes \dots \otimes \widehat{\bsigma}_{\bi_{k}}) \otimes (\widehat{\bsigma}_{\bj_1} \otimes \dots \otimes \widehat{\bsigma}_{\bj_{k}} )} \right] \nonumber\\
    = \ & \tr \parens[\Big]{ P(\bpi\tau_0^{-1}\bpi^{-1}) \otimes P(\bpi'\tau_0^{-1}(\bpi')^{-1}) \cdot \E_{\widehat{\bsigma}}\left[ (\widehat{\bsigma}_{\bi_1} \otimes \dots \otimes \widehat{\bsigma}_{\bi_{k}}) \otimes (\widehat{\bsigma}_{\bj_1} \otimes \dots \otimes \widehat{\bsigma}_{\bj_{k}} )\right]}
    \label{eq:fix_sample_indices_new}
\end{align}
where the second equality follows from $\tr(A)\cdot \tr(B) = \tr(A\otimes B)$ and the last equality follows from the linearity of expectation and trace. 

Since all the $k$-cycles form a conjugacy class in $S_k$ and $\tau_0$ is a fixed $k$-cycle, if $\bpi$ is a uniformly random permutation in $S_k$, then $\bpi \tau_0 \bpi^{-1}$ is a uniformly random $k$-cycle in $S_k$. 
Therefore, 
$\btau \coloneqq \bpi\tau_0^{-1}\bpi^{-1}$ and $\btau' \coloneqq \bpi'\tau_0^{-1}(\bpi')^{-1}$ are two $k$-cycles sampled independently and uniformly at random from $S_k$.

Now, for notational convenience, within the trace we will reorder the tensor products from
\begin{equation*}
    (\widehat{\bsigma}_{\bi_1} \otimes \dots \otimes \widehat{\bsigma}_{\bi_{k}}) \otimes (\widehat{\bsigma}_{\bj_1} \otimes \dots \otimes \widehat{\bsigma}_{\bj_{k}} )
\end{equation*}
to interleave the $\bi$'s and $\bj$'s, as in
\begin{equation*}
    (\widehat{\bsigma}_{\bi_1} \otimes \widehat{\bsigma}_{\bj_1}) \otimes \dots \otimes (\widehat{\bsigma}_{\bi_{k}} \otimes \widehat{\bsigma}_{\bj_{k}}).
\end{equation*}
Recall that $\bi$ and $\bj$ agree on the first $\bt$ indices: they are both equal to $\bh$.
So, using\footnote{
    For this proof, we will not use that $\tr(\sigma) = 1$; we will only use that $\tr(\sigma) \leq 1$.
    We will later call this proof to show that a slight variant of this moment estimator, which is used in our spectrum estimation algorithm, also succeeds for estimating sub-normalized states.
    Discussion of this variant appears in~\Cref{sec:subnormalized-moment-estimation}.
    Its variance analysis proceeds identically to this one except for this normalization, so we present it here in the slightly more general setting.
}
\begin{align*}
    \E[\widehat{\bsigma}_i \otimes \widehat{\bsigma}_j] =
    \begin{cases}
        \frac{(d+1)\swap - I^{\otimes 2}}{d+2} \left(I\otimes \sigma + \sigma \otimes I + \tr(\sigma) \cdot I\otimes I\right), & \text{if }i=j\\
        \sigma\otimes \sigma, &\text{if }i\neq j
    \end{cases}
\end{align*}
we have that
\begin{equation} \label{eq:fix_sample_indices_final}
    \eqref{eq:fix_sample_indices_new} = \underbrace{\tr\left( P(\btau\cdot \btau') \cdot \left(\E\left[ \widehat{\bsigma} \otimes \widehat{\bsigma} \right]\right)^{\otimes \bt} \otimes \left(\sigma \otimes \sigma \right)^{\otimes (k-\bt)} \right)}_{\coloneqq Q(k-\bt, \btau, \btau')}, 
\end{equation}
where, because of the re-ordering, $\btau \in S_{2k}$ (respectively, $\btau' \in S_{2k}$) is a uniformly random $k$-cycle permuting the odd (respectively, even) integers in $[2k]$. 
Notice that now, $Q(k-\bt, \btau, \btau')$ only depends on $\bt$, the number of overlapping samples between $\bi$ and $\bj$, but not the specific values of $\bi$ or $\bj$. 
Therefore, we can write
\begin{align*}
    \E[\bZ_k^2]
    &= \E_{\bt} \ \E_{\bh, \ba, \bb \mid \bt} \ \E_{\bpi, \bpi'} \ \E_{\widehat{\bsigma}}\left[ \tr(\widehat{\bsigma}_{\bi_{\bpi(1)}} \dots \widehat{\bsigma}_{\bi_{\bpi(k)}})\cdot \tr(\widehat{\bsigma}_{\bj_{\bpi'(1)}} \dots \widehat{\bsigma}_{\bj_{\bpi'(k)}}) \right] \\
    &= \E_{\bt} \ \E_{\bh, \ba, \bb \mid \bt} \ \E_{\btau, \btau'}  Q(k-\bt, \btau, \btau')  \\
    &= \E_{\bt} \ \E_{\btau, \btau'} Q(k-\bt, \btau, \btau')  \\
    &= \frac{1}{\binom{n}{k}} \sum_{i=0}^k \binom{k}{k-i}\binom{n-k}{i} \cdot \E_{\btau,\btau'} Q(i, \btau, \btau'). 
\end{align*}
where we first used \Cref{eq:k-moment-variance-alt}; then \Cref{eq:fix_sample_indices_final}; then that $\btau, \btau'$, and $Q$ do not depend on $\bh$, $\ba$, and $\bb$; and finally, that we can expand the expectation over $\bt$, with $i$ denoting the number of elements not in the intersection (the complement of $\bt$).
When $i = k$, since $\tau,\tau'$ are two disjoint $k$-cycles, we have that
\begin{equation*}
    Q(k,\tau,\tau')= \tr\left( P(\tau\cdot\tau')\cdot \sigma^{\otimes 2k}\right) = (\tr(\sigma^k))^2.
\end{equation*}
Therefore, 
\begin{equation}
    \Var[\bZ_k] = \E[\bZ_k^2] - (\E[\bZ_k])^2 \leq \frac{1}{\binom{n}{k}} \sum_{i=0}^{k-1} \binom{k}{k-i} \binom{n-k}{i}\cdot \E_{\btau, \btau'} Q(i, \btau, \btau') \label{eq:estimator_squared_new} .
\end{equation}
The expression $Q(i, \tau, \tau')$ naturally scales with $\tr(\sigma)^{k+i}$, so subsequently we will work with $Q(i, \tau, \tau') / \tr(\sigma)^{k+i}$, and let $\widetilde{\sigma} \coloneqq \sigma / \tr(\sigma)$ denote $\sigma$ normalized to have unit trace.
For any $k$-cycles $\tau,\tau'\in S_k$ and $i= 0,1,\ldots, k-1$,
\begin{align*}
    & Q(i, \tau, \tau') / \tr(\sigma)^{k+i} \\
    =\ & \frac{1}{\tr(\sigma)^{k+i}}\cdot \tr\left( P(\tau\cdot \tau') \left(\frac{(d+1)\swap - I^{\otimes 2}}{d+2} \cdot \left(I\otimes \sigma + \sigma \otimes I + \tr(\sigma) \cdot I\otimes I\right)\right)^{\otimes (k-i)} \otimes \sigma^{\otimes 2i} \right)\\
    =\ & \tr\left( P(\tau\cdot \tau') \left(\frac{(d+1)\swap - I^{\otimes 2}}{d+2} \cdot \left(I\otimes \widetilde{\sigma} + \widetilde{\sigma} \otimes I + I\otimes I\right)\right)^{\otimes (k-i)} \otimes \widetilde{\sigma}^{\otimes 2i} \right)\\
    = \ & \tr\Bigg( P(\tau\cdot \tau') \cdot \underbrace{\left(\frac{(d+1)\swap - I^{\otimes 2}}{d+2}\right)^{\otimes (k-i)}}_{\text{a sum of permutations}}\otimes I^{\otimes 2i} \cdot \left[\left(I\otimes \widetilde{\sigma} + \widetilde{\sigma} \otimes I + I \otimes I\right)^{\otimes (k-i)} \otimes \widetilde{\sigma}^{\otimes 2i} \right] \Bigg). 
\end{align*}
Let us expand
$\big(\frac{(d+1)\swap - I\otimes I}{d+2}\big)^{\otimes (k-i)}$ into a sum of permutations with positive and negative coefficients. 
Each of the positive coefficients is a product of $\frac{d+1}{d+2}$ and $\frac{1}{d+2}$. 
We would like to apply \Cref{item:always-real} in~\Cref{lem:perm_PSD}: each term with a positive coefficient can be upper bounded by replacing $\frac{d+1}{d+2}$ with $1$ and $\frac{1}{d+2}$ with $\frac{1}{d}$, and each negative coefficient can simply be replaced with any positive number. 
As a result, 
\begin{equation*}
    \frac{Q(i, \tau, \tau')}{\tr(\sigma)^{k+i}}
    \leq \tr\left( P(\tau\cdot \tau')\cdot \left(\swap + \frac{I\otimes I}{d}\right)^{\otimes (k-i)}\otimes I^{\otimes 2i} \cdot \left(I\otimes \widetilde{\sigma} + \widetilde{\sigma}\otimes I + I\otimes I\right)^{\otimes (k-i)} \otimes \widetilde{\sigma}^{\otimes 2i} \right). 
\end{equation*}
Since now $\left(\swap + \frac{I\otimes I}{d}\right)^{\otimes (k-i)}$ is a sum of permutations with only positive coefficients, we can apply \Cref{item:increasing} in~\Cref{lem:perm_PSD} and get 
\begin{equation*}
    \frac{Q(i, \tau, \tau')}{\tr(\sigma)^{k+i}}
    \leq \tr\left( P(\tau\cdot \tau') \cdot \left(\swap + \frac{I\otimes I}{d}\right)^{\otimes (k-i)}\otimes I^{\otimes 2i} \cdot \left(3I\otimes I\right)^{\otimes (k-i)} \otimes \widetilde{\sigma}^{\otimes 2i} \right). 
\end{equation*}
Now let us apply~\Cref{lem:cycles_swaps} and separate out the term without any SWAPs, i.e. $(\frac{I\otimes I}{d})^{\otimes (k-i)}$, and get 
\begin{align*}
    \frac{Q(i, \tau, \tau')}{\tr(\sigma)^{k+i}}
    &\leq 3^{k-i}\cdot \left( (2^{k-i} - 1)\cdot d^{k-i-1}\cdot \tr(\widetilde{\sigma}^{2i}) + \frac{1}{d^{k-i}} \cdot \tr\left( P(\tau\cdot \tau') \cdot I^{\otimes 2(k-i)} \otimes \widetilde{\sigma}^{\otimes 2i} \right) \right) \\
    &= 3^{k-i}\cdot \left( (2^{k-i} - 1)\cdot d^{k-i-1}\cdot \tr(\widetilde{\sigma}^{2i}) + \frac{1}{d^{k-i}} \cdot  (\tr(\widetilde{\sigma}^i))^2 \right). 
\end{align*}
Since $(\tr(\widetilde{\sigma}^{i}))^2 \leq d\cdot \tr(\widetilde{\sigma}^{2i})$ and $k-i-1\geq 0$, we finally have
\begin{equation}\label{eq:individual_term_in_var_bound}
    Q(i, \tau, \tau') \leq 3^{k-i}\cdot 2^{k-i} \cdot d^{k-i-1}\cdot \tr(\sigma)^{k+i} \cdot \tr(\widetilde{\sigma}^{2i}) \leq 6^k \cdot d^{k-i-1} \cdot \tr(\sigma^{2i}).
\end{equation}
Plugging this into \Cref{eq:estimator_squared_new}, we have that
\begin{equation}\label{eq:for-use-in-intro?}
    \Var[\bZ_k] \leq \frac{1}{\binom{n}{k}} \sum_{i=0}^{k-1} \binom{k}{k-i} \binom{n-k}{i}\cdot 6^k \cdot d^{k-i-1} \cdot \tr(\sigma^{2i}).
\end{equation}
Now we bound the coefficients in this expression.
Since $\binom{k}{i}\leq 2^k$ and $k! / (i!)\leq k^{k-i}$, we have
\begin{equation}\label{eq:use-this-in-intro}
    \frac{1}{\binom{n}{k}}\cdot \binom{k}{k-i} \binom{n-k}{i} \leq \frac{k!}{(n-k)^k} \cdot \binom{k}{i} \frac{(n-k)^i}{i!}\leq \left(\frac{k}{n-k}\right)^{k-i} \cdot 2^k. 
\end{equation}
Since $k\leq n/2$, 
we can further simplify the variance of $\bZ_k$ as
\begin{equation*}
\eqref{eq:for-use-in-intro?} \leq \frac{12^k}{d} \sum_{i=0}^{k-1}  \left(\frac{dk}{n-k}\right)^{k-i} \tr(\sigma^{2i}) \leq \frac{24^k}{d} \sum_{i=0}^{k-1} \left(\frac{k d}{n}\right)^{k-i} \tr(\sigma^{2i}). 
\end{equation*}
This completes the proof of~\Cref{thm:generalmomentvar}.

\subsection{Moment estimation on a sub-normalized state}
\label{sec:subnormalized-moment-estimation}

For our spectrum learning algorithm, we use a slight variant of the moment estimation procedure detailed in~\Cref{def:moment_estimator}.
Let $\{\Pi, \overline{\Pi}\}$ be a projective measurement which approximately splits the spectrum of $\rho$ into the large and small buckets.
Throughout this section, we will write $\sigma \coloneqq \overline{\Pi} \rho \overline{\Pi}$ for the ``small bucket'' part of $\rho$.
We would like to estimate the moments of the small eigenvalues within $\sigma$.
To do so, we will estimate $\sigma$ using a natural variant of the uniform POVM tomography algorithm (\Cref{def:uniform-povm-alg}) which first conditions on the $\overline{\Pi}$ outcome.

\begin{definition}[Conditioned uniform POVM]
Given a projective measurement $\{\Pi, \overline{\Pi}\}$,
we define the \emph{conditioned uniform POVM} via the following algorithm,
which acts on a mixed state $\rho$.
\begin{enumerate}
    \item Perform the projective measurement $\{\Pi, \overline{\Pi}\}$ on $\rho$.
    \item If the $\Pi$ outcome is observed, output $\bot$.
    \item If the $\overline{\Pi}$ outcome is observed, the state collapses to $\sigma/\tr(\sigma)$.
    Perform the uniform POVM on the collapsed state and receive the outcome $\ket{\bu} \in \C^d$. Output $\ket{\bu}$.
\end{enumerate}
The output probabilities of this measurement can be described as follows.
First, $\Pi$ is observed with probability $\tr(\Pi \cdot \rho)$ and $\overline{\Pi}$ is observed with probability $\tr(\overline{\Pi} \cdot \rho) = \tr(\sigma)$.
Next, conditioned on observing $\overline{\Pi}$, a fixed unit vector $\ket{u} \in \C^d$ is observed with measure
\begin{equation*}
d \cdot \bra{u} \frac{\sigma}{\mathrm{tr}(\sigma)} \ket{u} \diff u,
\end{equation*}
where $\diff u$ is the Haar measure on unit vectors.
Hence, this measurement has the following output distribution:
\begin{itemize}
    \item[$\circ$] Output $\bot$ with probability $\tr(\Pi \cdot \rho)$, and;
    \item[$\circ$] Output a unit vector $\ket{u} \in \C^d$ with measure $d \cdot \bra{u} \sigma \ket{u} \diff u$.
\end{itemize}
We refer to this distribution over vectors in $\C^d$ (and $\bot$) as $A(\Pi, \rho)$.
\end{definition}

We note that this measurement (and the estimator of $\rho$ that we will define based on it) is essentially the same as the ``projected estimator defined on the subspace $\overline{\Pi}$'' from \cite[Definition 5.6]{CHL+23}, except that in step 3 of our algorithm, we are performing the uniform POVM on the whole space $\C^d$ rather than just the subspace $\overline{\Pi}$.
This is merely out of simplicity; since we are typically in the regime where $\Pi$ has rank $d\cdot (\log\log d)^2/\log^2(d) \ll d$ (at least when $\epsilon$ is a small constant), $\overline{\Pi}$ will consist of almost the entire space, and so there should be little difference between performing the uniform POVM on $\overline{\Pi}$ or on all of $\C^d$.

The conditioned uniform POVM
yields a natural unbiased estimator for $\overline{\Pi} \rho \overline{\Pi}$.

\begin{proposition}[An unbiased estimator from the conditioned uniform POVM] \label{prop:conditioned-uniform-POVM}
    Given $\rho$ and $\{\Pi, \overline{\Pi}\}$,
    suppose we measure $\rho$ with the conditioned uniform POVM and let $\ket{\bu} \sim A(\Pi, \rho)$ be the outcome.
    Set
    \begin{equation*}
        \widehat{\bsigma}
        =\left\{\begin{array}{cl}
        (d+1) \cdot \ketbra{\bu} - I & \text{if } \ket{\bu} \neq \bot,\\
        0
        & \text{if } \ket{\bu} = \bot.
        \end{array}\right.
    \end{equation*}
    Then $\widehat{\bsigma}$ is an unbiased estimator for $\sigma$, i.e.
    $\E[\widehat{\bsigma}] = \sigma$.
    Further, we can compute its second moment:
    \begin{equation*}
        \E[\widehat{\bsigma} \otimes \widehat{\bsigma}] = \frac{1}{d+2} 
         \cdot ((d+1) \cdot \swap - I) \cdot \Big(\tr(\sigma) \cdot I \otimes I + \sigma \otimes I + I \otimes \sigma\Big).
    \end{equation*}
\end{proposition}
\begin{proof}
    We have
    \begin{equation*}
        \E[\widehat{\bsigma}]
        = \Pr[\Pi] \cdot \E[\widehat{\bsigma} \mid \text{outcome $\Pi$}] + \Pr[\overline{\Pi}] \cdot \E[\widehat{\bsigma} \mid \text{outcome $\overline{\Pi}$}].
    \end{equation*}
    The first expectation is 0 because $\ket{\bu} = \bot$ given outcome $\Pi$.
    Given outcome $\overline{\Pi}$,
    $\widehat{\bsigma}$ is an unbiased estimator for $\sigma/\tr(\sigma)$ by \Cref{prop:uniform-povm-is-an-unbiased-estimator}.
    Since $\Pr[\overline{\Pi}] = \tr(\sigma)$, this completes the proof.

    The second moment follows similarly: it is equal to $\Pr[\overline{\Pi}] \cdot \E[\widehat{\bsigma} \otimes \widehat{\bsigma} \mid \text{outcome $\overline{\Pi}$}]$.
    The probability is $\tr(\sigma)$, and the expectation is the second moment of the uniform POVM tomography algorithm applied to $\sigma / \tr(\sigma)$.
    The statement follows from~\Cref{prop:uniform-povm-tomograph-second-moment}.
\end{proof}

This motivates the following natural estimator for the $k$-th moment $\tr(\sigma^k)$ of $\sigma$, which is identical to \Cref{def:moment_estimator} except the uniform POVM is replaced by the conditioned uniform POVM.

\begin{definition}[Conditioned moment estimator]
    \label{def:conditioned-moment-estimator}
    Let $\{\Pi, \overline{\Pi}\}$ be a projective measurement.
    Suppose we have $n$ copies of $\rho$.
    For each $1 \leq i \leq n$, perform the conditioned uniform POVM on $\rho$, and let $\widehat{\bsigma}_i$ be the corresponding unbiased estimator of $\sigma$, as in \cref{prop:conditioned-uniform-POVM}.
    The \emph{conditioned $k$-th moment estimator} is defined as
    \begin{equation*}
        \bY_k \coloneqq \frac{1}{n(n-1)\cdots (n-k+1)}\cdot\sum_{\text{distinct }i_1, i_2, \ldots, i_k \in [n]} \tr\left(\widehat{\bsigma}_{i_1} \widehat{\bsigma}_{i_2} \cdots \widehat{\bsigma}_{i_k}\right).
    \end{equation*}
\end{definition}

As in the case for the normal $k$-th moment estimator, because each $\widehat{\bsigma}_i$ is an independent, unbiased estimator for $\sigma$, $\bY_k$ is an unbiased estimator for $\tr(\sigma^k)$.
We have the following bound on the variance of the conditioned moment estimator.

\begin{proposition}\label{prop:conditionedmomentvar}
    For any positive integer $k$ at most $n/2$, the variance of $\bY_k$ is at most
    \begin{align*}
        \frac{24^k}{d} \sum_{j=0}^{k-1} \left(\frac{k d}{n}\right)^{k-j} \tr(\sigma^{2j}).
    \end{align*}
\end{proposition}

The proof of this proposition proceeds identically to the proof of~\Cref{thm:generalmomentvar} given in~\Cref{sec:variance_proof}.
All that is used in this proof is the first and second moments of the $\widehat{\bsigma}$'s, which by \Cref{prop:conditioned-uniform-POVM}, are identical to the un-conditioned uniform POVM estimator (up to normalization, which the proof handles).

\section{The bucketing algorithm}\label{sec:bucket}

The goal of a bucketing algorithm is to find a projector $\bPi$ such that the two-outcome measurement $\{\bPi, \overline{\bPi}\}$, where $\overline{\bPi} = I - \bPi$, will split the spectrum of $\rho$ into a bucket of large eigenvalues and a bucket of small eigenvalues without incurring much disturbance to the original spectrum of $\rho$.
We suggest the following bucketing algorithm based on the uniform POVM.

\begin{definition}[Uniform POVM bucketing algorithm]
    Given a threshold $0 \leq B \leq 1$ and $n$ copies of $\rho$, the \emph{uniform POVM bucketing algorithm} acts as follows.
    \begin{enumerate}
        \item Run the uniform POVM tomography algorithm on $\rho^{\otimes n}$ to produce an estimator $\widehat{\brho}$ of $\rho$.
        \item Set $\bPi$ to be the projector onto the eigenvectors of $\widehat{\brho}$ with eigenvalues at least $B$.
        \item Output the estimator $\widehat{\brho}$ and the projective measurement $\{\bPi, \overline{\bPi}\}$.
    \end{enumerate}
\end{definition}

For convenience, we have chosen to have the uniform POVM bucketing algorithm additionally output the estimate $\widehat{\brho}$ as it turns out that this will already allow us to estimate the large eigenvalues of $\rho$,
saving us the step of separately estimating them later.
The following theorem describes the performance of the uniform POVM bucketing algorithm.

\begin{theorem}[Performance of the uniform POVM bucketing algorithm]\label{thm:bucket-algo}
Given a threshold $0 \leq B \leq 1$,
suppose we perform the uniform POVM bucketing algorithm on $n$ copies of $\rho$
and receive outputs $\widehat{\brho}$ and $\{\bPi, \overline{\bPi}\}$.
Let $\br$ be the rank of $\bPi$.
Then, when 
$n = C_2 dB^{-2}\eps^{-2}$ for a universal constant $C_2>0$,
with probability $0.99$, the following hold simultaneously:
\begin{enumerate} 
    \item(Learning the large eigenvalues):\label{item:split-large-pca} the large eigenvalues of $\rho$ can be estimated to $\eps$ error TV distance, i.e.,
    \begin{equation}\label{eq:split-large-pca}
        \dtv{\spec(\widehat{\brho})_{\leq \br}}{\spec(\rho)_{\leq \br}} \leq \eps
    \end{equation}
    and $\br\leq 3/(2B)$. 

    \item(Low misclassification error):\label{item:split-small-misclassification} the small eigenvalues of $\rho$ are classified into the small bucket, i.e.,
    \begin{equation}\label{eq:split-small-misclassification}
        \norm{\overline{\bPi}\rho\overline{\bPi} }_\infty \leq (1+\eps)B.
    \end{equation}

    \item(Low alignment error):\label{item:split-full-disturbance} the full spectrum of $\rho$ is disturbed by at most $\eps$ in TV distance, i.e.,
    \begin{equation}\label{eq:split-full-disturbance}
        \dtv{\spec(\bPi \rho \bPi + \overline{\bPi} \rho \overline{\bPi} )}{\spec(\rho)} \leq \epsilon.
    \end{equation}
\end{enumerate}
Here, we use $\spec(\cdot)$ to denote the eigenvalues of a matrix sorted from largest to smallest and $\spec(\cdot)_{\leq \br}$ to denote the $\br$ largest eigenvalues in sorted order.
\end{theorem}

Let us interpret this theorem. 
Our goal is to bucket $\rho$ into 
the large bucket, with eigenvalues $\geq B$,
and the small bucket, with eigenvalues $< B$.
Since $\rho$ is a density matrix,
it can have at most $s = 1/B$ eigenvalues which are $\geq B$,
so what we would like to do is perform rank-$s$ PCA on $\rho$ to discover the best rank-$s$ approximation to $\rho$.
Indeed it is known that the uniform POVM tomography algorithm can give rank-$s$ PCA-style guarantees with $n = O(ds^2 \epsilon^{-2}) = O(dB^{-2} \epsilon^{-2})$ copies~\cite[Theorem 4]{GKKT20},
and we show that this many copies is also sufficient for it to perform bucketing well.
\Cref{item:split-large-pca} implies that the bucketing algorithm naturally achieves a PCA-style result,
in that it learns the eigenvalues of the largest rank-$\br$ part of the state.
\Cref{item:split-small-misclassification,item:split-full-disturbance} show that it has small misclassification error and alignment error, respectively;
note that the misclassification error guarantee in \Cref{item:split-small-misclassification} is only stated for the small bucket $\overline{\bPi} \rho \overline{\bPi}$, which is because via \Cref{item:split-large-pca} we have already learned the eigenvalues on the large bucket $\bPi \rho \bPi$.

Note that for our purposes with bucketing in~\Cref{sec:full_algo}, it suffices to relax \Cref{item:split-small-misclassification} to $\norm{\overline{\bPi}\rho\overline{\bPi} }_\infty \leq 2B$. However, bucketing must only incur a small disturbance to the original spectrum of $\rho$, so \Cref{item:split-full-disturbance} is the main bottleneck here.

To prove \Cref{thm:bucket-algo}, we will need the following two well-known facts about matrices.
Both of these use the notation $\lambda_i(\cdot)$, which refers to the $i$-th largest eigenvalue of a matrix. 
\begin{theorem}[Weyl's inequality]
    For any $d\times d$ Hermitian matrices $A$ and $B$ and $i\in [d]$,
    \begin{equation*}
        |\lambda_i(A+B) - \lambda_i(A)| \leq \norm{B}_\infty. 
    \end{equation*}
\end{theorem}

\begin{theorem}[Cauchy's interlacing theorem]\label{thm:cauchy_interlace}
    For any $d\times d$ Hermitian matrix $A$ and projection matrix $\Pi$ of rank $r$, 
    \begin{equation*}
        \lambda_i(A) \geq \lambda_i(\Pi A\Pi) \geq \lambda_{d-r+i}(A), \quad \text{for all } i \in \{1, \cdots, r\}.  
    \end{equation*}
\end{theorem}

\begin{proof}[Proof of~\cref{thm:bucket-algo}]
    Using~\cref{thm:opnormconcentration} with $n = (3C_1)^2B^{-2}\eps^{-2}d$, 
    we have that with probability $0.99$,
    \begin{equation}\label{eq:splitopnormbound}
        \lVert \widehat{\brho} - \rho \rVert_\infty \leq C_1 \cdot B\epsilon/(3C_1) = B\epsilon/3.
    \end{equation}
    We will use this bound throughout the proof. 

    It is tempting to believe that the rank of $\bPi$ should satisfy $\br \leq 1/B$ because $\rho$ is a density matrix so it can only have at most $1/B$ eigenvalues which are $B$ or greater.
    However, $\bPi$ is defined as the projector onto the eigenvalues of $\widehat{\brho}$, not $\rho$, which are larger than $B$, and $\widehat{\brho}$ is not even necessarily a density matrix (in particular, it is not necessarily PSD).
    That said, we can still show that the rank satisfies the weaker bound $\br \leq 3/(2B)$,
    and this turns out to be sufficient for our purposes.
    To see this, let us use~\Cref{eq:splitopnormbound} and apply Weyl's inequality with $A = \bPi \rho \bPi$ and $B = \bPi (\widehat{\brho} - \rho) \bPi$: 
    \begin{equation*}
        \abs{\lambda_{\br}(\bPi \widehat{\brho} \bPi) - \lambda_{\br}(\bPi \rho \bPi)} \leq \norm*{\bPi (\widehat{\brho} - \rho) \bPi}_\infty \leq \norm*{\widehat{\brho} - \rho}_\infty \leq B\eps/3.
    \end{equation*}
    Since $\lambda_{\br}(\bPi \widehat{\brho} \bPi) \geq B$ by the definition of $\bPi$, we have
    \begin{equation*}
        \lambda_{\br}(\bPi \rho \bPi) \geq \lambda_{\br}(\bPi \widehat{\brho} \bPi) - B\eps/3 \geq (1 - \eps/3)B
        \geq 2B/3,
    \end{equation*}
    where we used $\epsilon \leq 1$ in the last step.
    But $\rho$ is a density matrix, and so it can only have at most $3/(2B)$ eigenvalues which are at least $2B/3$.
    Thus, we have $\br \leq 3/(2B)$.

    We are now ready to prove \Cref{item:split-small-misclassification}. 
    Note that the definition of $\overline{\bPi}$ directly implies that $\lVert \overline{\bPi} \cdot \widehat{\brho}\cdot\overline{\bPi}\rVert_\infty \leq B$. 
    Then, using the triangle inequality and~\Cref{eq:splitopnormbound}: 
    \begin{align*}
        \lVert \overline{\bPi}  \rho\overline{\bPi}\rVert_\infty
        &\leq\lVert \overline{\bPi}\cdot (\rho - \widehat{\brho})\cdot\overline{\bPi}\rVert_\infty + \lVert \overline{\bPi}\cdot \widehat{\brho}\cdot\overline{\bPi}\rVert_\infty \leq\lVert \rho - \widehat{\brho} \rVert_\infty + B \leq (1+\epsilon/3)B. 
    \end{align*}
    Next, applying Weyl's inequality with $A = \rho$ and $B = \widehat{\brho} - \rho$, we see that
    \begin{equation*}
        |\lambda_i(\widehat{\brho}) - \lambda_i(\rho)| \leq \Vert \widehat{\brho} - \rho \Vert_{\infty}
        \leq B\epsilon/3.
    \end{equation*}
    Summing this over all $1 \leq i \leq \br$ and using the fact that $\br\leq 3/(2B)$, we have
    \begin{equation}\label{eq:gonna-use-soon}
        \dtv{\spec(\widehat{\brho})_{\leq \br}}{ \spec(\rho)_{\leq \br}} \leq \tfrac{1}{2} \br\cdot B\epsilon/3 \leq \epsilon/4. 
    \end{equation} 
    This proves~\Cref{item:split-large-pca}. 
    A similar argument shows that
    \begin{equation}\label{eq:also-gonna-use-soon}
        \dtv{\spec(\bPi\rho\bPi)_{\leq \br}}{\spec(\bPi\cdot\widehat{\brho}\cdot\bPi)_{\leq \br}} \leq \tfrac{1}{2}\br\cdot  \left\lVert \bPi\rho\bPi - \bPi\cdot\widehat{\brho}\cdot\bPi\right\rVert_\infty
        \leq \tfrac{1}{2}\br\cdot \left\lVert \widehat{\brho} - \rho\right\rVert_\infty \leq \epsilon/4. 
    \end{equation}
    By the definition of $\bPi$, we know that $\spec(\bPi\cdot\widehat{\brho}\cdot\bPi)_{\leq \br} = \spec(\widehat{\brho})_{\leq \br}$. Then by the triangle inequality, we have
    \begin{align*}
        &\dtv{\spec(\bPi\rho\bPi)_{\leq \br}}{\spec(\rho)_{\leq \br}}\\
        \leq {}&
        \dtv{\spec(\bPi\rho\bPi)_{\leq \br}}{\spec(\widehat{\brho})_{\leq \br}} + \dtv{\spec(\widehat{\brho})_{\leq \br}}{\spec(\rho)_{\leq \br}}\\
        = {}&
        \dtv{\spec(\bPi\rho\bPi)_{\leq \br}}{\spec(\bPi\cdot\widehat{\brho}\cdot\bPi)_{\leq \br}} + \dtv{\spec(\widehat{\brho})_{\leq \br}}{\spec(\rho)_{\leq \br}}\\
        \leq{}&
        \epsilon/2.  \tag{by \Cref{eq:gonna-use-soon,eq:also-gonna-use-soon}}
    \end{align*}
    Finally, we shall prove~\Cref{item:split-full-disturbance}. 
    Let $\{\alpha_i\}_{i\in [d]}$ be the eigenvalues of $\rho$, and let $\{\bbeta_i\}_{i\in [\br]}$ be the eigenvalues of $\bPi\rho\bPi$.
    By Cauchy's interlacing theorem, we have $\alpha_i \geq \bbeta_i$ for $i\in \{1, \cdots, \br\}$. 
    Therefore,
    \begin{equation*}
        \dtv{\spec(\bPi\rho\bPi)_{\leq \br}}{\spec(\rho)_{\leq \br}}
        = \frac12 \sum_{i=1}^{\br} |\alpha_i - \bbeta_i| 
        = \frac12 \sum_{i=1}^{\br} (\alpha_i - \bbeta_i),
    \end{equation*}
    and we have shown above that this is at most $\epsilon/2$.
    Next, let $\{\bbeta_{i}\}_{i={\br + 1}}^d$ be the eigenvalues of $\overline{\bPi}\rho\overline{\bPi}$. Again by Cauchy's interlacing theorem, we have $\alpha_i \leq \bbeta_i$ for $i \in \{\br+1, \cdots, d\}$. 
    Note that it is not necessarily true that $\bbeta_{\br} \geq \bbeta_{\br+1}$, and so $\bbeta_1, \ldots, \bbeta_d$ are not necessarily in sorted order. But because the TV distance between two vectors is minimized when they are sorted~\cite[Proposition 2.2]{OW15}, we have 
    \begin{align*}
        \dtv{\spec(\bPi \rho \bPi + \overline{\bPi} \rho \overline{\bPi} )}{\spec(\rho)}
        &\leq \dtv{\spec(\bPi \rho \bPi)_{\leq \br}}{\spec(\rho)_{\leq \br}} + \dtv{\spec(\overline{\bPi} \rho \overline{\bPi})_{\leq d-\br}}{\spec(\rho)_{> \br}}\\
        &= \frac12 \left(\sum_{i=1}^{\br} |\alpha_i - \bbeta_i| + \sum_{i=\br+1}^d |\alpha_i - \bbeta_i|\right)\\
        &= \frac12 \left(\sum_{i=1}^{\br} (\alpha_i - \bbeta_i) + \sum_{i=\br+1}^d (\bbeta_i - \alpha_i)\right) = \sum_{i=1}^{\br} (\alpha_i - \bbeta_i) \leq \eps,
    \end{align*}
    where we use $\spec(\rho)_{>\br}$ to denote the $\br+1, \ldots, d$-th eigenvalues of $\rho$, sorted in descending order.
    In the last equality we used the fact that $\sum_{i=1}^d \alpha_i = \sum_{i=1}^d \bbeta_i = 1$, so that
    \begin{equation*}
        \sum_{i=\br+1}^d \bbeta_i - \sum_{i=\br+1}^d \alpha_i
        = \Big(1 - \sum_{i=1}^{\br} \bbeta_i\Big) - \Big(1 - \sum_{i=1}^{\br} \alpha_i\Big)
        = \sum_{i=1}^{\br} \alpha_i - \sum_{i=1}^{\br} \bbeta_i.
    \end{equation*}
    This completes the proof.
\end{proof}

\section{Local moment matching}\label{sec:local_moment_matching}

The main theorem of this section is the following.
\begin{theorem}[Error of local moment matching]\label{thm:smallestbucketLMM}
Let $\alpha = (\alpha_1, \cdots, \alpha_d)$ be a sorted vector such that $B\geq \alpha_1\geq \cdots\geq \alpha_d\geq 0$ and $\sum_{i=1}^d \alpha_i \leq 1$. 
Fix some $K \in \N$.
Suppose that for each $k \in [K]$ we have an estimate $\widehat{p}_k$ for $p_k(\alpha) = \sum_{i=1}^d \alpha_i^k$ with error $V_k$, i.e.\
\begin{equation*}
    \ABS{\widehat{p}_k - p_k(\alpha) } \leq V_k.
\end{equation*}
Then there is a randomized algorithm which produces a sorted estimate $\widehat{\balpha}$ of $\alpha$ such that
\begin{equation}\label{eq:lmm-error}
    \E_{\widehat{\balpha}} \dtv{\widehat{\balpha}}{\alpha} \leq O\Big( \frac1K\sqrt{Bd} + 2^{9K/2}B\sum_{k=1}^K B^{-k}V_k\Big).
\end{equation}
(In fact, although we will not use this, the first term can be replaced by the stronger $\sqrt{Bd(p_1(\alpha) + V_1)}/K$.)
\end{theorem}

We refer to the first term in \Cref{eq:lmm-error} as the \emph{bias} and the second term as the \emph{variance}.
The bias term results from the fact that we are only using the first $K$ moments of $\alpha$, and it decreases as the number of moments $K$ grows.
The variance term results from the fact that we only have approximations to the moments, and it increases exponentially as $K$ grows.
This exponential growth means that we will typically only be able to approximate the first $K$ moments, where $K$ is at most logarithmic in the dimension~$d$.

\Cref{thm:smallestbucketLMM} essentially corresponds to the local moment matching algorithm from~\cite{HJW18} for the smallest bucket,
except that in their case $B$, $K$, and $V_k$ were taken to be some fixed values in terms of $n$ and $d$ specific to their task, rather than being treated as variables for more general purposes.
We note that the smallest bucket is handled separately from the remaining buckets in~\cite{HJW18}, and has a simpler analysis.

\subsection{The randomized algorithm}

The randomized algorithm in \Cref{thm:smallestbucketLMM}
uses a classic approach of solving a linear programming relaxation and rounding.
Using linear programming to solve for sorted distributions dates back to a work of Efron and Thisted from 1976~\cite{ET76} and was also used in the works of Valiant and Valiant~\cite{VV11a,VV13} (see also the works of \cite{KV17,TKV17}).

\paragraph{The linear program relaxation.}

Given the sorted vector
$\alpha$ that we want to estimate,
let $\mu_\alpha$ be the discrete measure that places weight one on each $\alpha_i$, i.e. for a set $S \subseteq \R$,
\begin{equation*}
    \mu_\alpha(S) \coloneqq \sum_{i = 1}^d\mathbbm{1}[\alpha_i \in S].
\end{equation*}
This measure satisfies the following two properties:
\begin{equation*}
    \mu_{\alpha}([0, B]) = \int_{0}^B 1\cdot \mu_{\alpha}(\diff x) = d,
    \qquad\text{and}\qquad
    \int_0^B x^k \cdot \mu_{\alpha}(\diff x) = p_k(\alpha),
\end{equation*}
where we know that $|\widehat{p}_k - p_k(\alpha)|\leq V_k$.
Therefore, we will consider the following feasibility linear program:
find a measure $\widehat{\mu}$ on $[0,B]$ which satisfies
    \begin{align*}
        \widehat{\mu}([0,B])&=d, \\
        \abs[\Big]{\widehat{p}_k - \int_0^B  x^k \cdot \widehat{\mu}(\diff x)} &\leq V_k, \quad \text{for all }k\in [K].
    \end{align*}
This linear program is feasible because $\widehat{\mu} = \mu_{\alpha}$ is feasible, and so we can solve it to find some feasible solution $\widehat{\mu}$.
This is a semi-infinite linear program---intuitively, we can treat the values $\widehat{\mu}(x)$ for all $x \in [0, B]$ as the variables of this linear program.
This can be solved to any desired accuracy by discretizing the domain $[0, B]$~\cite{GL98}, and we omit these details for simplicity.

\paragraph{Rounding the linear program solution.}

Let $\widehat{\mu}$ be a solution to the linear program,
which we would now like to round to a sorted vector $\widehat{\balpha}$.
Han, Jiao, and Weissman~\cite{HJW18} proposed a rounding algorithm that does so with the following stability guarantee: 
if $\widehat{\mu}$ cannot be distinguished from the true measure $\mu_\alpha$ via any $1$-Lipschitz function, then $\alpha$ and the returned vector $\widehat{\balpha}$ are also close in total variation distance, at least in expectation. 
For any function $f:\Omega\to\R$, its Lipschitz constant is given by $\|f\|_{\mathrm{Lip}} \coloneq \sup_{x\neq y} \frac{|f(x)-f(y)|}{|x-y|}$. 

\begin{lemma}\label{lem:randomdiscretization}
    There exists a randomized algorithm that takes as input a measure $\widehat{\mu}$ over $\R$ and outputs a $d$-dimensional sorted vector $\widehat{\balpha}$ such that for any $d$-dimensional sorted vector $\alpha$, 
    \begin{equation}\label{eq:rounding-error}
        \E_{\widehat{\balpha}} \dtv{\widehat{\balpha}}{\alpha} = \frac12 \sup_{f : \|f\|_{\mathrm{Lip}} \leq 1} \int_{\R} f(x)(\mu_\alpha(\diff x) - \widehat{\mu}(\diff x)). 
    \end{equation}
\end{lemma}
\begin{proof}
    The algorithm is given as Definition 8 in~\cite{HJW18}: morally, it samples and outputs $d$ points drawn from $\widehat{\mu}$, but there is an additional caveat to handle ordering.
    The claim then follows by combining Lemmas 7, 9, and 10 in~\cite{HJW18}. 
\end{proof}

In the next section, we show that the error in \Cref{eq:rounding-error} is small when $\widehat{\mu}$ is a feasible solution to the linear program,
completing the proof of \Cref{thm:smallestbucketLMM}.

\subsection{Polynomial approximation and moment matching}

The proof of~\cref{thm:smallestbucketLMM} uses two standard facts about polynomials. 
The first is Jackson's inequality, which gives an upper bound on the quality of approximation to a Lipschitz function.
\begin{lemma}[{\cite[Lemma~22]{HJW18}}] \label{lem:jackson}
    For $f: [a, b] \to \R$ a 1-Lipschitz function, the best polynomial approximation $P$ of degree $K$, i.e.\ $P = \arg\min_Q \max_{x \in [a,b]} \abs{Q(x) - f(x)}$, satisfies
    \begin{equation*}
        \ABS{f(x) - P(x)} \leq \frac{C_3\sqrt{(b-a)(x-a)}}{K} \qquad \text{for all } x \in [a,b],
    \end{equation*}
    for a universal constant $C_3>0$.
\end{lemma}
The second is a bound on the coefficients of a bounded polynomial.
\begin{lemma}[{\cite[Lemma~27]{HJW18}}] \label{lem:coefs}
    Let $P(x) = \sum_{k=0}^K a_k x^k$ be a polynomial of degree at most $K$ such that $\ABS{P(x)} \leq A$ for $x \in [a,b]$.
    Then if $a + b \neq 0$, for any $k=0,1,\cdots, K$,
    \begin{equation*}
        \ABS{a_k} \leq 2^{7K/2}\cdot A\cdot \ABS{\frac{a+b}{2}}^{-k}\left(\ABS{\frac{b+a}{b-a}}^K + 1\right).
    \end{equation*}
\end{lemma}

\begin{proof}[Proof of~\cref{thm:smallestbucketLMM}]
For simplicity, we will write the true measure $\mu_\alpha$ as $\mu$. 
Let $\widehat{\mu}$ be any feasible solution to the linear program, meaning that it satisfies
\begin{equation} \label{eqn:equal-mass}
    \int_0^B 1 \cdot \widehat{\mu}(\diff x) = \int_0^B 1 \cdot \mu(\diff x) = d
\end{equation}
and
\begin{equation*}
    \ABS{\widehat{p}_k - \int_0^B  x^k \cdot \widehat{\mu}(\diff x)} \leq V_k, \quad \text{for all }k \in [K]. 
\end{equation*}
By the triangle inequality, $\widehat{\mu}$ must be close to the true measure $\mu$ up to the first $K$ moments: 
\begin{equation} \label{eqn:moments-good}
    \Big|\int_0^B x^k\cdot \mu(\diff x) - \int_0^B x^k\cdot \widehat{\mu}(\diff x)\Big| \leq 2V_k, \quad \text{for all }k \in [K]. 
\end{equation}
Using the rounding algorithm in~\cref{lem:randomdiscretization}, we can discretize $\widehat{\mu}$ into a sorted $d$-dimensional vector $\widehat{\balpha}$ such that
\begin{equation*}
    \E_{\widehat{\balpha}} \dtv{\widehat{\balpha}}{\alpha} = \frac12 \sup_{f : \|f\|_{\mathrm{Lip}} \leq 1} \int_0^B f(x)\cdot(\mu(\diff x) - \widehat{\mu}(\diff x)). 
\end{equation*}
We can make the above supremum only over $1$-Lipschitz functions $f: \mathbb{R} \to \mathbb{R}$ satisfying $f(0) = 0$, since by \Cref{eqn:equal-mass}, $\int_0^B f(0)\cdot(\mu(\diff x) - \widehat{\mu}(\diff x)) = 0$.
Consider such an $f$; we take the best degree-$K$ polynomial approximation to it.
In other words, let $P(x) = \sum_{k=0}^K a_kx^k$ be the degree-$K$ polynomial promised by \Cref{lem:jackson}.
Then, 
\begin{equation*}
    \ABS{\int_0^B f(x)\cdot(\mu(\diff x) - \widehat{\mu}(\diff x))}  \leq \underbrace{\ABS{\int_0^B (f(x) - P(x))\cdot(\mu(\diff x) - \widehat{\mu}(\diff x))}}_{T_1:\text{ bias}} + \underbrace{\ABS{\int_0^B P(x)\cdot(\mu(\diff x) - \widehat{\mu}(\diff x))}}_{T_2:\text{ variance}}. 
\end{equation*}
Let us first bound the bias term $T_1$ using~\cref{lem:jackson} with $[a, b] = [0, B]$.
\begin{align*}
    T_1
    &\leq \int_0^B \ABS{f(x) - P(x)}\cdot (\mu(\diff x) + \widehat{\mu}(\diff x))\\
    &\leq \frac{C_3\sqrt{B}}{K}\int_0^B \sqrt{x}\cdot(\mu(\diff x) + \widehat{\mu}(\diff x)) \\
    &\leq \frac{C_3\sqrt{B}}{K} \sqrt{\Big(\int_0^B \sqrt{x}^2 \cdot(\mu(\diff x) + \widehat{\mu}(\diff x))\Big)\Big(\int_0^B 1^2 \cdot (\mu(\diff x) + \widehat{\mu}(\diff x))\Big)} \tag{by Cauchy-Schwarz} \\
    &= \frac{C_3\sqrt{2Bd}}{K} \sqrt{\int_0^B x \cdot (\mu(\diff x) + \widehat{\mu}(\diff x))} \\
    &= \frac{C_3\sqrt{2Bd}}{K} \sqrt{\int_0^B x \cdot (2\mu(\diff x)) + \int_0^B x\cdot(\widehat{\mu}(\diff x) - \mu(\diff x))} \\
    &\leq \frac{C_3\sqrt{2Bd}}{K} \sqrt{2+2V_1} . \tag{by \Cref{eqn:moments-good}} 
\end{align*}
We now bound the variance term $T_2$. 
To begin,
\begin{align*}
    \ABS{\int_0^B P(x)\cdot(\mu(\diff x) - \widehat{\mu}(\diff x))}
    &= \ABS{\int_0^B \Big(\sum_{k=0}^K a_k x^k\Big) \cdot(\mu(\diff x) - \widehat{\mu}(\diff x))}\\
    &\leq \sum_{k=0}^K |a_k| \cdot \ABS{\int_0^B  x^k \cdot(\mu(\diff x) - \widehat{\mu}(\diff x))}\\
    &= \sum_{k=1}^K |a_k| \cdot \ABS{\int_0^B  x^k \cdot(\mu(\diff x) - \widehat{\mu}(\diff x))}\\
    &\leq \sum_{k=1}^K |a_k| \cdot 2 V_k \tag{by~\Cref{eqn:moments-good}},
\end{align*}
where the second equality uses $\int_0^B \mu(dx) = \int_0^B \widehat{\mu}(dx) = d$ by \Cref{eqn:equal-mass}. 
Since $f$ is $1$-Lipschitz and $f(0)=0$, we have that $|f(x)|\leq |x|$. 
It then follows from~\cref{lem:jackson} that for any $x\in [0, B]$, 
\begin{equation*}
    |P(x)| \leq |P(x) - f(x)| + |f(x)|
    \leq \frac{C_3B}{K} + B. 
\end{equation*}
Using \cref{lem:coefs}, the coefficient $|a_k|$ for each $k \in [K]$ is bounded by
\begin{equation*}
    |a_k| \leq 2^{7K/2 + 1}B\Big(1 + \frac{C_3}{K}\Big)\Big(\frac{B}{2}\Big)^{-k} \leq 2^{9K/2 + 1}\Big(1 + \frac{C_3}{K} \Big)B^{1-k}. 
\end{equation*}
Therefore,
\begin{equation*}
    T_2 \leq 2\sum_{k=1}^K 2^{9K/2 + 1}\Big(1 + \frac{C_3}{K}\Big)B^{1-k}V_k \leq (1+C_3)2^{9K/2 + 2}\sum_{k=1}^K B^{1-k}V_k. 
\end{equation*}
Putting everything together, we have that
\begin{equation*}
    \E_{\widehat{\balpha}} \dtv{\widehat{\balpha}}{\alpha} \leq \frac12 \sup_{f : \|f\|_{\text{Lip}} \leq 1} (T_1+T_2) \leq O\Big( \frac1K\sqrt{Bd(1+V_1)} +  2^{9K/2}\sum_{k=1}^K B^{1-k}V_k\Big).
\end{equation*}
This is the claimed bound, except for the factor of $(1 + V_1)$ under the square root in the first term.
However, since $p_1(\alpha) = \alpha_1 + \cdots + \alpha_d$ is between 0 and 1 by assumption, if $\widehat{p}_1$ is outside the interval $[0, 1]$, we can always move it to this interval while only decreasing $V_1$. But in this case $V_1 \leq 1$,
completing our proof.
\end{proof}

\section{The spectrum learning algorithm}\label{sec:full_algo}

We now state our full spectrum learning algorithm and prove its correctness. 

\begin{definition}[Spectrum learning algorithm]
    Let $\rho$ be an unknown $d$-dimensional quantum state with sorted eigenvalues $\alpha$. 
    Given $2n = O(d^3\cdot (\log\log d)^4/(\log^{4}(d)\cdot \eps^{6}))$ copies of $\rho$, our spectrum learning algorithm works as follows.
    \begin{enumerate}
    \item {\bf Bucketing}~(\cref{thm:bucket-algo}): 
    Use the first $n$ copies of $\rho$
    to perform the uniform POVM bucketing algorithm
    with threshold $B = O(\epsilon^2 \log^2(d)/((\log \log d)^2\cdot d))$.
    Let $\widehat{\brho}$ and $\{\bPi, \overline{\bPi}\}$ be its outputs.
    Let $\br$ be the rank of $\bPi$
    and let $\widehat{\balpha}_1, \cdots, \widehat{\balpha}_{\br}$ be the largest $\br$ eigenvalues of $\widehat{\brho}$. 

    \item {\bf Moment estimation}~(\cref{prop:conditionedmomentvar}): Set $K = c\log d/\log \log d$ for some small constant $c\in (0, 2/19)$ to be chosen later. 
    Use the remaining $n$ copies of $\rho$ and the two-outcome measurement $\{\bPi, \overline{\bPi}\}$ to run the conditioned moment estimator (\Cref{def:conditioned-moment-estimator}) to estimate the $k$-th moment $\tr((\overline{\bPi}\rho\overline{\bPi})^k)$ in parallel for each $1 \leq k \leq K$.
    Let $\bY_1, \ldots, \bY_K$ be the resulting estimators.

    \item {\bf Local moment matching}~(\cref{thm:smallestbucketLMM}): Use local moment matching to convert the estimators $\bY_1, \ldots, \bY_K$ into estimates $\widehat{\bbeta}_1\geq \cdots\geq \widehat{\bbeta}_{d-\br}$ of the eigenvalues of $\overline{\bPi}\rho\overline{\bPi}$. 
    Let $\widehat{\balpha}_{\br+i} = \widehat{\bbeta}_{i}$ for each $i\in \{1, \cdots, d-\br\}$.  
\end{enumerate}
\end{definition}

Our main result is the following, which characterizes the behavior of our spectrum learning algorithm.

\begin{theorem}[\Cref{thm:main} restated]
    Given $n = O(d^3\cdot (\log\log d)^{4}/(\log^{4}(d)\cdot \eps^{6}))$
    copies of a mixed state $\rho$ with spectrum $\alpha$, the spectrum learning algorithm uses only unentangled measurements and outputs an estimator $\widehat{\balpha}$ such that $\dtv{\alpha}{\widehat{\balpha}} \leq \epsilon$ with probability 99\%.
\end{theorem}

\begin{proof}
We will show how to achieve an error of $O(\eps)$ with probability at least $0.98$; then, the theorem follows from rescaling $\eps$ and using standard success amplification.

Set $K = c\log d/ \log\log d$ for some small constant $c\in (0, 2/19)$. Moreover, set $B = O(\eps^2K^2/d) = O(\eps^2 \log^2(d) /((\log \log d)^2 \cdot d))$. 
The bucketing algorithm in~\cref{thm:bucket-algo} takes $n$ copies of $\rho$ and returns an estimate $\widehat{\brho}$ and a projector $\bPi$ of rank $\br\leq 3/(2B)$.
Recall that $\bPi$ is the projector onto the eigenvectors of $\widehat{\brho}$ with eigenvalues at least $B$. 
With 
\begin{equation*}
    n = O(dB^{-2}\eps^{-2}) = O(d^3 / (K^4 \eps^6)) = O(d^3 \cdot (\log \log d)^4 / (\log^4(d)\cdot \eps^6) ), 
\end{equation*}
copies,
it follows from~\Cref{item:split-large-pca} of \Cref{thm:bucket-algo} that
the $\widehat{\balpha}_1, \ldots, \widehat{\balpha}_{\br}$ approximate the largest $\br$ eigenvalues of $\rho$ up to $\epsilon$ error in TV distance. 
We also know from~\Cref{item:split-full-disturbance} from~\Cref{thm:bucket-algo} that the full spectrum of $\rho$ is disturbed by at most $\eps$ in TV distance by the measurement $\{\bPi, \overline{\bPi}\}$. 
Therefore, it suffices to estimate the eigenvalues of $\bsigma = \overline{\bPi}\rho\overline{\bPi}$ up to $\eps$ error. 

Let $\bbeta_1 \geq \cdots \geq \bbeta_d$ be the eigenvalues of $\bsigma$. 
By~\Cref{item:split-small-misclassification} from \Cref{thm:bucket-algo}, we know that $B(1+\eps)\geq \bbeta_1$. 
Therefore, for all integers $1 \leq j \leq k$,
\begin{equation*}
\tr(\sigma^{2j}) = \sum_{i=1}^d \bbeta_i^{2j}\leq d \cdot (B(1+\epsilon))^{2j}
\leq d \cdot (2B)^{2j}
\leq d 2^{2k} B^{2k} \cdot B^{2(j-k)},
\end{equation*}
where we have used the trivial bound $\eps \leq 1$. As a result, by \Cref{prop:conditionedmomentvar}, each $\bY_k$ is an unbiased estimator of $\tr(\bsigma^k)$ with variance at most
\begin{align*}
    \Var[\bY_k]
    &=\frac{24^k}{d} \sum_{j=0}^{k-1} \left(\frac{k d}{n}\right)^{k-j} \tr(\sigma^{2j}) \\
    &\leq 96^k B^{2k}k^k \sum_{j=0}^{k-1} \left(\frac{d}{nB^2}\right)^{k-j} \\
    &\leq 96^k B^{2k}k^k \cdot k\epsilon^2 \tag{because $n = O(dB^{-2}\epsilon^{-2})$} \\
    &\leq 96^k B^{2k} k^{2k} \epsilon^2 \leq 96^k B^{2k}d^{2c}\epsilon^2.
\end{align*}
where the last inequality follows because $k \leq K = c \log d / \log \log d$.

Recall that our goal is to show that $\bY_k$ is close to $\tr(\bsigma^k)$ for all $k \in [K]$ with probability $0.99$. Towards applying Chebyshev's inequality, we choose $V_k = \sqrt{100 K}\cdot \sqrt{\Var[\bY_k]} = O(\sqrt{K}\cdot 10^k B^{k}d^{c}\epsilon)$ such that
\begin{equation*}
    \Pr\left[\ABS{\bY_k - \tr(\sigma^k)} \geq V_k\right] \leq \frac{1}{100K}, \quad \text{for all } k\in [K].
\end{equation*}
Applying the union bound over all $K$ moments, we conclude that the following holds with probability $0.99$:
\begin{equation*}
    \ABS{\bY_k - \tr(\bsigma^k)} < V_k , \quad \text{for all } k\in [K].
\end{equation*}
By~\cref{thm:smallestbucketLMM}, we can find an estimate $\widehat{\bbeta}$ via a randomized algorithm which satisfies
\begin{align*}
    \E_{\widehat{\bbeta}} \dtv{\widehat{\bbeta}}{\bbeta}
     &\leq O\Big( \frac1K\sqrt{Bd} + 2^{9K/2}B\sum_{k=1}^K B^{-k}V_k\Big) \\
     &\leq O\Big( \frac1K\sqrt{
     \frac{\epsilon^2 K^2}{d} \cdot d} + 2^{9K/2}B\sum_{k=1}^K B^{-k} \cdot \sqrt{K} \cdot 10^k B^{k}d^{c}\epsilon\Big) \\
     &\leq O\Big( \eps + 2^{9K/2}B K^{3/2} 10^K d^{c}\epsilon\Big) \\
     &= O\Big( \eps + 2^{17K/2}BK^{3/2} d^c\epsilon\Big) \\
     &= O\Big( \eps + \frac{2^{17K/2}K^{7/2} \eps^3}{d^{1-c}}\Big) \tag{because $B = O(\epsilon^2 K^2 / d)$} \\
     &= O\Big( \eps + \frac{(c \log d/\log \log d)^{7/2}\eps^3}{d^{(2 - 19c)/2}}\Big) \leq O(\eps),
\end{align*}
where the last equality is due to $K = c\log d / \log \log d \leq c\log d$, and the last inequality is because $c\in (0, 2/19)$. 
The claim then follows from applying Markov's inequality. 
\end{proof}

\section{Bucketing, alignment error, and tomography} \label{sec:lmm-fails}

Recall that if $\{\Pi, \overline{\Pi}\}$ is a projective measurement which approximately splits the spectrum of~$\rho$ into the large and small buckets, the alignment error is the uniquely quantum error resulting from $\rho$ being disturbed by the measurement $\{\Pi, \overline{\Pi}\}$,
measured by the distance between the spectrum of $\rho$ and the spectrum of $\Pi \rho \Pi + \overline{\Pi} \rho \overline{\Pi}$.
In this section, we show that learning a good bucketing of $\rho$ essentially requires learning $\rho$, i.e.\ performing tomography of $\rho$,
and moreover that the relationship goes both ways.
In particular, we will show the following two results.

\begin{enumerate}
    \item First, we will show that if we have a tomography algorithm that can perform fidelity principal component analysis (PCA) tomography, then we can use it to perform bucketing with a small alignment error.
    \item Second, we will consider a natural family of quantum states and show that a good bucketing algorithm for this family can be used to design a good tomography algorithm for this family of quantum states.
    The family of quantum states we consider is those that are maximally mixed on a subspace of rank $r$.
    As there are known lower bounds for performing tomography on states of this form,
    this implies a lower bound for bucketing.
\end{enumerate}

\subsection{Fidelity PCA tomography implies bucketing with small alignment error}\label{sec:pca}

Perhaps the most natural method for learning a bucketing $\{\bPi, \overline{\bPi}\}$ of $\rho$ is to run a tomography algorithm to produce an estimate $\widehat{\brho}$ of $\rho$ and set $\bPi$ to be the projection onto $\widehat{\brho}$'s top $r$ eigenvalues, for some number $r$.
Letting $\widehat{\brho}_{\leq r} = \bPi \cdot \widehat{\brho} \cdot \bPi$ be the projection of $\widehat{\brho}$ onto its top $r$ eigenvectors, we will show that this bucketing has low alignment error if $\widehat{\brho}$ is a good approximation to the top $r$ eigenspace of $\rho$.
In particular, we want $\widehat{\brho}$ to satisfy the following \emph{principal component analysis (PCA)} guarantee.

\begin{definition}[Fidelity PCA error]
    Let $\rho \in \C^{d \times d}$ be a mixed state with eigenvalues $\alpha_1 \geq \cdots \geq \alpha_d$.
    Let $\widehat{\rho}_{\leq r}$ be a rank-$r$ PSD matrix. Then $\widehat{\rho}_{\leq r}$ has \emph{rank-$r$ fidelity PCA error $\epsilon$} with respect to~$\rho$ if
    \begin{equation*}
        \sum_{i=1}^r \alpha_i + \tr(\widehat{\rho}_{\leq r}) - 2\cdot \fidelity(\rho, \widehat{\rho}_{\leq r}) \leq  \epsilon.
    \end{equation*}
\end{definition}

To understand this fidelity PCA measure, suppose $\widehat{\rho}_{\leq r}$ were equal to the projection of $\rho$ onto its top $r$ eigenvalues.
Then
\begin{equation*}
    \fidelity(\rho, \widehat{\rho}_{\leq r}) = \alpha_1 + \cdots + \alpha_r,
    \qquad
    \text{and so}
    \qquad 
    \sum_{i=1}^r \alpha_i + \tr(\widehat{\rho}_{\leq r}) - 2\cdot \fidelity(\rho, \widehat{\rho}_{\leq r}) = 0,
\end{equation*}
meaning that $\widehat{\rho}_{\leq r}$ has a rank-$r$ fidelity PCA error of $0$ with respect to $\rho$.
More generally, the quantity $\sum_{i=1}^r \alpha_i + \tr(\widehat{\rho}_{\leq r}) - 2\cdot \fidelity(\rho, \widehat{\rho}_{\leq r})$ is actually minimized by this $\widehat{\rho}_{\leq r}$, meaning that
\begin{equation*}
    \sum_{i=1}^r \alpha_i + \tr(\widehat{\rho}_{\leq r}) - 2\cdot \fidelity(\rho, \widehat{\rho}_{\leq r}) \geq 0
\end{equation*}
for all $\widehat{\rho}_{\leq r}$ (which corresponds to every $\widehat{\rho}_{\leq r}$ have nonnegative fidelity PCA error).
To see this, if $\Pi$ is the projection onto $\widehat{\rho}_{\leq r}$'s $r$ nonzero eigenvalues, we have
\begin{align}
    \fidelity(\rho, \widehat{\rho}_{\leq r})
    & = \fidelity(\Pi\rho\Pi, \widehat{\rho}_{\leq r})
    \tag{by \cite[Proposition 3.12 (4.)]{Wat18}}\nonumber\\
    & \leq \sqrt{\tr(\Pi\rho\Pi) \cdot \tr(\widehat{\rho}_{\leq r})} \tag{by \cite[Proposition 3.12 (6.)]{Wat18}}\nonumber\\
    & \leq \frac{1}{2} \cdot \tr(\Pi \rho \Pi) + \frac{1}{2} \cdot \tr(\widehat{\rho}_{\leq r}) \label{eq:gonna-use-this-later}\\
    & \leq\frac{1}{2} \cdot (\alpha_1 + \cdots + \alpha_r) + \frac{1}{2} \cdot \tr(\widehat{\rho}_{\leq r}),\nonumber
\end{align}
where the second inequality is because $2 a b \leq a^2 + b^2$.
Thus,
\begin{equation*}
    \sum_{i=1}^r \alpha_i + \tr(\widehat{\rho}_{\leq r}) - 2\cdot \fidelity(\rho, \widehat{\rho}_{\leq r})
    \geq \sum_{i=1}^r \alpha_i + \tr(\widehat{\rho}_{\leq r}) - 2\cdot\Big(\frac{1}{2} \cdot (\alpha_1 + \cdots + \alpha_r) + \frac{1}{2} \cdot \tr(\widehat{\rho}_{\leq r})\Big)
    = 0.
\end{equation*}
Finally, when $r = d$ and $\widehat{\rho} \coloneqq \widehat{\rho}_{\leq d}$ is a density matrix (i.e.\ it has trace 1), then the rank-$d$ fidelity PCA error is just $2 \cdot (1 - \fidelity(\rho, \widehat{\rho}))$, twice the infidelity.
We note that a related fidelity PCA measure was studied in \cite[Theorem 1.19]{OW17a}, with quantum affinity used in place of the fidelity.

We now show that a small fidelity PCA error 
implies a bucketing with a small alignment error.

\begin{lemma}[PCA implies bucketing]
    \label{lem:pca-implies-bucketing}
    Let $\rho$ be a $d$-dimensional quantum state and let $\widehat{\rho}_{\leq r}$ have rank-$r$ fidelity PCA error $\epsilon$ with respect to $\rho$. Setting $\Pi$ to be the projector onto $\widehat{\rho}_{\leq r}$'s nonzero eigenspace, we have that
    \begin{equation*}
        \dtv{\spec(\Pi \rho \Pi + \overline{\Pi} \rho \overline{\Pi} )}{\spec(\rho)}
    \leq \epsilon.
    \end{equation*}  
\end{lemma}

\begin{proof}
    Let $\alpha_1 \geq \cdots \geq \alpha_d$ be the eigenvalues of $\rho$.
    Let us denote the eigenvalues of $\Pi \rho\Pi$ as $\beta_1, \dots, \beta_r$ and the eigenvalues of $\overline{\Pi} \rho \overline{\Pi}$ as $\beta_{r+1}, \dots, \beta_d$. 
    It follows from Cauchy's interlacing theorem (\cref{thm:cauchy_interlace}) that
    \begin{align*}
        \alpha_i &\geq \beta_i, \quad \text{for }i\in \{1,\ldots, r\}, \\
        \alpha_i &\leq \beta_i, \quad \text{for }i\in \{r+1,\ldots, d\}. 
    \end{align*}
    As we have seen in the proof of~\Cref{thm:bucket-algo}, it is not necessarily true that $\beta_{r} \geq \beta_{r+1}$, and so $\beta_1, \ldots, \beta_d$ are not necessarily in sorted order. But because the TV distance between two vectors is minimized when they are sorted~\cite[Proposition 2.2]{OW15}, we have
    \begin{align*}
        \dtv{\spec(\Pi\rho \Pi + \overline{\Pi} \rho \overline{\Pi})}{\spec(\rho)} &\leq \frac12 \left(\sum_{i=1}^r (\alpha_i - \beta_i) + \sum_{i=r+1}^d (\beta_i - \alpha_i)\right) \\
        &= \sum_{i=1}^r (\alpha_i - \beta_i) \tag{because $\sum_{i=1}^d \alpha_i = \sum_{i=1}^d \beta_i = 1$}\\
        &= \sum_{i=1}^r \alpha_i - \tr(\Pi \rho \Pi)\\
        &\leq \sum_{i=1}^r \alpha_i + \tr(\widehat{\rho}_{\leq r}) - 2\fidelity(\rho, \widehat{\rho}_{\leq r}). \tag{by \Cref{eq:gonna-use-this-later}}
    \end{align*}
    But this is at most $\epsilon$ since $\widehat{\rho}_{\leq r}$ has rank-$r$ fidelity PCA error $\epsilon$.
\end{proof}

How many copies are actually needed to perform rank-$r$ fidelity PCA?
Prior to answering this, let us first consider the related problem of rank-$r$ fidelity tomography, in which $\rho$ is promised to have rank $r$ (rather than in the PCA setting, where we make no such assumption).
The best known rank-$r$ fidelity tomography algorithms with entangled measurements use $\widetilde{O}(dr/\epsilon)$ copies~\cite{HHJ+16}
to achieve infidelity $\epsilon$,
and the best known algorithms with unentangled measurements use $\widetilde{O}(dr^2/\epsilon)$ copies~\cite{CHL+23,FO24}.
We expect that the best rank-$r$ fidelity PCA algorithms should be able to match the copy complexity of these best known tomography algorithms, although this is not yet known.
The closest existing result is \cite[Theorem 1.19]{OW17a}, which  gives a rank-$r$ PCA-style algorithm with entangled measurements using $n = \widetilde{O}(dr/\epsilon)$ copies; however, the precise guarantee is for quantum affinity rather than quantum fidelity, and it is slightly weaker than the best PCA-type bound one would hope for.
We do believe it might be possible to show that the unentangled measurement fidelity tomography algorithm from~\cite{CHL+23} might also have a fidelity PCA result.
However, we have chosen not to explore this as their bound for the simpler rank-$r$ tomography case comes with additional log factors that we can't afford to lose.

We can also obtain a bound on the number of copies needed for fidelity PCA by instead performing trace distance PCA.
To begin, let us define trace distance PCA.

\begin{definition}[Trace distance PCA]
    \label{def:trace-distance-pca}
    Let $\rho \in \C^{d \times d}$ be a mixed state with eigenvalues $\alpha_1 \geq \cdots \geq \alpha_d$.
    Let $\widehat{\rho}_{\leq r}$ be a rank-$r$ PSD matrix. Then $\widehat{\rho}_{\leq r}$ has \emph{rank-$r$ trace distance PCA error $\epsilon$} with respect to~$\rho$ if
    \begin{equation*}
        2\cdot \dtr(\rho, \widehat{\rho}_{\leq r}) - \sum_{i=r+1}^d \alpha_i \leq  \epsilon.
    \end{equation*}
\end{definition}

Just as in the case of fidelity PCA,
if $\widehat{\rho}_{\leq r}$ is equal to the projection of $\rho$ onto its top $r$ eigenvalues, then $\widehat{\rho}_{\leq r}$ has rank-$r$ trace distance PCA error $\epsilon = 0$,
and otherwise its error is $> 0$.
The following lemma shows that 
an algorithm for trace distance PCA can be converted to an algorithm for fidelity PCA.

\begin{lemma}[Trace distance PCA implies fidelity PCA]
    \label{lem:trace-pca-implies-fidelity-pca}
    Let $\widehat{\rho}_{\leq r}$ be a rank-$r$ PSD matrix. If $\widehat{\rho}_{\leq r}$ has rank-$r$ trace distance PCA error $\epsilon$ with respect to~$\rho$, it also has rank-$r$ fidelity PCA error at most $\epsilon$.
\end{lemma}

\begin{proof}
    Since $\widehat{\rho}_{\leq r}$ is not normalized, we cannot directly use the Fuchs--van de Graaf inequalities from~\Cref{lem:trace-vs-fidelity} to lower bound the trace distance in terms of fidelity. Instead, we define the related density matrices $\sigma, \widehat{\sigma}_{\leq r} \in \mathbb{C}^{(d+1)\times (d+1)}$ that satisfy
    \begin{equation*}
        \sigma = \rho, \qquad\widehat{\sigma}_{\leq r} = \widehat{\rho}_{\leq r} + (1 - \tr(\widehat{\rho}_{\leq r}))\cdot \ketbra{d+1}.  
    \end{equation*}
    The trace distance and fidelity of these two mixed states relate to those of $\rho$ and $\widehat{\rho}_{\leq r}$ as below
    \begin{equation*}
        2\cdot \dtr(\sigma, \widehat{\sigma}_{\leq r}) = 2\cdot \dtr(\rho, \widehat{\rho}_{\leq r}) + 1 - \tr(\widehat{\rho}_{\leq r}),\qquad
        \fidelity(\sigma, \widehat{\sigma}_{\leq r}) = \fidelity(\rho, \widehat{\rho}_{\leq r}).
    \end{equation*}
    Thus we can use the Fuchs--van de Graaf inequalities (\Cref{lem:trace-vs-fidelity}) to deduce 
    \begin{align*}
        2\cdot\dtr(\rho, \widehat{\rho}_{\leq r})
        &= 2\cdot\dtr(\sigma, \widehat{\sigma}_{\leq r}) - 1 + \tr(\widehat{\rho}_{\leq r}) \nonumber \\
        &\geq 2 - 2\cdot \fidelity(\sigma, \widehat{\sigma}_{\leq r}) - 1 + \tr(\widehat{\rho}_{\leq r}) \\
        &=1 + \tr(\widehat{\rho}_{\leq r}) - 2\fidelity(\rho, \widehat{\rho}_{\leq r}).
    \end{align*}
    If we rearrange this by writing $1 = \sum_{i=1}^r \alpha_i + \sum_{i=r+1}^d \alpha_i$, we have that
    \begin{equation*}
        \sum_{i=1}^r \alpha_i + \tr(\widehat{\rho}_{\leq r}) - 2\fidelity(\rho, \widehat{\rho}_{\leq r})
        \leq 2\cdot \dtr(\rho, \widehat{\rho}_{\leq r}) - \sum_{i=r+1}^d \alpha_i.
    \end{equation*}
    But this is at most $\epsilon$ because $\widehat{\rho}_{\leq r}$ has rank-$r$ trace distance PCA.
\end{proof}

\cite[Corollary 1.6]{OW16} shows that it is possible to perform rank-$r$ trace distance PCA with entangled measurements up to error $\epsilon$ using $O(dr/\epsilon^2)$ copies.
This is essentially tight, as Haah et al.~\cite{HHJ+16} showed that $\widetilde{\Omega}(dr/\epsilon^2)$ copies are necessary to perform trace distance tomography on a state $\rho$, promised that it is rank $r$ (which is a special case of rank-$r$ trace distance PCA).
\Cref{lem:trace-pca-implies-fidelity-pca} then implies that $O(dr/\epsilon^2)$ samples also suffice to perform rank-$r$ fidelity PCA.
From our above discussion, we believe that this has a suboptimal $\epsilon$ dependence but an optimal dependence on $d$ and $r$.

\subsection{Bucketing implies tomography with small infidelity}

We consider the class of states $\rho = \frac{1}{r} \cdot P$, where $P$ is a rank-$r$ projector, and show that a good enough bucketing algorithm implies a tomography algorithm for this class of states.
This implies a lower bound for the number of copies needed to perform bucketing, as there are known lower bounds for the number of copies needed to perform tomography on this class of states.
First, however, we must answer the question: how to formally define a good enough bucketing algorithm?
 Even though bucketing may be complicated to define in full generality, when we restrict our attention to the family of states described above, our requirements for bucketing are simpler to state.

\begin{definition}[Simple bucketing for maximally mixed states over subspace]
    \label{def:bucketing}
    Given a state of the form $\rho = \frac{1}{r} \cdot P$, a projective measurement $\{\Pi, \overline{\Pi} = I - \Pi\}$ defines a \emph{simple bucketing with error $\epsilon$} if it satisfies the following two properties.
    \begin{itemize}
        \item[$\circ$] \textbf{Classification of eigenvalues.} $\Pi \rho \Pi$ contains all the large eigenvalues, and $\overline{\Pi}\rho\overline{\Pi}$ the small eigenvalues. Formally,
        \begin{equation*}
            \Pi \rho \Pi \succcurlyeq \frac{1}{2r}\cdot \Pi,
            \quad \text{and} \quad \overline{\Pi}\rho\overline{\Pi} \preccurlyeq \frac{1}{2r}\cdot \overline{\Pi}.
        \end{equation*}

        \item[$\circ$] \textbf{Small alignment error.} Measuring $\rho$ using $\{\Pi, \overline{\Pi}\}$ disturbs its spectrum by at most $\epsilon$ in total variation distance:
        \begin{equation*}
            \dtv{\spec\left(\Pi \rho \Pi + \overline{\Pi} \rho \overline{\Pi} \right)}{\spec(\rho)} \leq \epsilon.
        \end{equation*}
    \end{itemize}
\end{definition}

We refer to this as a ``simple'' bucketing because a more general bucketing scheme need not look as simple as this; for example, it might involve more than just two buckets, or it might allow for some (slight) overlap between the buckets.
(Indeed, even our bucketing scheme from \Cref{thm:bucket-algo} does not precisely fit this mold.)
Thus, although we do not claim that this definition captures all possible bucketing schemes, we use it as a simple proof-of-concept to demonstrate the challenges that a bucketing scheme must overcome.
The following theorem shows that a simple bucketing with $\{\Pi, \overline{\Pi}\}$ for a quantum state $\rho$ drawn from the family of states described above can be converted to a state $\widehat{\rho}$ that is close to $\rho$.

\begin{theorem}[Bucketing implies learning]
    \label{thm:bucketing-implies-tomography-fid}
    Let $P$ be a rank-$r$ projector and $\rho = \frac{1}{r}\cdot P$ be a mixed state. Let $\{\Pi, \overline{\Pi}=I-\Pi\}$ be a simple bucketing with error $\epsilon$ for $\rho$.
    Set $\widehat{\rho} = \Pi/\tr(\Pi)$. Then it holds that
    \begin{equation*}
        1 - \fidelity(\rho, \widehat{\rho})
        \leq \epsilon.
    \end{equation*}
\end{theorem}

Hence, to perform simple bucketing, one must perform fidelity tomography, at least for this family of quantum states.
This implies a lower bound for the number of copies needed to perform bucketing because, as stated above, there are known lower bounds for the number of copies needed to perform tomography on this class of states;
in particular, it was shown by Haah et al.~\cite{HHJ+16} that $\widetilde{\Omega}(dr/\epsilon)$ copies are necessary to learn a state which is maximally mixed over a rank-$r$ subspace to infidelity $\epsilon$.
In fact, this follows from a stronger lower bound that $\widetilde{\Omega}(dr/\delta^2)$ copies are necessary to learn such a state to trace distance error $\delta$.
For this, we believe that the ``tilde'' is an artifact of their proof,
and that the right lower bound should be $\Omega(dr/\delta^2)$, which would be optimal since it is known that $O(dr/\delta^2)$ copies suffice to perform tomography on rank-$r$ states using entangled measurements~\cite{OW16,HHJ+16}.
This would imply that $\Omega(dr/\epsilon)$ copies are ncessary to learn these states to infidelity $\epsilon$.
In addition, we believe that $n = \Omega(dr^2/\delta^2)$ copies should be required to learn these states to trace distance error $\delta$ using unentangled measurements, which would match the known upper bound for rank-$r$ unentangled tomography~\cite{GKKT20}, although we stress that as far as we know, showing a lower bound of $\Omega(dr^2/\delta^2)$ for any family of rank-$r$ states is an open problem (the known lower bound of $\Omega(d^3/\delta^2)$ from~\cite{CHL+23} applies to states which are full rank).
Again, this would imply that $\Omega(dr^2/\epsilon)$ copies are necessary to learn these states to infidelity $\epsilon$ using unentangled measurements.
Together, however, we believe that these suggest that simple bucketing should require $n = \Omega(dr/\epsilon)$ copies for entangled measurements and $n = \Omega(dr^2/\epsilon)$ copies for unentangled measurements.

Let us now try to interpret this state of affairs.
Typically, as in \Cref{thm:bucket-algo},
a bucketing scheme picks a threshold $0 \leq B \leq 1$ and tries to bucket $\rho$'s eigenvalues into those which are bigger than $B$ and those which are smaller than $B$.
If, say, $B$ were to equal $1/2r$ for some integer $r$ and the provided $\rho$ was maximally mixed on a subspace of dimension $r$, then this would entail a simple bucketing of $\rho$, which the previous paragraph suggests would require $n = \Omega(dr/\epsilon) = \Omega(dB^{-1}/\epsilon)$ copies in the entangled case and
$n = \Omega(dB^{-2}/\epsilon)$ copies in the unentangled case.
(We note that our unentangled bucketing algorithm from \Cref{thm:bucket-algo} uses $O(dB^{-2}/\epsilon^2)$ copies, suggesting that it is optimal at least for constant $\epsilon$. That said, we stress again that this bucketing algorithm does not quite give a ``simple'' bucketing.)
On the flip side, any $\rho$ will have at most $r' = B^{-1}$ eigenvalues greater than $B$, and  a natural way to bucket them is to perform rank-$r'$ fidelity PCA, which \Cref{sec:pca} suggests might be achievable with $O(dr'/\epsilon) = O(dB^{-1}/\epsilon)$ copies using entangled measurements.
All in all, we believe that these results suggest that $\Theta(dB^{-1}/\epsilon)$ copies might be the optimal number of copies needed to bucket based on a threshold $B$ using entangled measurements (perhaps for any natural notion of ``bucketing''), and $\Theta(dB^{-2}/\epsilon)$ might be the optimal number of copies needed to bucket using unentangled measurements.

\begin{proof}[Proof of~\Cref{thm:bucketing-implies-tomography-fid}]
    First, we observe that the projector $\Pi$ must have rank exactly equal to $r$.
    To see why $\Pi$ cannot have rank greater than $r$,
    note that
     $\Pi \rho \Pi$ has rank at most $r$
     because $\rho$ is rank $r$.
     Thus, the ``classification of eigenvalues'' property $\Pi \rho \Pi \succeq \frac{1}{2r}\cdot \Pi$ of~\Cref{def:bucketing}
    cannot hold if $\Pi$'s rank is greater than $r$. 
    To see why $\Pi$ cannot have rank less than $r$,
    note that if does, then there exists an eigenvector of $\rho$ that is orthogonal to $\Pi$, which means that $\overline{\Pi} \rho \overline{\Pi}$ contains this eigenvector, with eigenvalue $\frac{1}{r}$. This again contradicts the property $\overline{\Pi} \rho \overline{\Pi} \preceq \frac{1}{2r}\cdot \overline{\Pi}$. This implies that $\Pi$ has rank $r$.

    Since $P$ and $\Pi$ are both rank $r$,
    \begin{align*}
        \fidelity(\rho, \widehat{\rho})
        &= \tr\sqrt{\sqrt{\widehat{\rho}}\rho\sqrt{\widehat{\rho}}}\\
        &= \frac{1}{r}\cdot \tr\sqrt{\sqrt{\Pi} P \sqrt{\Pi}}  \\
        &= \frac{1}{r}\cdot \tr(\sqrt{\Pi P \Pi})  \\
        &\geq \frac{1}{r}\cdot \tr(\Pi P \Pi) \tag{because $\Pi P \Pi \preceq I$} \\
        &= \tr(\Pi \rho \Pi).
    \end{align*}
    Let the eigenvalues of $\Pi \rho \Pi + \overline{\Pi} \rho \overline{\Pi}$ be $\alpha_1 \geq \dots \geq \alpha_d$. From the ``classification of eigenvalues'' property of~\Cref{def:bucketing}, we know that the top $r$ eigenvalues $\alpha_1, \dots, \alpha_r$ are the eigenvalues of $\Pi \rho \Pi$,
    and the bottom $(d-r)$ eigenvalues $\alpha_{r+1}, \ldots, \alpha_d$ are the eigenvalues of $\overline{\Pi} \rho \overline{\Pi}$.
    Moreover, we know that the $\alpha_i$'s are all $\leq 1/r$
    since $\rho$'s maximum eigenvalue is $1/r$.
    Thus,
    \begin{align*}
        1 - \fidelity(\rho, \widehat{\rho})
        \leq \tr(\overline{\Pi} \rho \overline{\Pi})
        = \sum_{i=r+1}^d \alpha_i
        &= \frac{1}{2} \cdot \Big(1 - \sum_{i=1}^r \alpha_i\Big) + \frac{1}{2} \cdot \sum_{i=r+1}^d \alpha_i \tag{because $\sum_i \alpha_i = 1$}\\
        &=\frac{1}{2} \cdot \sum_{i=1}^r\Big(\frac{1}{r} - \alpha_i\Big) + \frac{1}{2} \cdot \sum_{i=r+1}^d \alpha_i\\
        &=\frac{1}{2} \cdot \sum_{i=1}^r\Big|\frac{1}{r} - \alpha_i\Big| + \frac{1}{2} \cdot \sum_{i=r+1}^d |\alpha_i|.
    \end{align*}
    This is equal to the total variation distance between $\alpha$
    and the distribution $(\frac{1}{r}, \ldots, \frac{1}{r}, 0, \ldots, 0)$, which is the spectrum of $\rho$.
    Thus, we have shown that
    \begin{equation*}
        1 - \fidelity(\rho, \widehat{\rho})
        \leq \dtv{\spec\left(\Pi \rho \Pi + \overline{\Pi} \rho \overline{\Pi} \right)}{\spec(\rho)},
    \end{equation*}
    and this is at most $\epsilon$ because $\{\Pi, \overline{\Pi}\}$ is a simple bucketing with error $\epsilon$ for $\rho$.
\end{proof}

\section{Computational evidence for lower bounds} \label{sec:lower}

In this section, we perform numerical experiments to understand the optimal number of copies needed for spectrum estimation, in the setting where fully entangled measurements are allowed.
To do so, we consider the following two-point distinguishing game.
\begin{definition}[$\alpha$-versus-$\beta$ spectrum distinguishing game] \label{def:distinguishing-game}
    Let $\alpha = (\alpha_1, \ldots, \alpha_d)$ and $\beta = (\beta_1, \ldots, \beta_d)$ be two possible mixed state spectra.
    The \emph{$\alpha$-versus-$\beta$ spectrum distinguishing game} refers to the following task.
    A distinguisher is given $n$ copies of a random mixed state $\brho$ sampled from the following distribution.
    \begin{enumerate}
        \item Flip a fair $\{\mathrm{H}, \mathrm{T}\}$ coin and let $\bc$ be the outcome. If $\bc = \mathrm{H}$, set $\bgamma = \alpha$. If $\bc = \mathrm{T}$, set $\bgamma = \beta$.
        \item Sample a Haar random unitary $\bU \sim \mathrm{U}(d)$.
        \item Set $\brho = \bU \cdot \bgamma \cdot \bU^{\dagger}$.
    \end{enumerate}
    The distinguisher performs a measurement on $\brho^{\otimes n}$ and outputs a guess for whether $\brho$'s spectrum is equal to $\alpha$ or $\beta$, and it succeeds if it guesses correctly.
\end{definition}

\noindent
Suppose there is an algorithm $\calA$ for spectrum estimation which uses $f(d, \epsilon, \delta)$ copies;
in particular, given $n$ copies of a mixed state $\rho \in \C^{d \times d}$ with spectrum $\gamma$, $\calA$ outputs an estimator $\widehat{\bgamma}$ such that $\dtv{\gamma}{\widehat{\bgamma}} \leq \epsilon$ with probability $1-\delta$.
Then we can use it to design a distinguisher for the $\alpha$-versus-$\beta$ spectrum distinguishing game, as follows.
Suppose $\dtv{\alpha}{\beta} > 2\epsilon$.
Then given $\brho^{\otimes n}$, the distinguisher runs $\calA$ to produce an estimate $\widehat{\bgamma}$ of $\brho$'s spectrum;
if $\dtv{\alpha}{\bgamma} < \dtv{\beta}{\bgamma}$, the distinguisher guesses that $\brho$'s spectrum is equal to $\alpha$, and otherwise it guesses that it is equal to $\beta$.
We claim that the distinguisher succeeds with probability at least $1 - \delta$.
To see this, suppose without loss of generality that $\bgamma$ is selected to be $\alpha$.
Then with probability $1-\delta$,
we will have $\dtv{\alpha}{\widehat{\bgamma}} \leq \epsilon$.
This means that $\dtv{\beta}{\widehat{\bgamma}} > \epsilon$,
as otherwise we would have $\dtv{\alpha}{\beta} \leq \dtv{\alpha}{\widehat{\bgamma}} + \dtv{\beta}{\widehat{\bgamma}} \leq 2\epsilon$, a contradiction.
Thus, $\dtv{\alpha}{\widehat{\bgamma}} \leq \epsilon < \dtv{\beta}{\widehat{\bgamma}}$,
and so the algorithm will correctly guess that $\brho$'s spectrum is equal to $\alpha$ with probability at least $1-\delta$.

This means that a lower bound on the number of samples needed to win the $\alpha$-versus-$\beta$ spectrum distinguishing game translates to a lower bound on the number of samples $f(d, \epsilon, \delta)$ needed to perform spectrum estimation.
In the classical setting of sorted distribution estimation, the lower bound of $n = \Omega(d / \log(d))$ samples is proven using essentially a two-point distinguishing game of this form~\cite{WY16,HJW18}, 
so we expect this distinguishing task to capture most of the difficulty of spectrum estimation.
We will primarily consider the cases when $\epsilon$ and $\delta$ are constants,
in which case we are aiming to lower bound the ``$d$ dependence'' of spectrum estimation.
Below, we describe our approach for numerically simulating the optimal distinguisher.

\subsection{Defining the optimal distinguisher}

It is well-known that the optimal distinguisher using entangled measurements in the $\alpha$-versus-$\beta$ distinguishing game has a natural definition in terms of the representation theory of the groups $S_n$ and $\text{U}(d)$.
We will outline this distinguisher and our numerical simulations of it below.
We assume familiarity with representation theory,
and refer the reader to~\cite{Wri16} for a thorough treatment of this topic.

Write $\rho_{\alpha}$ for the average mixed state the distinguisher receives in the case when $\bgamma = \alpha$, i.e.
\begin{equation*}
    \rho_{\alpha} = \E_{\bU \sim \text{U}(d)} (\bU \cdot \alpha \cdot \bU^{\dagger})^{\otimes n}.
\end{equation*}
Define $\rho_{\beta}$ similarly. Then the distinguisher's task is to distinguish $\rho_{\alpha}$ from $\rho_{\beta}$, and its optimal success probability is given by
\begin{equation}\label{eq:opt-success}
    \frac{1}{2} + \frac{1}{2} \cdot \dtr(\rho_{\alpha},\rho_{\beta})
    = \frac{1}{2} + \frac{1}{2} \cdot \tr(Q \cdot (\rho_{\alpha} - \rho_{\beta}))
    = \frac{1}{2} \cdot \tr(Q \cdot \rho_{\alpha}) + \frac{1}{2} \cdot \tr(\overline{Q} \cdot \rho_{\beta}),
\end{equation}
where $Q$ is the projector onto the positive eigenvalues of the matrix $\rho_{\alpha} - \rho_{\beta}$.
In particular, the optimal distinguisher measures its input with the projective measurement $\{Q, \overline{Q}\}$,
guesses that $\brho$'s spectrum is $\alpha$ if it observes $Q$,
and guess that $\brho$'s spectrum is $\beta$ if it observes $\overline{Q}$.

To define this projector $Q$, we need to understand the eigenbases of $\rho_{\alpha}$ and $\rho_{\beta}$, and this can be done via representation theory.
In particular, the representation theoretic result known as \emph{Schur-Weyl duality} states that there is a unitary change of basis $U_{\mathrm{Schur}}$ on $(\C^d)^{\otimes n}$ such that
\begin{equation*}
    U_{\mathrm{Schur}} \cdot \brho^{\otimes n} \cdot U_{\mathrm{Schur}}^\dagger
    = \sum_{\lambda \vdash n, \ell(\lambda) \leq d} \ketbra{\lambda} \otimes I_{\dim(\lambda)} \otimes \nu_{\lambda}(\brho).
\end{equation*}
Here, given a Young diagram $\lambda$,
we write $(\kappa_{\lambda}, \mathrm{Sp}_{\lambda})$ for the corresponding irrep of $S_n$
and $(\nu_{\lambda}, V_{\lambda}^d)$ for the corresponding irrep of the general linear group $\mathrm{GL}(d)$ (which also serves as an irrep of the unitary group $\mathrm{U}(d)$).
Then $\dim(\lambda)$ is the dimension of the Specht module $\mathrm{Sp}_{\lambda}$, and so the $I_{\dim(\lambda)}$ term is just the identity over the symmetric group irrep corresponding to $\lambda$.
Given this, we can write
\begin{align*}
    \rho_{\alpha}
    &=\E_{\bU \sim \text{U}(d)} (\bU \cdot \alpha \cdot \bU^{\dagger})^{\otimes n} \\
    &= U_{\mathrm{Schur}}^{\dagger} \cdot \Big(\sum_{\lambda \vdash n, \ell(\lambda) \leq d} \ketbra{\lambda} \otimes I_{\dim(\lambda)} \otimes \E_{\bU \sim \text{U}(d)}\nu_{\lambda}(\bU \cdot \alpha \cdot \bU^{\dagger})\Big) \cdot U_{\mathrm{Schur}}\\
    &= U_{\mathrm{Schur}}^{\dagger} \cdot \Big(\sum_{\lambda \vdash n, \ell(\lambda) \leq d} \ketbra{\lambda} \otimes I_{\dim(\lambda)} \otimes \frac{s_{\lambda}(\alpha)}{\dim(V_{\lambda}^d)} \cdot I_{\dim(V_{\lambda}^d)}\Big) \cdot U_{\mathrm{Schur}}. \tag{by Schur's lemma}
\end{align*}
Similarly, we have that
\begin{equation*}
    \rho_{\beta} = U_{\mathrm{Schur}}^{\dagger} \cdot \Big(\sum_{\lambda \vdash n, \ell(\lambda) \leq d} \ketbra{\lambda} \otimes I_{\dim(\lambda)} \otimes \frac{s_{\lambda}(\beta)}{\dim(V_{\lambda}^d)} \cdot I_{\dim(V_{\lambda}^d)}\Big) \cdot U_{\mathrm{Schur}}
\end{equation*}
Hence, if we define
\begin{equation*}
    \Pi_{\lambda} = U^{\dagger}_{\mathrm{Schur}} \cdot \big(\ketbra{\lambda} \otimes I_{\dim(\lambda)} \otimes I_{\dim(V_{\lambda}^d)}\big) \cdot U_{\mathrm{Schur}}
\end{equation*}
to be the projector onto the $\lambda$-irrep space, then we have the following eigendecompositions for our two matrices:
\begin{equation*}
    \rho_{\alpha} = \sum_{\lambda \vdash n, \ell(\lambda) \leq d} \frac{s_{\lambda}(\alpha)}{\dim(V_{\lambda}^d)} \cdot \Pi_{\lambda},
    \quad \text{and} \quad 
    \rho_{\beta} = \sum_{\lambda \vdash n, \ell(\lambda) \leq d} \frac{s_{\lambda}(\beta)}{\dim(V_{\lambda}^d)} \cdot \Pi_{\lambda}.
\end{equation*}
Hence, the optimal distinguisher is defined in terms of the projector
\begin{equation*}
    Q = \sum_{\lambda: s_{\lambda}(\alpha) > s_{\lambda}(\beta)} \Pi_{\lambda}.
\end{equation*}
The set of matrices $\{\Pi_{\lambda}\}$ gives a projective measurement known as \emph{weak Schur sampling}. 
Given $\rho_{\alpha}$, weak Schur sampling produces the Young diagram $\lambda$ with probability $\dim(\lambda) \cdot s_{\lambda}(\alpha)$,
and given $\rho_{\beta}$ it produces this Young diagram with probability $\dim(\lambda) \cdot s_{\lambda}(\beta)$.
Putting everything together, we can view the optimal tester as being equal to the following maximum likelihood tester on the Young diagram produced by weak Schur sampling.
\begin{enumerate}
    \item Perform weak Schur sampling on $\brho^{\otimes n}$ obtain a random Young diagram $\blambda \vdash n$. \label{item:wss}
    \item Compare the probability of the measurement outcome being $\blambda$ if the underlying state $\brho$ has spectrum $\alpha$ or if it has spectrum $\beta$; output $\alpha$ if it gives the larger probability, and $\beta$ otherwise.
\end{enumerate}

\subsection{Numerically simulating the optimal distinguisher}

By \Cref{eq:opt-success}, we can now compute the success probability of the optimal distinguisher as
\begin{equation*}
    \frac{1}{2} \cdot \tr(Q \cdot \rho_{\alpha}) + \frac{1}{2} \cdot \tr(\overline{Q} \cdot \rho_{\beta}) = \frac{1}{2} \cdot \sum_{\lambda : s_{\lambda}(\alpha) > s_{\lambda}(\beta)} \dim(\lambda) \cdot s_{\lambda}(\alpha)
    + \frac{1}{2} \cdot \sum_{\lambda : s_{\lambda}(\alpha) \leq s_{\lambda}(\beta)} \dim(\lambda) \cdot s_{\lambda}(\beta).
\end{equation*}
This gives an explicit formula which can be computed in theory.
However, in practice, computing this formula is intractable because it involves a sum over the $2^{\Theta(\sqrt{n})}$ Young diagrams of size $n$.
Instead, we approximate this sum by sampling.
To begin, write $\mathrm{SW}^n(\bgamma)$ for the distribution on Young diagrams in \Cref{item:wss} above produced by performing weak Schur sampling on $\brho^{\otimes n}$, assuming that $\brho$ has spectrum $\bgamma$.
It was shown by O'Donnell and Wright~\cite{OW15} that the following classical algorithm is able to sample a Young diagram $\blambda$ from this distribution.
\begin{enumerate}
    \item Sample an $n$-letter $\bgamma$-random word $\bw = (\bw_1, \ldots, \bw_n) \in [d]^n$, meaning that each coordinate $\bw_i$ is sampled independently from the distribution $\bgamma$.
    \item Perform the the Robinson--Schensted--Knuth algorithm on $\bw$ to attain a Young diagram $\blambda = \mathrm{shRSK}(\bw)$.
\end{enumerate}
The RSK algorithm is an efficient, polynomial-time algorithm, and so this gives an efficient algorithm for sampling from $\mathrm{SW}^n(\bgamma)$.
With this in place, we use the following algorithm for approximating the success probability of the optimal distinguisher.
\begin{definition}[Approximating the success probability of the optimal distinguisher] \label{def:code-alg}
Let $m$ be a specified number of samples.
The following algorithm produces an estimate of the success probability of the optimal distinguisher in the $\alpha$-versus-$\beta$ spectrum distinguishing game.
\begin{enumerate}
    \item Sample $m$ Young diagrams from $\mathrm{SW}^n(\alpha)$. Let $\mathrm{succ}_\alpha$ be the number of Young diagrams $\blambda$ such that $s_{\blambda}(\alpha) > s_{\blambda}(\beta)$.
    \item Sample $m$ Young diagrams from $\mathrm{SW}^n(\beta)$. Let $\mathrm{succ}_\beta$ be the number of Young diagrams $\blambda$ such that $s_{\blambda}(\alpha) \leq s_{\blambda}(\beta)$.
    \item Output $(\mathrm{succ}_{\alpha} + \mathrm{succ}_{\beta})/(2m)$.
\end{enumerate}
\end{definition}

As above, producing samples from $\mathrm{SW}^n(\alpha)$ and $\mathrm{SW}^n(\beta)$ can be done efficiently.
In addition, the comparisons between $s_{\blambda}(\alpha)$ and $s_{\blambda}(\beta)$ in steps 1 and 2 can be made efficient as well.
For these, it suffices to compute $s_{\blambda}(\alpha)$ and $s_{\blambda}(\beta)$, and this can be done efficiently and stably due to the algorithm of~\cite{CDEKK19}.
Overall, then, this is an efficient algorithm.
Moreover, standard Chernoff bounds say that the estimate it produces is within $\epsilon$ of the true optimal success probability except with probability $2e^{-4m \epsilon^2}$.

\subsection{Hard to distinguish pairs of spectra}

Now we describe the spectra $\alpha$ and $\beta$ we run our numerical experiments on.
To motivate the spectra we choose, let us consider the classical analogue of the $\alpha$-versus-$\beta$ spectrum distinguishing game.
Doing so requires defining the following natural classical analogue of a mixed state's spectrum.
\begin{definition}[Classical spectrum]
Given a distribution $p = (p_1, \ldots, p_d)$, we say that \emph{$p$ has spectrum $\gamma = (\gamma_1, \ldots, \gamma_d)$} if $\mathrm{sort}(p) = \gamma$,
where we recall that $\mathrm{sort}(\cdot)$ is the function that sorts its input from highest to lowest.
\end{definition}
\noindent
In the classical analogue of the $\alpha$-versus-$\beta$ spectrum distinguishing game, a distribution $\bq = (\bq_1, \ldots, \bq_d)$ is chosen as follows: with probability $1/2$, it is a uniformly random distribution with spectrum $\alpha$,
and with probability $1/2$, it is a uniformly random distribution with spectrum $\beta$.
(This can be sampled by setting $\bq$ to a uniformly random permutation of $\alpha$ in the first case a uniformly random permutation of $\beta$ in the second case.)
The distinguisher is then given $n$ samples from $\bq$ and asked to guess whether $\bq$ has spectrum $\alpha$ or $\beta$.

We have already seen an example of this distinguishing game in the uniformity testing problem from \Cref{ex:uniformity}.
There, the the two spectra were $\alpha = (\tfrac1d, \ldots, \tfrac1d)$ and $\beta = (\tfrac2d, \ldots, \tfrac2d, 0, \ldots, 0)$, and we saw that $n = \Theta(d^{1/2})$ samples are necessary and sufficient to win the distinguishing game with high probability.
The intuition was that $\alpha$ and $\beta$ differ in their second moment, i.e.\ $p_2(\alpha) = 1/d$ and $p_2(\beta) = 2/d$, where $p_2(\cdot)$ is the power sum symmetric polynomial, and this difference in second moments can be noticed by looking at the pairwise collisions in a sample $\bx = (\bx_1, \ldots, \bx_n)$ drawn from these distributions.
In particular, we expect more pairwise collisions if $\bx$ is drawn from $\beta$ rather than $\alpha$. However, for this distinguisher to work, we have to see \emph{some} pairwise collisions in the sample,
and since both spectra have all probability values at most $O(1/d)$, we should only expect to see a pairwise collision when $n = \Omega(d^{1/2})$.
This is because in either case, the expected number of collisions is
\begin{equation*}
    \E_{\bx}\Big[\sum_{i < j} \mathbbm{1}[\bx_i = \bx_j]\Big]
    = \sum_{i < j} \Pr[\bx_i = \bx_j]
    = \sum_{i < j} \Big(\sum_{k=1}^d \bq_k^2\Big)
    \leq \sum_{i < j} \Big(\sum_{k=1}^d O(1/d^2)\Big)
    = O(n^2/d),
\end{equation*}
which is $o(1)$ if $n = o(d^{1/2})$.
Thus, we only expect to see collisions once $n = \Omega(d^{1/2})$.

This shows an example of a pair of spectra $\alpha$ and $\beta$ where $n = \Omega(d^{1/2})$ samples are necessary to win the distinguishing game.
It also suggests that to design spectra which require more samples to distinguish, we simply need to ensure that their second moments match so that they are not easily distinguished from the pairwise collisions of their samples.
In particular, we want two spectra $\alpha$ and $\beta$ such that $\dtv{\alpha}{\beta} = \Omega(1)$, $p_2(\alpha) = p_2(\beta)$, and all $\alpha_i$'s and $\beta_i$'s are $O(1/d)$.
Then $\alpha$ and $\beta$ might differ noticeably on their \emph{third} moments $p_3(\alpha)$ and $p_3(\beta)$, in which case we would hope to distinguish them by counting the number of 3-wise collisions in the sample $\bx$. (Equivalently, we want to guess which spectrum we're given by estimating $p_3(\bq)$, and as shown in \Cref{sec:learning-sorted}, the (normalized) number of 3-wise collisions in the sample $c_3(\bx)$ is an unbiased estimator for $p_3(\bq)$.) But the above reasoning also implies that because all the $\alpha_i$'s and $\beta_i$'s are $O(1/d)$, then the expected number of 3-wise collisions is $O(n^3/d^2)$, and so we need at least $n = \Omega(d^{2/3})$ samples to distinguish these two distributions. 
Extending this further, if $\alpha$ and $\beta$ agree on their first $k-1$ moments for any constant $k$,
i.e.\ $p_2(\alpha) = p_2(\beta)$, \ldots, $p_{k-1}(\alpha) = p_{k-1}(\beta)$, then we expect that $n = \Omega(d^{1-1/k})$ samples should be required to distinguish them.
Indeed, essentially all of the lower bounds for estimating various symmetric properties of a distribution (which only depend on that distribution's spectrum),
as well as for computing the entire distribution's spectrum,
proceed
by constructing pairs of spectra with matching moments along these lines~\cite{RRSS09,Val08,VV11a,WY16,HJW18}.

These are the pairs of spectra we will consider in the $\alpha$-versus-$\beta$ spectrum distinguishing game.
If $\alpha$ and $\beta$ agree on their first $k-1$ moments, then it is natural to distinguish them using their $k$-th moment, which one can do by estimating $p_{k}(\bgamma) = \tr(\brho^{k})$.
As in the classical case, there is a natural minimum variance unbiased estimator for this quantity.
It was first introduced by O'Donnell and Wright in~\cite{OW15} but it was given a much cleaner interpretation by Badescu, O'Donnell, and Wright in~\cite{BOW19}, who showed that it is a natural quantum analogue of the classical collision statistics.
When all the $\alpha_i$'s and $\beta_i$'s are $O(1/d)$, it can be shown to require $\Omega(d^{2-2/k})$ copies to produce a good enough estimate to distinguish $\alpha$ and $\beta$.
We believe that this should essentially be the best algorithm for distinguishing $\alpha$ and $\beta$,
which would mean that for any constant $k$, we expect a lower bound of $n = \Omega(d^{2-2/k})$ copies.
If this scaling is accurate, it would imply that spectrum estimation cannot be performed in $n = O(d^{2 - \gamma})$ copies for any constant $\gamma > 0$.

For our numerics, we will only focus on the three cases when $\alpha$ and $\beta$ agree on their first $k - 1$ moments, for $k = 2, 3, 4$.
For $k = 2$, we take the distributions
\begin{align}
    \alpha^{(2,d)} &= (1/d, \dots, 1/d) \label{eq:2-moments-matching} \\
    \beta^{(2,d)} &= (2/d, \dots, 2/d, 0, \dots, 0). \nonumber
\end{align}
The two distributions (trivially) have the same first moment but differ in second moments ($1/d$ versus $2/d$), and $\dtv{\alpha^{(2,d)}}{\beta^{(2,d)}} = \frac{1}{2}$.
From our heuristic scaling, these should require $\Theta(d^{2 - 2/2}) = \Omega(d)$ copies to distinguish,
and indeed it is a theorem of Childs, Harrow, and Wocjan~\cite{CHW07} that $n = \Theta(d)$ copies are necessary and sufficient to distinguish these spectra.
We use this pair of spectra to sanity check our numerics and show that they do match this theoretically predicted sample complexity.
Next, for $k = 3$, we use the following two spectra, defined when $d$ is a multiple of $3$. Let 
\begin{align}
    \alpha^{(3,d)} &= \frac{1}{d}\cdot \Big( \underbrace{\frac{3}{2}, \cdots, \frac{3}{2}}_{\frac{2}{3}d}, \underbrace{\vphantom{ \frac{3}{2} }0, \cdots, 0}_{\frac{1}{3}d} \Big), \label{eq:3-moments-matching} \\
    \beta^{(3,d)} &= \frac{1}{d}\cdot \Big( \underbrace{\vphantom{\frac{1}{3}}2, \cdots, 2}_{\frac{1}{3}d}, \underbrace{\frac{1}{2}, \cdots, \frac{1}{2}}_{\frac{2}{3}d} \Big) .  \nonumber
\end{align}
These distributions match on the first and second moments but differ on the third moments (their third moments are $9/(4d^2)$ and $11/(4d^2)$, respectively), and $\dtv{\alpha^{(3,d)}}{\beta^{(3,d)}} = \frac{1}{3}$.
Thus, our heuristic suggests that these should require $\Omega(d^{2 - 2/3}) = \Omega(d^{4/3})$ copies to distinguish.
Finally, for $k = 4$, we construct another family of distribution pairs which match on the first three moments (defined when $d$ is a multiple of $4$): 
\begin{align}
    \alpha^{(4,d)} &= \frac1d \cdot \Big( \underbrace{1 + \frac{1}{\sqrt{2}}, \cdots, 1 + \frac{1}{\sqrt{2}}}_{\frac{1}{2}d}, \underbrace{1 - \frac{1}{\sqrt{2}}, \cdots, 1 - \frac{1}{\sqrt{2}}}_{\frac{1}{2}d} \Big), \label{eq:4-moments-matching} \\
    \beta^{(4,d)} &= \frac1d \cdot \Big( \underbrace{2, \cdots, 2}_{\frac{1}{4}d}, \underbrace{1, \cdots, 1}_{\frac{1}{2}d}, \underbrace{0, \cdots, 0}_{\frac{1}{4}d} \Big). \nonumber
\end{align}
Their fourth moments are $17/(4d^3)$ and $18/(4d^3)$, and they satisfy $\dtv{\alpha^{(4,d)}}{\beta^{(4,d)}} = \frac{1}{4}$.
Our heuristic suggests that these should require $\Omega(d^{2 - 2/4}) = \Omega(d^{3/2})$ copies to distinguish.

\subsection{Results of our simulations}

We now describe the results of our simulations.
Our code can be found at \texttt{\href{https://github.com/ewin-t/spectrum-game}{github.com/ewin-t/spectrum-game}}.

Recall that our goal is to empirically estimate the sample complexity of the $\alpha$-versus-$\beta$ spectrum distinguishing game, as described in \Cref{def:distinguishing-game}.
The distinguishing game reduces to performing spectrum estimation up to error $\frac12 \cdot \dtv{\alpha}{\beta}$, so this sample complexity lower bounds the sample complexity of spectrum estimation.

For a given $\alpha$, $\beta$, and $n$, the optimal success probability for distinguishing can be estimated efficiently, as described in \Cref{def:code-alg}.
We wrote code to perform this estimator, with $m = 10^5$ samples.
We then consider this game for the classes of distributions which match on the first $k-1$ moments, $\alpha^{(k,d)}$ and $\beta^{(k,d)}$, for $k = 2$ (see \Cref{eq:2-moments-matching}), $k = 3$ (see \Cref{eq:3-moments-matching}), and $k = 4$ (see \Cref{eq:4-moments-matching}); then, we iterate over a range of $d$'s, and for each $d$ we find the smallest $n$ for which our estimated optimal success probability exceeds $0.7$.
Our hypothesis predicts that $n$ scales as $\Theta(d^{2 - 2/k})$, so for $k = 2$, $3$, and $4$, this corresponds to scalings of $\Theta(d)$, $\Theta(d^{4/3})$, and $\Theta(d^{3/2})$, respectively.

In the plots that follow, we graph the data, along with lines of best fit among the class of power law functions $a \cdot x^c + b$, and among functions with the predicted scaling---for example, functions of the form $a \cdot x^{4/3} + b$ when $k = 3$.
We use non-linear least squares, provided by \texttt{scipy.optimize.curve\_fit}, to find this line of fit for $n$ as a function of $d$.

\begin{figure}[!htp]
    \centering
    \includegraphics[width=0.95\linewidth]{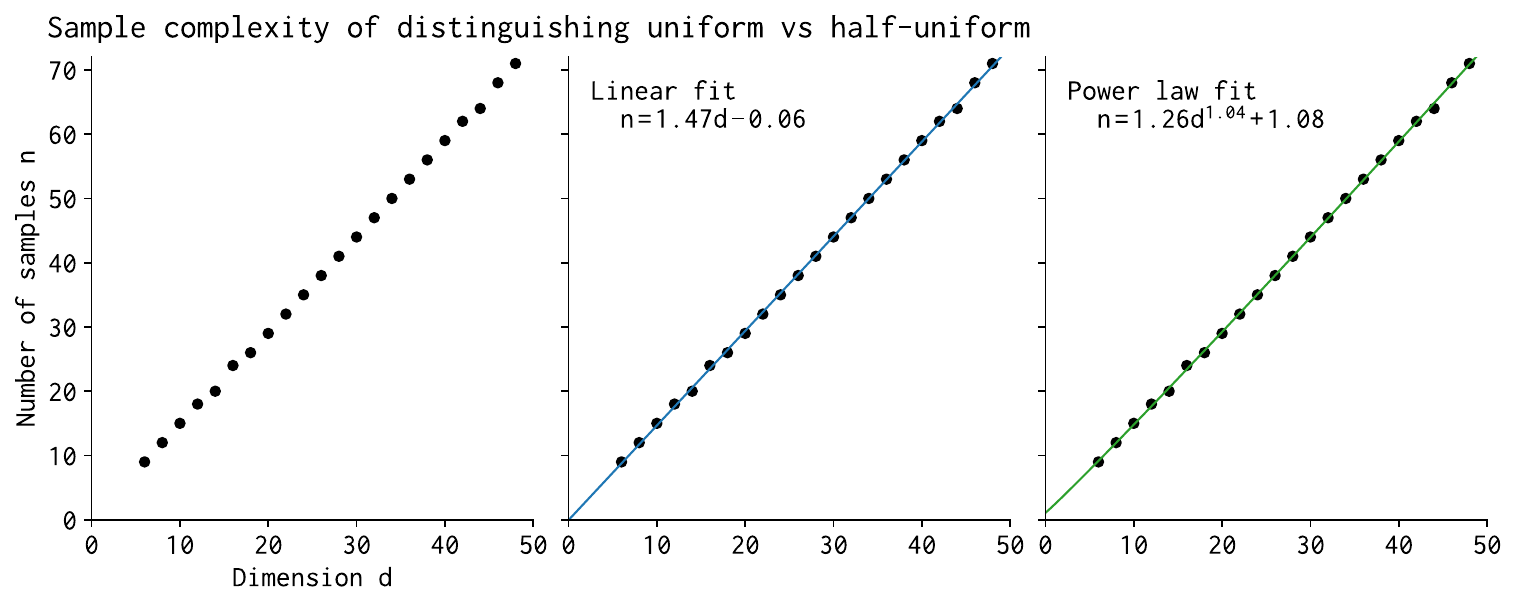}
    \caption{Testing uniformity: This plot displays, for a dimension $d$, the smallest number of samples $n$ necessary to correctly distinguish between $\alpha^{(2,d)}$ and $\beta^{(2,d)}$ with success probability $0.7$. Success probabilities are estimated by taking the empirical probability from $10^5$ trials. $d$ is taken be a multiple of $2$ ranging from $6$ to $48$. The corresponding $n$ values of the data points are 
    $9, 12, 15, 18, 20, 24, 26, 29, 32, 35, 38, 41, 44, 47, 50, 53, 56, 59, 62, 64, 68, 71$. }
    \label{fig:1-moment}
\end{figure}

We plot our results for $k = 2$ in \Cref{fig:1-moment}: the algorithm appears to have a rate $n = \Theta(d)$ for constant $\eps$, matching the heuristic as well as the theoretical upper and lower bounds in~\cite{CHW07}. 

\begin{figure}[!htp]
    \centering
    \includegraphics[width=0.95\linewidth]{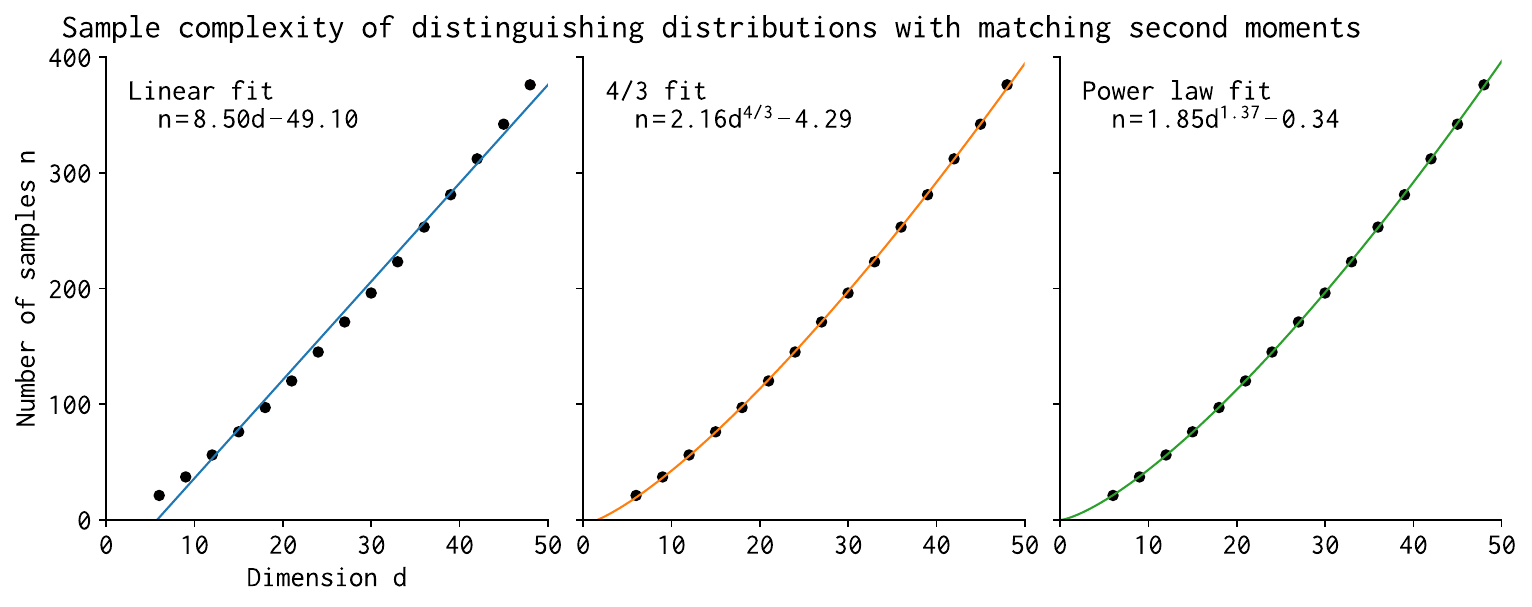}
    \caption{Testing distributions with matching second moments: This plot displays, for a dimension $d$, the smallest number of samples $n$ necessary to correctly distinguish between $\alpha^{(3,d)}$ and $\beta^{(3,d)}$ with success probability $0.7$. Success probabilities are estimated by taking the empirical probability from $10^5$ trials.
    $d$ is taken to be a multiple of $3$ ranging from $6$ to $48$. The corresponding $n$ values of the data points are $21, 37, 56, 76, 97, 120, 145, 171, 196, 223, 253, 281, 312, 342, 376$. }
    \label{fig:2-moments}
\end{figure}

We plot our results for $k = 3$ in \Cref{fig:2-moments}: the best power law fit is a scaling of $n = \Theta(d^{1.37})$.
We also see that the $4/3$ line of fit matches the data better than the linear fit.
This aligns closely with our predicted scaling of $n = \Theta(d^{4/3})$.

\begin{figure}[!htp]
    \centering
    \includegraphics[width=0.95\linewidth]{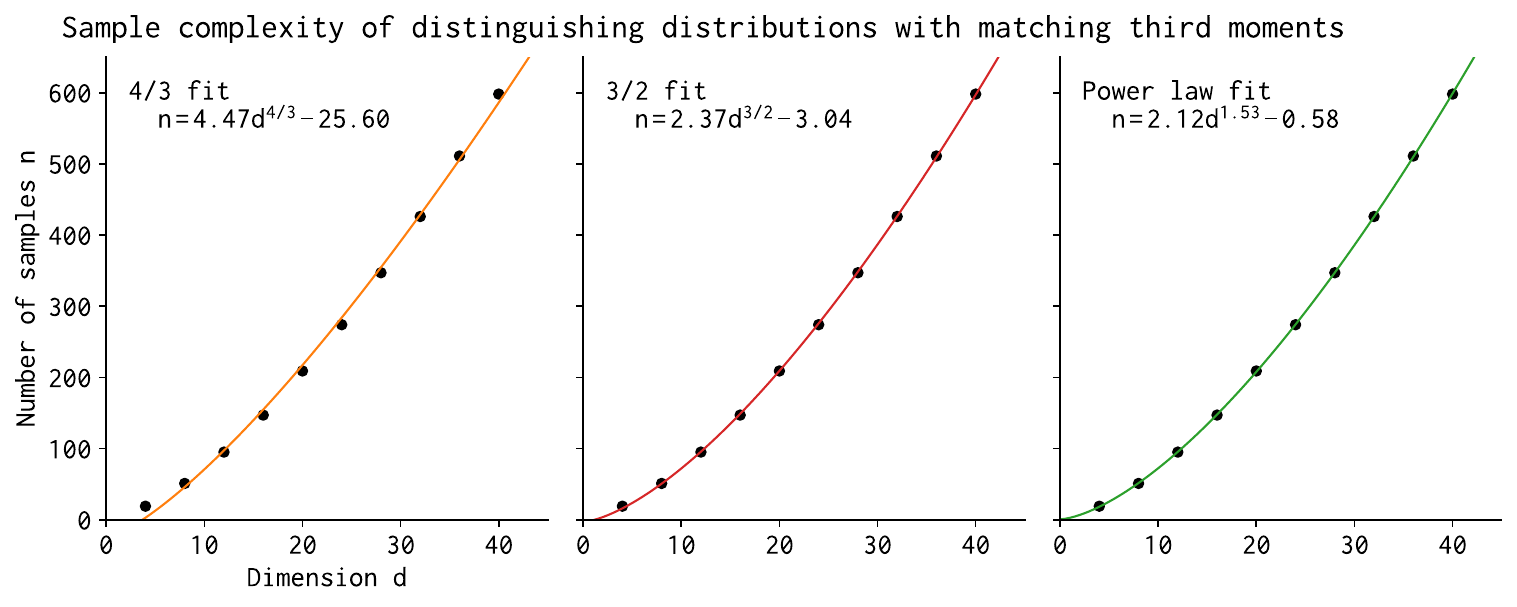}
    \caption{Testing distributions with matching third moments: This plot displays, for a dimension $d$, the smallest number of samples $n$ necessary to correctly distinguish between $\alpha^{(4,d)}$ and $\beta^{(4,d)}$ with success probability $0.7$. Success probabilities are estimated by taking the empirical probability from $10^5$ trials.
    $d$ is taken to be a multiple of $4$ ranging from $4$ to $40$. The corresponding $n$ values of the data points are $19, 51, 95, 147, 209, 274, 347, 426, 511, 598$.}
    \label{fig:3-moments}
\end{figure}

We plot our results for $k = 4$ in \Cref{fig:3-moments}: the best power law fit is a scaling of $n = \Theta(d^{1.53})$.
We also see that the $3/2$ line of fit matches the data better than the $4/3$ fit.
This aligns closely with our predicted scaling of $n = \Theta(d^{3/2})$.

Overall, the empirical scaling matches our hypothesis that the distinguishing task for $k-1$ matching moments has a scaling of $n = \Theta(d^{2 - 2/k})$, supposing that $k$ and the success probability are held constant.
This gives evidence for the hypothesis, which suggests that the scaling for spectrum estimation (with constant error and success probability) is larger than $d^{2 - \gamma}$ for any constant $\gamma > 0$.

\section*{Acknowledgments}

We thank Ryan O'Donnell for several insightful discussions over the course of this project and Yu Tang for providing his personal computing server.

A.P. is supported by DARPA under Agreement No. HR00112020023.
X.T. is supported by the U.S. Department of Energy, Office of Science, National Quantum Information Science Research Centers, Co-design Center for Quantum Advantage (C2QA) under contract number DE-SC0012704. 
This work was done in part while X.T. was visiting the Simons Institute for the Theory of Computing, supported by DOE QSA grant No. FP00010905.
E.T. is supported by the Miller Institute for Basic Research in Science, University of California Berkeley.

\bibliographystyle{alpha}
\bibliography{wright}

\end{document}